\documentclass[11pt]{article}
\usepackage{times}
\usepackage{hyperref}
\hypersetup{colorlinks,citecolor=blue}
\usepackage[dvips, paper=letterpaper, top=1in, bottom=.75in, left=1in, right=1in, nohead, includefoot, footskip=.25in]{geometry}

\usepackage{color}
\usepackage{amsmath,amssymb,amsthm,amsfonts,latexsym,bbm,xspace,graphicx,float,mathtools,dsfont,bm}
\usepackage{mathrsfs}
\usepackage{mathtools}

\usepackage{comment}
\usepackage{tikz}
\usepackage{algorithm}
\usepackage{algorithmic}
\usepackage{verbatim}
\usepackage{cleveref}
\usepackage{enumitem}

\usepackage[export]{adjustbox}
\usepackage{subcaption}

\newtheorem{theorem}{Theorem}
\newtheorem{corollary}{Corollary}
\newtheorem{lemma}{Lemma}

\newtheorem{definition}{Definition}

\newtheorem{example}{Example}

\newcommand{\fb}{\textsc{FB-GFT}}
\newcommand{\opt}{\textsc{SB-GFT}}
\newcommand{\fbsd}{\textsc{FB-GFT}^\textsc{SD}}
\newcommand{\optsd}{\textsc{SB-GFT}^\textsc{SD}}

\newcommand{\opts}{\textsc{OPT-S}}
\newcommand{\optb}{\textsc{OPT-B}}

\newcommand{\cfpp}{constrained FPP\xspace}

\newcommand{\sapp}{\textsc{GFT}_\textsc{SAPP}}
\newcommand{\pp}{\textsc{GFT}_\textsc{FPP}}
\newcommand{\gcfpp}{\textsc{GFT}_\textsc{CFPP}}
\newcommand{\gft}{\textsc{GFT}}

\newcommand{\mingfeinote}[1]{{\color{magenta}{#1}}}

\newcommand{\kgnote}[1]{{\color{orange}{#1}}}
\newcommand{\notshow}[1]{{}}

\newcommand{\ksppabbrev}{\text{KS}}

\newcommand{\propname}{\text{exchangeable}}

\newcommand{\propnoun}{\text{exchangeability}}
\newcommand{\Msapp}{\textsc{SAPP}}

\newcommand{\rppname}{\text{Fixed Posted Price}}
\newcommand{\rppabbrev}{\text{FPP}}
\newcommand{\rpp}{\textsc{GFT}_\textsc{FPP}}
\newcommand{\bo}{\textsc{GFT}_\textsc{BO}}
\newcommand{\so}{\textsc{GFT}_\textsc{SO}}

\DeclareMathOperator{\argmax}{argmax}

\DeclareMathOperator*{\E}{\mathbb{E}}
\DeclareMathOperator{\var}{\mathrm{Var}}

\newenvironment{prevproof}[2]{\noindent {\em {Proof of {#1}~\ref{#2}:}}}{$\Box$\vskip \belowdisplayskip}

\def \MM  {\mathcal{M}}
\def \cF  {\mathcal{F}}
\def \cJ  {\mathcal{J}}
\def \cA  {\mathcal{A}}
\def \DD  {\mathcal{D}}

\def\Bs{\ensuremath{\textbf{s}}}
\def\Bb{\ensuremath{\textbf{b}}}
\def\Bv{\ensuremath{\textbf{v}}}

\def\Bq{\ensuremath{\textbf{q}}}

\newcommand{\ind}{\mathds{1}}

\newcommand*\circled[1]{\tikz[baseline=(char.base)]{
            \node[shape=circle,draw,inner sep=1pt] (char) {#1};}}

\title{On Multi-Dimensional Gains from Trade Maximization}
\author{ Yang Cai\footnote{Supported by a Sloan Foundation Research Fellowship and the NSF Award CCF-1942583 (CAREER) .} \\Yale University, USA\\yang.cai@yale.edu
 \and Kira Goldner\footnote{Supported by NSF Award DMS-1903037 and a Columbia Data Science Institute postdoctoral fellowship.}\\Columbia University, USA\\ kgoldner@cs.columbia.edu
  \and Steven Ma\\Yale University, USA\\ steven.ma@yale.edu
  \and Mingfei Zhao\\ Yale University, USA\\ mingfei.zhao@yale.edu
}
\begin{document}

\maketitle

\begin{abstract}
    We study gains from trade in multi-dimensional two-sided markets.  Specifically, we focus on a setting with $n$ heterogeneous items, where each item is owned by a different seller $i$, and there is a constrained-additive buyer with feasibility constraint $\mathcal{F}$. Multi-dimensional settings in one-sided markets, e.g. where a seller owns multiple heterogeneous items but also is the mechanism designer, are well-understood.  In addition, single-dimensional settings in two-sided markets, e.g. where a buyer and seller each seek or own a single item, are also well-understood.  Multi-dimensional two-sided markets, however, encapsulate the major challenges of both lines of work: optimizing the sale of heterogeneous items, ensuring incentive-compatibility among both sides of the market, and enforcing budget balance.  We present, to the best of our knowledge, the first worst-case approximation guarantee for gains from trade in a multi-dimensional two-sided market.

Our first result provides an $O(\log (1/r))$-approximation to the first-best gains from trade for a broad class of downward-closed feasibility constraints (such as matroid, matching, knapsack, or the intersection of these). Here $r$ is the minimum probability over all items that a buyer's value for the item exceeds the seller's cost. Our second result removes the dependence on $r$ and provides an unconditional $O(\log n)$-approximation to the second-best gains from trade. We extend both results for a general constrained-additive buyer, losing another $O(\log n)$-factor en-route. The first result is achieved using a \emph{fixed posted price mechanism}, and the analysis involves a novel application of the prophet inequality or a new concentration inequality.  Our second result follows from a stitching lemma that allows us to upper bound the second-best gains from trade by the first-best gains from trade from the ``likely to trade" items (items with trade probability at least $1/n$) and the optimal profit from selling the ``unlikely to trade" items. We can obtain an $O(\log n)$-approximation to the first term by invoking our $O(\log (1/r))$-approximation on the ``likely to trade" items. We introduce a generalization of the fixed posted price mechanism---\emph{seller adjusted posted price}---to obtain an $O(\log n)$-approximation to the optimal profit for the ``unlikely to trade" items. Unlike fixed posted price mechanisms, not all seller adjusted posted price mechanisms are incentive compatible and budget balanced. We develop a new argument based on ``allocation coupling'' to show the seller adjusted posted price mechanism used in our approximation is indeed budget balanced and incentive-compatible.

\end{abstract}
\thispagestyle{empty}
\addtocounter{page}{-1}
\newpage
\section{Introduction}\label{sec:intro}

Two-sided markets are ubiquitous in today's economy: take for example the New York Stock Exchange, online ad exchange platforms (e.g., Google’s Doubleclick, Micorosoft’s AdECN, etc.), crowdsourcing platforms, FCC's spectrum auctions, or sharing economy platforms such as Uber, Lyft, and Airbnb.  
Yet mechanism design for such two-sided markets, where both the buyer(s) and seller(s) are strategic, is known to be substantially harder than for one-sided markets, i.e. auctions where the seller designs the mechanism. The additional challenges stem from the following requirements: (1) now the allocation rule must satisfy incentive-compatibility for \emph{both} sides of the market; and (2) the buyer and seller payments must satisfy \emph{budget balance}, that is, the mechanism must not run a deficit.  The limitations of these constraints are best illustrated by the seminal impossibility result of Myerson and Satterthwaite~\cite{MyersonS83}. They show that even in the simplest possible two-sided market---bilateral trade, when one seller is selling a single item to a buyer---no Bayesian incentive compatible (BIC), individually rational (IR), and budget balanced (BB) mechanism can achieve the \emph{first-best efficiency}: the maximum efficiency achievable without any of the previous constraints. The \emph{second-best efficiency} is the maximum efficiency achievable by any BIC, IR, and BB mechanism.

  
  
Despite the additional challenges, significant progress has very recently been made in understanding single-dimensional two-sided markets~\cite{BabaioffCGZ18,BabaioffGG20,BlumrosenD16,BlumrosenM16,BrustleCWZ17,colini2017fixed}. Yet, in reality, many two-sided markets involve agents with multi-dimensional preferences. For example, a customer searching for a place to stay on Airbnb typically values a listing based on its location, number of rooms, amenities, reviews, and more. For one-sided markets, multi-dimensional mechanism design has been the core of algorithmic mechanism design in the past decade, producing a long list of impressive results. See~\cite{chawla2014bayesian,daskalakis2015multi} and the references therein for more details.  Our goal in this paper is to study efficiency maximization in multi-dimensional two-sided markets.  

There are two ways to measure efficiency in two-sided markets. One is the standard notion of welfare. The other is the gains from trade (GFT), which is the welfare of the final allocation minus the total cost of the sellers. Intuitively, the GFT captures how much more welfare the mechanism brings to the market. Clearly, the two measures are equivalent if efficiency is maximized. However, approximating the GFT is much more challenging than welfare. For example, if the buyer’s value is $10$ and the seller’s cost is $9$, not trading the item is a $9/10$-approximation to the welfare but a $0$-approximation to the GFT. Obviously, any good approximation to the GFT immediately gives a good approximation to the welfare, but the opposite direction is rarely true. 

Several results show that generalizations of posted price mechanisms can achieve a constant fraction of the first-best welfare in fairly general multi-dimensional two-sided markets~\cite{BlumrosenD16,Colini-Baldeschi16,colini2020approximately,DuttingRT14}. However, GFT maxmization in multi-dimensional settings has remained elusive. We present, to the best of our knowledge, \textbf{the first worst-case approximation guarantee for GFT in a multi-dimensional two-sided market}. We focus on a setting with $n$ heterogeneous items, where each item is owned by a different seller $i$, and there is a constrained-additive buyer with feasibility constraint $\mathcal{F}$. The Airbnb example is a special case of our setting, where the customer is a unit-demand buyer, and there are $n$ hosts, each listing a property. We further assume that the prior distributions of the buyer's valuations and sellers' costs are public knowledge and independent; the realized valuations and costs are private. 

 Recall that in one-sided markets, maximizing revenue for even a single (non-constrained) additive buyer is far more challenging than for single-dimensional buyers, both optimally and approximately~\cite{BabaioffILW14,DaskalakisDT15,HartN12,ManelliV07,Myerson81}. Maximizing GFT in two-sided markets suffers from this curse of dimensionality as well. As with revenue, single-dimensional settings can leverage an analog to Myerson's virtual value theory by using the optimal dual variables, as shown in~\cite{BrustleCWZ17}, but this does not extend to multiple dimensions. Note also that while Colini-Baldeschi et al. \cite{colini2020approximately} are able to extend an $O(1)$-approximation to welfare to a two-sided market with XOS buyers and additive sellers, their mechanism gives no guarantee for GFT.

\vspace{-.17in}
\paragraph{Our Results:} The first main result is a distribution-parameterized approximation to the \emph{first-best GFT}. 

\medskip \noindent \hspace{0.7cm}\begin{minipage}{0.92\textwidth}
\begin{enumerate}
\item[{\bf Result I:}] There is a \emph{fixed posted price mechanism} whose GFT is an $O(\frac{\log(1/r)}{\delta\eta})$-approximation to the \emph{first-best GFT} when the buyer's feasibility constraint $\cF$ is $(\delta,\eta)$-selectable (Definition~\ref{def-selectablity}), and an $O(\log(n) \cdot \log(1/r))$-approximation for a general constrained-additive buyer.
$r$ is a distributional parameter: the minimum \emph{trade probability} over all items. We define the trade probability of item $i$ as the probability that the buyer's value for $i$ exceeds the seller's cost.
\end{enumerate}
\end{minipage}

\vspace{.1in}
The notion of $(\delta,\eta)$-selectability is introduced by Feldman et al.~\cite{FeldmanSZ16} as a sufficient condition for prophet-inequality-type online algorithms to exist. Many familiar feasibility constraints such as matroid, matching, knapsack, and the compositions of each, are known to be $(\delta,\eta)$-selectable with constant $\delta$ and $\eta$~\cite{FeldmanSZ16}, so our result provides an $O(\log (1/r))$-approximation for all of these environments. See \Cref{def-selectablity} in Section~\ref{subsec:seller-surplus} for the formal definition of $(\delta,\eta)$-selectability.

Next we introduce the class of \emph{fixed posted price} mechanisms.
\vspace{-.15in}
\paragraph{Fixed Posted Price (FPP):} In a fixed posted price mechanism, there is a collection of fixed prices $\left\{(\theta_i^B,\theta_i^S)\right\}_{i\in [n]}$, where $\theta_i^B\geq \theta_i^S$ for each item $i$. Let $R$ be the set of sellers that are willing to sell their item at price $\theta_i^S$. The buyer can purchase any item $i$ in $R$ at price $\theta_i^B$. Trade only occurs when the buyer wants to buy the item and the seller is willing to sell it.

Our result is a generalization of the result by Colini-Baldeschi et al.~\cite{colini2017fixed}, where they provide the same approximation using a fixed posted price mechanism for bilateral trade. Importantly, our approximation ratio has the optimal dependence on $r$ up to a constant factor. Example~\ref{example-log(1/r)-tight} (adapted from an example by Blumrosen and Dobzinski~\cite{BlumrosenD16}) in Appendix~\ref{sec:lower bounds} shows that, for any $r >0$, there is an instance of our problem with minimum trade probability $r$ such that no fixed posted price mechanism can achieve more than a $\frac{c}{\log(1/r)}$-fraction of even the second-best GFT for some absolute constant $c$. In our fixed posted price mechanism, we allow $\theta_i^B$ to be strictly greater than $\theta_i^S$. This is crucial for our analysis, but makes the mechanism only ex-post weakly budget balanced. We leave it as an interesting open question as to whether our approximation ratio can be achieved by an ex-post strongly budget balanced fixed posted price mechanism. 

When the trade probability of each item is not too low, our first result provides a good approximation to the first-best GFT using a simple fixed posted price mechanism. However, $r$ can be arbitrarily small in the worst-case, making our approximation too large to be useful. Is it possible to produce an \emph{unconditional worst-case approximation guarantee}? We provide an affirmative answer to this question with an unconditional $O(\log n)$-approximation to the second-best GFT.

\medskip \noindent \hspace{0.7cm}\begin{minipage}{0.92\textwidth}
\begin{enumerate}
\item[{\bf Result II:}] There is a \emph{dominant strategy incentive compatible} (DSIC), ex-post IR, and BB mechanism whose GFT is at least {$\Omega(\frac{\delta\eta}{\log n})$}-fraction of the \emph{second-best GFT} when the buyer's feasibility constraint $\cF$ is $(\delta,\eta)$-selectable, and at least $\Omega(\frac{1}{\log^2 (n)})$-fraction of the second-best GFT when the buyer is general constrained-additive.
\end{enumerate}
\end{minipage}

\vspace{.1in}
As we show in Example~\ref{example-log(1/r)-tight}, no fixed posted price mechanism can provide such a guarantee. We develop two new mechanisms. The first one is a multi-dimensional extension of the ``Generalized Buyer Offering Mechanism'' by Brustle et al.~\cite{BrustleCWZ17}. We provide a full description of the mechanism in Section~\ref{sec:bounding super buyer}. The second mechanism is a generalization of the fixed posted price mechanism that we call the \emph{Seller Adjusted Posted Price Mechanism}.
\vspace{-.2in}
\paragraph{Seller Adjusted Posted Price (\Msapp):} The sellers report their costs $\Bs$. The mechanism maps the cost profile to a collection of posted prices  $\left\{\theta_i(\Bs)\right\}_{i\in [n]}$ for the buyer. The buyer can purchase at most one item, and pays price $\theta_i(\Bs)$ if she buys item $i$. An item trades  if the buyer decides to purchase that~item.

The main advantage of using a \Msapp~mechanism is that it provides the flexibility to set prices based on the sellers' costs, which allows a \Msapp~mechanism to achieve GFT that could be unboundedly higher than the GFT attainable by even the best fixed posted price mechanism (see Example~\ref{example-sapp-beat-fp}). {Example~\ref{example:sapp-necessary} in Appendix~\ref{sec:lower bounds} shows that the class of SAPP mechanisms is necessary to obtain any finite approximation ratio to the second-best: both the best FPP mechanisms and the ``Generalized Buyer Offering Mechanism''~\cite{BrustleCWZ17} have an unbounded gap compared to the second-best GFT, even in the bilateral trade setting.}

An astute reader may have already realized that the payments to the sellers are not yet defined in the \Msapp~mechanism. This is because the allocation rule of a \Msapp~mechanism is not necessarily monotone in the sellers' costs if the mappings $\{\theta_i(\cdot)\}_{i\in[n]}$ are not chosen carefully. Interestingly, we show that if the mappings $\{\theta_i(\cdot)\}_{i\in[n]}$ satisfy a strong type of monotonicity that we call \emph{bi-monotonicity} (Definition~\ref{def:bi-monotonicity}), then the allocation rule is indeed monotone in each seller's reported cost. Since the sellers are single-dimensional, we can apply Myerson's payment identity to design an incentive compatible payment rule. The final property we need to establish is budget balance, which turns out to be the major technical challenge for us. We provide more details and intuition about our solution to this challenge in the discussion of the techniques.

In Section~\ref{sec:reduction}, we draw a connection between a lower bound to our analysis and one of the major open problems in single dimensional two-sided markets.
 We prove a reduction from approximating the first-best GFT in the \emph{unit-demand} setting to bounding the gap between the first-best and second-best GFT in a related \emph{single-dimensional} setting (Theorem~\ref{thm:reduction}). If in the latter market, the gap between first-best and second-best GFT is at most $c$, then our mechanism is a $2c$-approximation to the first-best GFT in the former market.

\subsection{Our Approach and Techniques}\label{sec:techniques}

\paragraph{$\log (1/r)$-Approximation (Section~\ref{sec:log1/r}):} Our starting point is similar to Colini-Baldeschi et al.~\cite{colini2017fixed}. We first argue that the probability space of each item $i$ can be partitioned into $O(\log (1/r))$ events $\{E_{ij}\}_{j\in[\log (2/r)]}$, such that in each event $E_{ij}$, the median of the buyer's value $b_i$ for item $i$ dominates the median of the $i$-th seller's cost $s_i$. The first-best GFT is upper bounded by the sum of the contribution to GFT from each of these events. In bilateral trade, simply setting the posted price to be the median of the buyer's value is sufficient to obtain $1/2$ of the optimal GFT from $E_{ij}$ as shown by McAfee~\cite{McAfee08}. The $\log(1/r)$-approximation by Colini-Baldeschi et al.~\cite{colini2017fixed} essentially follows from this argument. 

To illustrate the added difficulty from multiple items, it suffices to consider a unit-demand buyer. Setting the posted price on each item to be the median of the buyer's value does not provide a good approximation, because the buyer will purchase the item that gives her the highest surplus, which could be very different from the item that generates the most GFT. Similar scenarios are not uncommon in \emph{multi-dimensional auction design}, and prophet inequalities~\cite{KleinbergW12,krengel1978semiamarts} have been proven to be effective in addressing similar challenges. The main barrier for applying the prophet inequality to two-sided markets is choosing the appropriate random variable as the reward for the prophet/gambler. It is not obvious how to choose a random variable that will translate to a two-sided market mechanism, and in fact, for some choices, no translation between the thresholding policy for the gambler and a two-sided market mechanism is possible.\footnote{For example, one can  choose the GFT from the $i^{\mathrm{th}}$ item  $(b_i-s_i)^+$ as the reward of the $i^{\mathrm{th}}$ round, but no fixed posted price mechanism corresponds to the policy that only accepts items whose GFT is above a certain threshold. Indeed, no BIC, IR, and BB mechanism can implement a thresholding policy with threshold $0$ due to the impossibility result by Myerson and Satterthwaite~\cite{MyersonS83}.} 
Our key insight is to replace event $E_{ij}$ with a related but different event $\overline{E}_{ij}$ where there is a fixed number $\theta_{ij}$ such that $s_i$ and $b_i$ are always separated by $\theta_{ij}$ ($s_i\leq \theta_{ij}\leq b_i$). We further  show that the GFT contribution from event $\overline{E}_{ij}$ is at least half of the GFT contribution from $E_{ij}$. Importantly, the GFT contributed by item $i$ in event $\overline{E}_{ij}$: $(b_i-s_i)^+~\footnote{$x^+=\max\{x,0\}$.} =(b_i-\theta_{ij})^+\cdot \ind[s_i\leq \theta_{ij}]+(\theta_{ij}-s_i)^+\cdot \ind[b_i\geq \theta_{ij}]$. Note that if we replace $\overline{E}_{ij}$ with $E_{ij}$, the LHS can exceed the RHS when $\theta_{ij}>b_i>s_i$. The decomposition of $(b_i - s_i)^+$ using $\theta_{ij}$ is critical for us to apply the prophet inequality. We can now choose the reward for the gambler to be $v_i=(\theta_{ij}-s_i)^+\cdot \ind[b_i\geq \theta_{ij}]$, and the thresholding policy with a threshold $T$ can be implemented with a posted price mechanism where the price for the buyer is $\theta_{ij}$ and the price for the seller is $\theta_{ij}-T$.\footnote{A similar fixed posted price mechanism can take care of $(b_i-\theta_{ij})^+\cdot \ind[s_i\leq \theta_{ij}]$.}

When the buyer's feasibility constraint is general downward-closed, the only known prophet inequalities are due to Rubinstein~\cite{Rubinstein16} and are $O(\log n)$-competitive. Unfortunately, the prophet inequalities in~\cite{Rubinstein16} are highly adaptive, and thus cannot translate into prices for a single buyer. Further, an almost matching lower bound of $O(\log n/\log\log n)$ is shown by Babaioff et al.~\cite{babaioff2007matroids}, precluding much possible improvement for this approach. Instead, we use a constrained fixed posted price mechanism that forces the buyer to buy at least $h$ items (at their posted prices) if she wants to buy any; otherwise, she must leave with nothing.  We divide the same variables $v_i$ into $O(\log n)$ buckets based on their contribution to seller surplus. Within each bucket $k$, all variables $v_i$ lie in $[L_k,2L_k]$ for some $L_k$. We prove a concentration inequality for the \emph{maximum size of a feasible and affordable set}.
It guarantees that with constant probability, the buyer will be willing to purchase at least $h$ items (for an appropriate choice of $h$), generating sufficient GFT.

\vspace{-.15 in}
\paragraph{Benchmark of the Second-Best GFT (Section~\ref{sec:UB of SB}):} As our goal is to obtain a benchmark of the second-best GFT that is unconditional, the benchmark from the previous (distribution-parameterized) result cannot be used here. We derive a novel benchmark in two steps. Step (i): we create two imaginary one-sided markets: the \emph{super seller auction} and the \emph{super buyer procurement auction}. 
We show that the second-best GFT of the two-sided market is upper bounded by the \emph{optimal profit} from the super seller auction and the \emph{optimal buyer utility} from the super buyer procurement auction. Step (ii): we provide an extension of the marginal mechanism lemma~\cite{HartN12,CaiH13} to the optimal profit. We show that the optimal profit for selling all items in $[n]$ is upper bounded by the \emph{first-best} GFT from items in $T$ and the optimal profit for selling items in $[n]\backslash T$, where $T$ is an arbitrary subset of $[n]$. Our key insight is to choose $T$ to be the ``likely to trade'' items, which are the ones with trade probability at least $1/n$, and apply the marginal mechanism lemma. This partition allows us to use our first result to provide an $O(\log n)$-approximation to the first-best GFT of the ``likely to trade'' items using a fixed posted price mechanism. Moreover, we prove that the optimal buyer utility from the super buyer procurement auction is upper bounded by the GFT of an extension of the ``generalized buyer offering mechanism'' \cite{BrustleCWZ17}. Finally, we provide an $O(\log n)$-approximation to the optimal profit for selling the ``unlikely to trade'' items using a \Msapp~mechanism. Note that the approximation crucially relies on the fact that in expectation at most one item can trade among the ``unlikely to trade'' items. 

\vspace{-.15in}
\paragraph{Budget Balance of Seller Adjusted Posted Price Mechanisms (Section~\ref{sec:SAPP}):} As mentioned earlier, we restrict our attention to  bi-monotonic mappings from cost profiles to buyer prices $\{\theta_i(\cdot)\}_{i\in[n]}$ to guarantee incentive-compatibility. However, budget balance does not follow from bi-monotonic mappings. We extend the definition of bi-monotonicity to allocation rules and show that all bi-monotonic allocation rules can be transformed into a DSIC, IR, and BB \Msapp~mechanism. In our proof of the budget balance property, we identify an auxiliary allocation rule $q$, which may not be implementable by a BB mechanism. We then show that the allocation rule of our \Msapp~mechanism is ``coupled'' with $q$. In particular, our allocation probability is always between $q/4$ and $q/2$. The upper bound $q/2$ allows us to upper bound the payment to the seller, and the lower bound $q/4$ allows us to lower bound the payment we collect from the buyer. Surprisingly, we can prove that the upper bound of the payment to the seller is no more than the lower bound of the buyer's payment. We suspect this type of allocation coupling argument may also be useful in other problems.
\subsection{Related Work}\label{subsec:relwork}

\paragraph{Gains from Trade.} The main related works are on worst-case GFT approximation.  Blumrosen and Mizrahi \cite{BlumrosenM16} guarantee an $e$-approximation to the first-best GFT in the setting of bilateral trade---one buyer, one seller, one item---when the buyer's distribution satisfies the monotone hazard rate condition.  Brustle et al. \cite{BrustleCWZ17} study the more general double auction setting: there are many buyers and sellers, but the goods are identical, and each buyer and seller is unit-demand or unit-supply respectively. In addition, they allow any downward-closed feasibility constraint over the buyer-seller pairs that can trade simultaneously. They use the better of a ``seller-offering'' or ``buyer-offering'' mechanism to achieve a $2$-approximation to the second-best GFT, for general buyers' and sellers' distributions.  Colini-Baldeschi et al.~\cite{colini2017fixed} show that a simple fixed price mechanism obtains an $O(\frac{1}{r})$-approximation to GFT in the bilateral trade and double auction settings, but a more careful setting of the fixed price gives an $O(\log \frac{1}{r})$-approximation for bilateral trade.  Our setting is the first multi-dimensional setting with a worst-case approximation guarantee, and we match the $O(\log \frac{1}{r})$-approximation of \cite{colini2017fixed} while providing an unconditional $O(\log n)$-approximation.   

Other lines of work provide (1) \emph{asymptotic} approximation guarantees in the number of items optimally traded for settings as general as multi-unit buyers and sellers and $k$ types of items \cite{McAfee92,SegalHaleviHA18a,SegalHaleviHA18b}, (2) dual asymptotic and worst-case guarantees for double auctions and matching markets \cite{BabaioffCGZ18}, and (3) Bulow-Klemperer-style guarantees of the number of additional buyers (or sellers) needed in double auctions in order for the GFT of the new setting running a simple mechanism to beat the first-best GFT of the original setting \cite{BabaioffGG20}.  

\vspace{-.15in}
\paragraph{Multi-Dimensional Revenue.} In the setting where one seller owns all of the items, has no cost for the items, and is the mechanism designer, much more is known. However, even when selling to a single additive bidder (e.g. with no feasibility constraints), posted prices can achieve at best an $O(\log n)$-approximation \cite{HartN12,LiY13}.  In order to obtain a constant-factor approximation for an additive buyer, Babaioff et al. \cite{BabaioffILW14} use the better of posted prices and posting a price on the grand bundle, and a variation works for a single subadditive (which includes constrained-additive) buyer as well \cite{RubinsteinW15}.  However, in a two-sided market where items are owned by separate sellers, it is not clear how to implement bundling in an incentive-compatible way.  The mechanisms used to obtain constant-approximations for multiple constrained-additive, XOS, or subadditive buyers \cite{ChawlaM16,CaiZ17} are only more complex.

\vspace{-.15in}
\paragraph{Welfare in Two-Sided Markets.} 
Colini-Baldeschi et al.~\cite{Colini-Baldeschi16} consider welfare maximization in the double auction setting with matroid feasibility constraints.  They generalize sequential posted price mechanisms (SPMs) to the two-sided market setting, guaranteeing a constant-factor approximation to welfare.  The mechanism posts prices for each buyer-seller combination (not just for each item), visits the buyers and sellers simultaneously in the given order, and advances on either side when the price is rejected.  Trade occurs when both sides accept the trade.  Follow up work of Colini-Baldeschi et al.~\cite{colini2020approximately} generalizes the idea to the setting where buyers are XOS and sellers are additive.  Here, there is a posted price for each item, but only ``high welfare'' items are considered.  The buyers visit and pick out the bundles they want among the high welfare items.  Then, sellers are given the opportunity to sell their entire bundle of items demanded by the buyers (but not any subset), and they are skipped with some probability.  Like the previous work, this mechanism is ex-post IR, DSIC, and strongly BB (buyer payments equal seller payments). As only ``high welfare'' items are considered, it is possible for their mechanism to not trade any item when the minimum trade probability $r$ is a constant.

Blumrosen and Dobzinski~\cite{BlumrosenD16} give an IR, BIC and strongly BB mechanism for bilateral trade that obtains in expectation a constant-fraction of the optimal welfare.  D\"{u}tting et al.~\cite{DuttingRT14} study welfare maximization in the prior-free setting and present DSIC, IR, and weakly BB (buyer payments exceed seller payments) mechanisms for double auctions with feasibility constraints on either side.

\vspace{-.15in}

\section{Preliminaries}\label{sec:prelim}

\vspace{-.05in}
\paragraph{Two-sided Markets.} We focus on two-sided markets between a single buyer and $n$ unit-supply sellers. Every seller $i$ sells a heterogeneous item. For simplicity we denote the item sold by seller $i$ as item $i$. Each seller $i$ has cost $s_i$ for producing item $i$, where $s_i$ is drawn independently from distribution $\DD_i^S$. The buyer has value $b_i$ for every item $i$ where $b_i$ is drawn independently from distribution $\DD_i^B$.  $\DD_i^S$ and $\DD_i^B$ are public knowledge. Let $\DD^B=\times_{i=1}^n \DD_i^B$ be the distribution of the buyer's value profile and $\DD^S=\times_{i=1}^n \DD_i^S$ be the distribution of the cost profile for all sellers. Let $\Bb=(b_1,...,b_n)$ and $\Bs=(s_1,...,s_n)$ denote the value (or cost) profile for the buyer and all sellers. For notational convenience, for every $i$ we denote $b_{-i}$ (or $s_{-i}$) the value (or cost) profile without item $i$. For every $i$, $F_i,f_i$ (or $G_i,g_i$) denote the cumulative distribution function and density function of $\DD_i^B$ (or $\DD_i^S$). Throughout the paper we assume that all distributions are continuous over their support, and thus the inverse cumulative function $F_i^{-1}$ and $G_i^{-1}$ exist.\footnote{Any discrete distribution can be made continuous by replacing each point mass $a$ with a uniform distribution on $[a-\epsilon,a+\epsilon]$, for arbitrarily small $\epsilon$. Thus our result applies to discrete distributions as well by losing arbitrarily small GFT.}

Throughout this paper, we assume that the buyer has a constrained-additive valuation over the items, which means that the buyer is additive over the items, but is only allowed to take a feasible set of items with respect to a downward-closed\footnote{$\cF\subseteq 2^{[n]}$ is downward-closed if for every $S\in \cF$, we have $S'\in \cF, \forall S'\subseteq S$.} constraint $\cF\subseteq 2^{[n]}$.
Formally, for every $\Bb$ and $S\subseteq [n]$, the buyer's value for a set of items $S$ is: $v(\Bb,S)=\max_{T\in \cF,T\subseteq S}\sum_{i\in T}b_i$. 

\vspace{-.15in}
\paragraph{Mechanism and Constraints.} Any mechanism in the two-sided market defined above is specified by the tuple $(x,p^B,p^S)$ where $x$ is the allocation rule of the mechanism and $p^B,p^S$ are the payment rules. For every profile $(\Bb,\Bs)$ and every $i$, $x_i(\Bb,\Bs)$ is the probability that the buyer trades with seller $i$ under profile $(\Bb,\Bs)$. $p^B(\Bb,\Bs)$ is the payment from the buyer and $p_i^S(\Bb,\Bs)$ is the gains for (or payment to) seller $i$. All agents in the market have linear utility functions.\footnote{Without loss of generality we can assume that the mechanism will only allow the buyer to trade with a (possibly randomized) set $S$ of sellers where $S\in \cF$. For any trading set $T$, let $S^*$ denote the utility-maximizing feasible subset, $S^* = \argmax_{S\in \cF,S\subseteq T}\sum_{i\in S}b_i$. If we only allow the buyer to trade with the sellers in $S^*$ instead of all sellers in $T$, the gains from trade of the mechanism will not decrease.\label{footnote:tie-breaking}}
We call the mechanism ex-ante Strongly Budget Balanced (SBB) or Weakly Budget Balanced (WBB) if the buyer's expected payment equals, or is greater than, the sum of all sellers' expected gains, respectively, over the randomness of the mechanism and the profiles of all agents. We call the mechanism ex-post SBB (or ex-post WBB) if this property holds for every agent's profile. The definition of incentive compatibility and individual rationality are as follows. 

\begin{itemize}
    \item BIC: For every agent, reporting her true value (or cost) maximizes her expected utility over the profiles of other agents.
    \item DSIC: For every agent, reporting her true value (or cost) maximizes her expected utility, no matter what other agents report.
    \item (Bayesian) IR: For every agent, reporting her true value (or cost) derives non-negative utility over the profiles of other agents.
    \item Ex-post IR: For every agent, reporting her true value (or cost) derives non-negative utility, no matter what other agents report.
\end{itemize}


\notshow{

\begin{itemize}
    \item \textbf{Ex-post Strong Budget Balance (ex-post SBB)}: For every profile $(\Bb,\Bs)$, the buyer's expected payment is equal to the sum of all sellers' expected gains, over the randomness of mechanism.
    \item \textbf{Ex-post Weak Budget Balance (ex-post WBB)}: For every profile $(\Bb,\Bs)$, the buyer's expected payment is at least the sum of all sellers' expected gains, over the randomness of mechanism.
    \item \textbf{Ex-ante Strong Budget Balance (ex-ante SBB)}: The buyer's expected payment equals to the sum of all sellers' expected gains, over the randomness of mechanism and the profile of all agents.  
    \item \textbf{Ex-ante Weak Budget Balance (ex-ante WBB)}: The buyer's expected payment is at least the sum of all sellers' expected gains, over the randomness of mechanism and the profile of all agents. 
\end{itemize}
}

\vspace{-.15in}

\paragraph{Gains from Trade.}We aim to maximize the Gains from Trade (GFT), i.e. the gains of social welfare induced by the mechanism. Formally, given a mechanism $M=(x, p^B, p^S)$, the expected GFT of $M$ is 
$$\textstyle\gft(M)=\E_{\Bb\sim \DD^B,\Bs\sim \DD^S} \left[\sum_{i=1}^nx_i(\Bb,\Bs)\cdot(b_i-s_i)\right].$$

We use $\opt$ to denote the optimal GFT attainable by any BIC, IR, ex-ante WBB mechanism (also known as the ``second-best'' mechanism). On the other hand, let $\fb$ denote the maximum possible gains of social welfare among all feasible allocations (known as the ``first-best''). 
Formally
$$\textstyle\fb=\E_{\Bb\sim \DD^B,\Bs\sim \DD^S} \left[\max_{S\in \cF}\sum_{i\in S}(b_i-s_i)\right].$$



In Section~\ref{sec:log1/r}, the distribution-parameterized approximation uses the parameter $r$, the minimum probability over all items $i$ that the buyer's value for item $i$ is at least seller $i$'s cost. Formally, for every item $i\in [n]$, let $r_i=\Pr_{b_i\sim \DD_i^B, s_i\sim \DD_i^S}[b_i\geq s_i]$ denote the probability that the buyer's value for item $i$ exceeds seller $i$'s cost. Without loss of generality, assume that $r_i>0$ for all $i\in [n]$.\footnote{If $r_i=0$ the mechanism should never trade between the buyer and seller $i$, and so it can remove seller $i$ from the market. This will not decrease the GFT of the mechanism as $b_i<s_i$ with probability 1.} Let $r=\min_{i\in [n]}r_i>0$. 


\section{A Distribution-Parameterized Approximation}\label{sec:log1/r}


In this section, we present an $O(\frac{\log(1/r)}{\delta\eta})$-approximation to $\fb$ when the buyer's feasibility constraint $\cF$ is $(\delta,\eta)$-selectable, and an $O(\log(n)\cdot\log(\frac{1}{r}))$-approximation  for a general constrained-additive buyer. 
In Section~\ref{subsec:breaking terms}, we show that $\fb$ can be bounded by the sum of four separate terms. 
In Section~\ref{sec:buyer surplus} we show that two of the terms (``buyer surplus'') are relatively easy to bound using fixed posted price (FPP) mechanisms with the same prices posted on both sides. 
In Section~\ref{subsec:seller-surplus-UD}, we consider the special case of a unit-demand buyer and bound the other two terms (``seller surplus'') using FPP mechanisms combined with the prophet inequality. In Section~\ref{subsec:seller-surplus}, we introduce the concept of selectability~\cite{FeldmanSZ16} and bound the seller surplus for any selectable feasibility constraint by using a 
\emph{constrained}
FPP mechanism. In Section~\ref{subsec:proof-downward-close}, we present our result for a general constrained-additive buyer. 

\subsection{Upper Bound of $\fb$}\label{subsec:breaking terms}

For every $i$, let $\overline{F_i} = 1 - F_i$ denote the complementary CDF of $b_i$. 
Let $x_i$ and $y_i$ be the $\frac{r_i}{2}$-quantile of the buyer's and seller's distribution for item $i$, respectively. Formally, $x_i=\overline{F_i}^{-1}(\frac{r_i}{2}), y_i=G_i^{-1}(\frac{r_i}{2})$. We first prove that $x_i\geq y_i$.

\begin{lemma}\label{lem:x_i-vs-y_i}
For every $i\in [n]$, $x_i\geq y_i$.
\end{lemma}

\begin{proof}
Note that for every $i\in [n]$, $b_i<x_i\wedge s_i>x_i$ implies that $b_i<s_i$.  We have
\begin{align*}
    1-r_i&=\Pr_{b_i\sim \DD_i^B,s_i\sim \DD_i^S}[b_i<s_i]
\geq \Pr_{b_i,s_i}[b_i<x_i\wedge s_i>x_i]=(1-\frac{r_i}{2})\cdot (1-\Pr_{s_i}[s_i\leq x_i]). 
\end{align*}

Suppose $x_i<y_i$. Then 
$(1-\frac{r_i}{2})\cdot (1-\Pr_{s_i}[s_i\leq x_i])\geq (1-\frac{r_i}{2})^2>1-r_i.$
This is a contradiction. Thus $x_i\geq y_i$.
\end{proof}

In the following upper bound, we will separate the probability space for each item $i$ into $2\lceil\log(2/r)\rceil$ events, 
and then divide the GFT into buyer surplus and seller surplus terms according to the cutoff for each event.
For every $\Bb,\Bs$, define {the feasible set that maximizes the GFT} as $S^*(\Bb,\Bs)=\argmax_{S\in \cF}\sum_{k\in S}(b_k-s_k)$, and break ties arbitrarily. Observe the following upper bound for the first-best GFT:
\begin{align*}
        \fb=\E_{\Bb,\Bs}[\max_{S\in \cF}\sum_{i\in S}(b_i-s_i)^+]
        &\leq \textstyle\E_{\Bb,\Bs}\left[\sum_i(b_i-s_i)\cdot\ind[i\in S^*(\Bb,\Bs)]\cdot \ind[b_i\geq s_i\wedge s_i< x_i]\right]~~~~~\circled{1}\\
        &+\textstyle\E_{\Bb,\Bs}\left[\sum_i(b_i-s_i)\cdot\ind[i\in S^*(\Bb,\Bs)]\cdot \ind[b_i\geq s_i\geq y_i]\right]~~~~~\circled{2},
\end{align*}
where the inequality holds because $x_i\geq y_i$ for all $i$. We first consider term \circled{1}. For every $i\in [n], j\in 1,2,\ldots,\lceil \log(\frac{2}{r})\rceil$, let $\theta_{ij}=\overline{F_i}^{-1}(\frac{1}{2^{j}})$. 
Let $E_{ij}$ be the event that 
$\overline{F_i}^{-1}(\frac{1}{2^{j-1}})\leq s_i\leq \overline{F_i}^{-1}(\frac{1}{2^{j}})\wedge b_i\geq \overline{F_i}^{-1}(\frac{1}{2^{j-1}})$. Then we have $~~~\textstyle\circled{1}\leq \textstyle\sum\limits_{j=1}^{\lceil \log(\frac{2}{r})\rceil}\E_{\Bb,\Bs}\left[\sum_i(b_i-s_i)^+\cdot\ind[i\in S^*(\Bb,\Bs)\wedge E_{ij}] \right]$.

As discussed in Section~\ref{sec:techniques}, in order to bound the benchmark with fixed posted price mechanisms, we will consider a more restrictive event $\overline{E}_{ij}$ and show that the GFT contribution from event $\overline{E}_{ij}$ is at least half of the GFT contribution from $E_{ij}$. Both events are depicted in Figure~\ref{fig:events}. 

\begin{figure}[h!] 
	\centering
	\includegraphics[scale=.5]{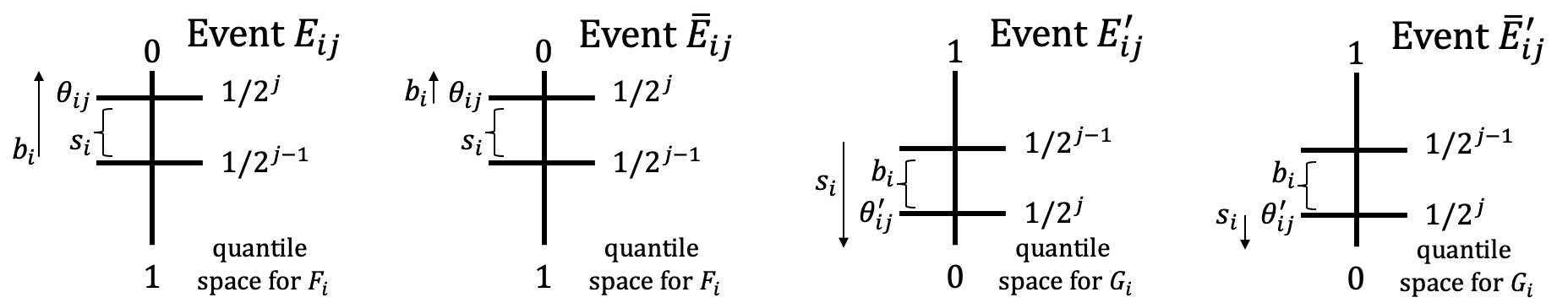} 
	\caption{Events $E_{ij}$ and $E'_{ij}$ used in the upper bound of GFT, and the restricted events $\overline{E}_{ij}$ and $\overline{E}'_{ij}$.}
	\label{fig:events} 
\end{figure}

\begin{lemma}\label{lem:breaking-terms-1}
For every $i,j$, let $\overline{E}_{ij}$ be the event that $\overline{F_i}^{-1}(\frac{1}{2^{j-1}})\leq s_i\leq \overline{F_i}^{-1}(\frac{1}{2^{j}})\wedge b_i\geq \overline{F_i}^{-1}(\frac{1}{2^{j}})$. Then the following inequality holds for every $j=1,\ldots,\lceil\log(2/r)\rceil$:
$$\textstyle \E_{\Bb,\Bs}\left[\sum_i(b_i-s_i)^+\cdot\ind[i\in S^*(\Bb,\Bs)\wedge E_{ij}] \right]\leq 2\cdot \E_{\Bb,\Bs}\left[\sum_i(b_i-s_i)^+\cdot\ind[i\in S^*(\Bb,\Bs)\wedge \overline{E}_{ij}] \right].
$$
\begin{align*}
\text{Moreover,} \quad \quad \circled{1} &\leq 2\cdot\sum\limits_{j=1}^{\lceil \log(2/r)\rceil}\E_{\Bb,\Bs}\left[\max_{S\in \cF}\sum_{i\in S}\left\{(b_i-\theta_{ij})^+\cdot\ind[s_i\leq \theta_{ij}] \right\}\right] \quad \circled{3} \\
&+ 2\cdot\sum\limits_{j=1}^{\lceil \log(2/r)\rceil}\E_{\Bb,\Bs}\left[\max_{S\in \cF}\sum_{i\in S}\left\{(\theta_{ij}-s_i)^+\cdot\ind[b_i\geq \theta_{ij}]\right\} \right] \quad \circled{4}.
\end{align*}
\notshow{
$$
\circled{1} \leq
\underbrace{2\cdot\sum\limits_{j=1}^{\lceil \log(2/r)\rceil}\E_{\Bb,\Bs}\left[\max_{S\in \cF}\sum_{i\in S}\left\{(b_i-\theta_{ij})^+\cdot\ind[s_i\leq \theta_{ij}] \right\}\right]}_{\circled{3}} 
+ \underbrace{2\cdot\sum\limits_{j=1}^{\lceil \log(2/r)\rceil}\E_{\Bb,\Bs}\left[\max_{S\in \cF}\sum_{i\in S}\left\{(\theta_{ij}-s_i)^+\cdot\ind[b_i\geq \theta_{ij}]\right\} \right]}_{\circled{4}}.
$$}
\end{lemma}

Readers may notice that $\Pr[\overline{E}_{ij}]=\frac{1}{2}\cdot\Pr[E_{ij}]$. However, this alone does not prove the first statement of Lemma~\ref{lem:breaking-terms-1}, since both the indicator $\ind[E_{ij}]$ and the contributed GFT $(b_i-s_i)^+\cdot\ind[i\in S^*(\Bb,\Bs)]$ depend on the realization of $b_i,s_i$. In Lemma~\ref{lem:breaking-terms-1} we show that the two random variables are \emph{positively correlated} with respect to $b_i$, which allows us to prove the first statement. The second statement follows from the fact that $(b_i-s_i)^+\leq (b_i-\theta_{ij})^+ +(\theta_{ij}-s_i)^+$ for every $b_i,s_i$, and that $S^*(\Bb,\Bs)\in \cF$ for every $\Bb,\Bs$.

In Lemma~\ref{lem:breaking-terms-2}, we bound term \circled{2} in a similar way. The proof of Lemmas~\ref{lem:breaking-terms-1} and~\ref{lem:breaking-terms-2} can be found in Appendix~\ref{sec:appx_log1/r}.

\begin{lemma}\label{lem:breaking-terms-2}
For every $i\in [n]$ and $j=1,\ldots,\lceil \log(2/r)\rceil$, let $\theta_{ij}'= G_i^{-1}(\frac{1}{2^{j}})$. 
Then 
\begin{align*}
 \quad\quad \circled{2} &\leq 2\cdot\sum\limits_{j=1}^{\lceil \log(2/r)\rceil}\E_{\Bb,\Bs}\left[\max_{S\in \cF}\sum_{i\in S}\left\{(b_i-\theta_{ij}')^+\cdot\ind[s_i\leq \theta_{ij}'] \right\}\right] \quad \circled{5} \\
&+ 2\cdot\sum\limits_{j=1}^{\lceil \log(2/r)\rceil}\E_{\Bb,\Bs}\left[\max_{S\in \cF}\sum_{i\in S}\left\{(\theta_{ij}'-s_i)^+\cdot\ind[b_i\geq \theta_{ij}'] \right\}\right] \quad \circled{6}.
\end{align*}
\notshow{$$\circled{2} \leq \underbrace{2\cdot\sum\limits_{j=1}^{\lceil \log(2/r)\rceil}\E_{\Bb,\Bs}\left[\max_{S\in \cF}\sum_{i\in S}\left\{(b_i-\theta_{ij}')^+\cdot\ind[s_i\leq \theta_{ij}'] \right\}\right]}_{\circled{5}} 
+ \underbrace{2\cdot\sum\limits_{j=1}^{\lceil \log(2/r)\rceil}\E_{\Bb,\Bs}\left[\max_{S\in \cF}\sum_{i\in S}\left\{(\theta_{ij}'-s_i)^+\cdot\ind[b_i\geq \theta_{ij}'] \right\}\right]}_{\circled{6}}.$$}
\end{lemma}

We refer to terms $\circled{3}$ and $\circled{5}$ as buyer surplus, and $\circled{4}$ and $\circled{6}$ as seller surplus. In the rest of this section we will bound each term separately.





\notshow{
\begin{align*}
    \circled{1}&
    \leq \textstyle\sum\limits_{j=1}^{\lceil \log(\frac{2}{r})\rceil}\E_{\Bb,\Bs}\left[\sum_i(b_i-\theta_{ij})^+\cdot\ind[i\in S^*(\Bb,\Bs)\wedge E_{ij}] \right]
    + \textstyle\sum\limits_{j=1}^{\lceil \log(\frac{2}{r})\rceil}\E_{\Bb,\Bs}\left[\sum_i(\theta_{ij}-s_i)^+\cdot\ind[i\in S^*(\Bb,\Bs)\wedge E_{ij}] \right]
\end{align*}


Similarly, for every $i\in [n], j\in 1,2,\ldots,\lceil \log(\frac{2}{r})\rceil$, let $\theta_{ij}'= G_i^{-1}(\frac{1}{2^{j}})$.
For every $i,j$, let $E_{ij}'$ be the event that both $b_i$ and $s_i$ fall in the bottom $\frac{1}{2^{j-1}}$ quantile of the seller's distribution, but $b_i$ does not fall in the bottom $\frac{1}{2^{j}}$ quantile of the seller's distribution.  That is, $G_i^{-1}(\frac{1}{2^j})\leq b_i\leq G_i^{-1}(\frac{1}{2^{j-1}})\wedge s_i\leq G_i^{-1}(\frac{1}{2^{j-1}})$. We have
\begin{align*}
    \circled{2}
    &\leq \textstyle\sum\limits_{j=1}^{\lceil \log(\frac{2}{r})\rceil}\E_{\Bb,\Bs}\left[\sum_i(b_i-\theta_{ij}')^+\cdot\ind[i\in S^*(\Bb,\Bs)\wedge E_{ij}'] \right]
    + \textstyle\sum\limits_{j=1}^{\lceil \log(\frac{2}{r})\rceil}\E_{\Bb,\Bs}\left[\sum_i(\theta_{ij}'-s_i)^+\cdot\ind[i\in S^*(\Bb,\Bs)\wedge E_{ij}'] \right]
\end{align*}

We refer to the two terms on the right hand side as \circled{5} and \circled{6}. We summarize the notation and argument above in the following lemma.

\begin{lemma}\label{lem: fb-breaking terms}
For any two-sided market between a single constrained additive buyer and $n$ unit-supply sellers, $\fb$ is bounded by the sum of terms  $\circled{3},\circled{4},\circled{5}$ and $\circled{6}$. 

Here, for every $i$, $x_i=\overline{F_i}^{-1}(\frac{r_i}{2}), y_i=G_i^{-1}(\frac{r_i}{2})$. For every $i\in [n], j\in 1,2,\ldots,\lceil \log(\frac{2}{r})\rceil$, $\theta_{ij}=\overline{F_i}^{-1}(\frac{1}{2^{j}})$, $\theta_{ij}'= G_i^{-1}(\frac{1}{2^{j}})$; $E_{ij}$ is the event that $\overline{F_i}^{-1}(\frac{1}{2^{j-1}})\leq s_i\leq \overline{F_i}^{-1}(\frac{1}{2^{j}})\wedge b_i\geq \overline{F_i}^{-1}(\frac{1}{2^{j-1}})$; $E_{ij}'$ is the event that $G_i^{-1}(\frac{1}{2^j})\leq b_i\leq G_i^{-1}(\frac{1}{2^{j-1}})\wedge s_i\leq G_i^{-1}(\frac{1}{2^{j-1}})$. 
\end{lemma}

We refer to the term $\circled{3}$ and $\circled{5}$ as buyer surplus, and $\circled{4}$ and $\circled{6}$ as seller surplus. In the rest of the section, we will design mechanisms to bound each term separately.

}
\subsection{Bounding Buyer Surplus}\label{sec:buyer surplus}

We bound terms $\circled{3}$ and $\circled{5}$ using fixed posted price mechanisms. Let $\pp$ denote the optimal GFT among all fixed posted price mechanisms. Recall that our market is not symmetric: a single multi-dimensional buyer with a feasibility constraint faces multiple single-dimensional sellers.
As a result, even for the general constrained-additive buyer, bounding buyer surplus is fairly straightforward using fixed price mechanisms that set $\theta_i^S=\theta_i^B=\theta_{ij}$ (or $\theta_i^S=\theta_i^B=\theta_{ij}'$) for each term.


\begin{lemma}\label{lem:term3&5}\label{lem:buyer-surplus-main-body}
For any $\{p_i\}_{i\in [n]}\in \mathbb{R}_{+}^n$,
$$\E_{\Bb,\Bs}\left[\max_{S\in \cF}\sum_{i\in S}\{(b_i-p_i)^+\cdot\ind[s_i\leq p_i]\}\right]\leq \pp.$$
Thus both $\emph{\circled{3}}$ and $\emph{\circled{5}}$ are upper bounded by $O(\log(\frac{1}{r}))\cdot \pp$.
\end{lemma}


\begin{proof}
Consider the fixed posted price mechanism $\MM$ with $\theta_i^S=\theta_i^B=p_i$. For every $\Bs$, let $A(\Bs)=\{i\in [n]~|~s_i\leq p_i\}$ be the set of available items. Then the buyer will choose the best set $S\subseteq A(\Bs),S\in \cF$ that maximizes $\sum_{i\in S}(b_i-p_i)^+$ (and not buy any item if $b_i - p_i \leq 0$ for all $i$). Thus the gains from trade $\sum_{i\in S}(b_i-s_i)$ is at least $\sum_{i\in S}(b_i-p_i)^+\geq 0$. We have

$$\gft(\MM)\geq \E_{\Bb,\Bs}\left[\max_{S\subseteq A(\Bs),S\in \cF}\sum_{i\in S}(b_i-p_i)^+\right]= \E_{\Bb,\Bs}\left[\max_{S\in \cF}\sum_{i\in S}\left\{(b_i-p_i)^+\cdot \ind[s_i\leq p_i]\right\}\right].$$

To bound terms \circled{3} and \circled{5}, just apply the above inequality with $p_i=\theta_{ij}$ (or $\theta_{ij}'$).
\end{proof}

\subsection{Bounding Seller Surplus for One Unit-Demand Buyer}\label{subsec:seller-surplus-UD}

In the remainder of this section, we will bound the seller surplus terms (\circled{4} and \circled{6}). As a warm-up, we first focus on the case where the buyer is unit-demand, i.e. the buyer is only interested in at most one item. In this case the prophet inequality suffices for our bound. 

\notshow{
Recall that
$$\textstyle\circled{4}=\sum_{j=1}^{\lceil \log(\frac{2}{r})\rceil}\E_{\Bb,\Bs}\left[\sum_i(\theta_{ij}-s_i)^+\cdot\ind[i\in S^*(\Bb,\Bs)]\cdot \ind[E_{ij}] \right]$$
$$\textstyle\circled{6}=\sum_{j=1}^{\lceil \log(\frac{2}{r})\rceil}\E_{\Bb,\Bs}\left[\sum_i(\theta_{ij}'-s_i)^+\cdot\ind[i\in S^*(\Bb,\Bs)]\cdot \ind[E_{ij}'] \right]$$
We will show that the terms in the sum can again be bounded by a constant multiple of $\pp$.  En route, we will use a the following lemma to make the terms in  \circled{4} look more like those in \circled{6}.

\begin{lemma} \label{lem:6 bounded by theta}
For $\theta_{ij}=\overline{F_i}^{-1}(\frac{1}{2^{j}})$ and $E_{ij}$ denoting the event that $\overline{F_i}^{-1}(\frac{1}{2^{j-1}})\leq s_i\leq \overline{F_i}^{-1}(\frac{1}{2^{j}})\wedge b_i\geq \overline{F_i}^{-1}(\frac{1}{2^{j-1}})$, then
$$\E_{\Bb,\Bs}\left[\sum_i(\theta_{ij}-s_i)^+\cdot\ind[i\in S^*(\Bb,\Bs)\wedge E_{ij}] \right] 
\leq 2\cdot\E_{\Bb,\Bs}\left[\max_{i\in [n]}\{(\theta_{ij}-s_i)^+\cdot\ind[b_i\geq \theta_{ij}]\}\right]$$
\end{lemma}

The proof of this lemma appears in \kgnote{add pointer}.  Then, we can bound the desired terms.
}

\begin{lemma}\label{lem:term4&6-UD}
When the buyer is unit-demand, for any $\{p_i\}_{i\in [n]}\in \mathbb{R}_{+}^n$, we have
$$\textstyle\E_{\Bb,\Bs}\left[\max_i\{(p_i-s_i)^+\cdot\ind[b_i\geq p_i]\}\right]\leq 2\cdot \pp.$$
Hence terms \emph{\circled{4}} and \emph{\circled{6}} are both upper-bounded by $O(\log(\frac{1}{r}))\cdot \pp$.
\end{lemma}

\begin{proof}
For every $i$, let $v_i=(p_i-s_i)^+\cdot\ind[b_i\geq p_i]$ be a random variable that depends on $b_i$ and $s_i$. Let $\Bv=\{v_i\}_{i\in [n]}$. Let $V_i$ be the distribution of $v_i$ where $b_i\sim \DD_i^B, s_i\sim \DD_i^S$, and $V=\times_{i=1}^nV_i$ be the distribution of $\Bv$. Then the LHS of the inequality in the Lemma statement is equal to $\E_{\Bv\sim V}[\max_iv_i]$.

Consider any threshold $\xi>0$.  Observe that  $v_i\geq\xi$ if and only if $b_i\geq p_i\wedge p_i-s_i\geq \xi$. Consider the fixed posted price mechanism $\MM$ with $\theta_i^B=p_i$ and $\theta_i^S=p_i-\xi$ for every $i\in [n]$.  
Whenever the buyer purchases some item $i$, we must have $b_i\geq p_i$ (the buyer buys) and $s_i\leq p_i-\xi$ (the seller sells), and the contributed GFT satisfies $b_i-s_i\geq p_i-s_i\geq \xi$. In addition, the buyer will purchase \emph{some} item if and only if there exists some $i$ such that $v_i\geq \xi$. 
Therefore we can apply the prophet inequality~\cite{KleinbergW12,krengel1978semiamarts,samuel1984comparison} with threshold $\xi=\frac{1}{2}\cdot \E_{\Bv\sim V}[\max_iv_i]$ to ensure that the GFT of mechanism $\MM$ is at least $\frac{1}{2}\E_{\Bv\sim V}[\max_iv_i]$. 
\end{proof}

\notshow{
For term $\circled{4}$, the proof is similar to term $\circled{6}$ but more involved. Now for every $j$, the $j^{\mathrm{th}}$ term in $\circled{4}$ is bounded by $\E_{\Bb,\Bs}\left[\sum_i(\theta_{ij}-s_i)^+\cdot\ind[i\in S^*(\Bb,\Bs)]\cdot \ind[b_i\geq \overline{F_i}^{-1}(\frac{1}{2^{j-1}})]\right]$ by relaxing the event $E_{ij}$. Thus before we can apply the above argument, we first need to show that the benchmark is no more than  $2\cdot\E_{\Bb,\Bs}\left[\sum_i(\theta_{ij}-s_i)^+\cdot\ind[i\in S^*(\Bb,\Bs)]\cdot \ind[b_i\geq \theta_{ij}]\right]$. Since $\theta_{ij}=\overline{F_i}^{-1}(\frac{1}{2^j})$, then $\Pr[b_i\geq \theta_{ij}]=1/2\cdot \Pr[b_i\geq \overline{F_i}^{-1}(\frac{1}{2^{j-1}})]$ and the two events $i\in S^*(\Bb,\Bs)$ and $b_i\geq \overline{F_i}^{-1}(\frac{1}{2^{j-1}})$ are positively correlated. 

Using Lemma~\ref{lem:6 bounded by theta}, we have that the terms from $\circled{4}$ are upper bounded as follows 
$$\E_{\Bb,\Bs}\left[\sum_i(\theta_{ij}-s_i)^+\cdot\ind[i\in S^*(\Bb,\Bs)\wedge E_{ij}] \right] 
\leq 2\cdot\E_{\Bb,\Bs}\left[\max_{i\in [n]}\{(\theta_{ij}-s_i)^+\cdot\ind[b_i\geq \theta_{ij}]\}\right],$$
hence the above analysis applies to \circled{4} as well, with $\theta_{ij}$ in place of $\theta_{ij}'$, with a factor of 2.

}

\notshow{

\begin{proof} 
\mingfeinote{Goes to Appendix.}


For every $j$, note that when $E_{ij}$ occurs, by definition, $b_i\geq \overline{F_i}^{-1}(\frac{1}{2^{j-1}})$. Thus
\begin{equation}\label{equ:proof-term4-UD}
\begin{aligned}
&\E_{\Bb,\Bs}\left[\sum_i(\theta_{ij}-s_i)^+\cdot\ind[i\in S^*(\Bb,\Bs)\wedge E_{ij}] \right]
\leq& \E_{\Bb,\Bs}\left[\sum_i(\theta_{ij}-s_i)^+\cdot\ind[i\in S^*(\Bb,\Bs)\wedge b_i\geq \overline{F_i}^{-1}(\frac{1}{2^{j-1}})] \right]    
\end{aligned}
\end{equation}

For every $i,\Bs,b_i$, define $q_i(b_i,\Bs)=\Pr_{b_{-i}}[i\in S^*(\Bb,\Bs)]$. Then we have that $q_i(b_i,\Bs)$ is non-decreasing in $b_i$. Since $\theta_{ij}=\overline{F_i}^{-1}(\frac{1}{2^{j}})$ and $\Pr_{b_i}\left[b_i\geq \overline{F_i}^{-1}(\frac{1}{2^{j}})\right] = \frac{1}{2} \Pr_{b_i}\left[b_i\geq \overline{F_i}^{-1}(\frac{1}{2^{j-1}})\right]$, we have

\begin{equation}\label{equ:proof-term4-2-UD}
\begin{aligned}
\E_{b_i}\left[q_i(b_i,\Bs)\cdot\ind\left[b_i\geq\theta_{ij}\right]\right]
\geq&q_i(\theta_{ij},\Bs)\cdot\Pr_{b_i}\left[b_i\geq\theta_{ij}\right]\\
= & q_i(\theta_{ij},\Bs)\cdot\Pr_{b_i}\left[\overline{F_i}^{-1}(\frac{1}{2^{j-1}})\leq b_i< \theta_{ij}\right]\\
\geq & \E_{b_i}\left[q_i(b_i,\Bs)\cdot\ind\left[\overline{F_i}^{-1}(\frac{1}{2^{j-1}})\leq b_i< \theta_{ij}\right]\right].
\end{aligned}
\end{equation}


Continuing the RHS of \eqref{equ:proof-term4-UD}, we have

\begin{align*}
    &\E_{\Bb,\Bs}\left[\sum_i(\theta_{ij}-s_i)^+\cdot\ind[i\in S^*(\Bb,\Bs)\wedge b_i\geq \overline{F_i}^{-1}(\frac{1}{2^{j-1}})] \right] \\
    =&\E_{\Bs}\left[\sum_i(\theta_{ij}-s_i)^+\cdot\E_{b_i}\left[q_i(b_i,\Bs)\cdot\ind[b_i\geq \overline{F_i}^{-1}(\frac{1}{2^{j-1}})]\right] \right]\\
    \leq& 2\cdot\E_{\Bs}\left[\sum_i(\theta_{ij}-s_i)^+\cdot\E_{b_i}\left[q_i(b_i,\Bs)\cdot\ind[b_i\geq\theta_{ij}]\right] \right]&\text{(Inequality~\eqref{equ:proof-term4-2-UD})}\\
    =&2\cdot\E_{\Bb,\Bs}\left[\sum_i(\theta_{ij}-s_i)^+\cdot\ind[i\in S^*(\Bb,\Bs)\wedge b_i\geq \theta_{ij}] \right] \\
    \leq&2\cdot\E_{\Bb,\Bs}\left[\max_{i\in [n]}\{(\theta_{ij}-s_i)^+\cdot\ind[b_i\geq \theta_{ij}]\}\right]&\text{($S^*(\Bb,\Bs)$ contains at most one item)}\\
\end{align*}

\end{proof}

We summarize the result for a unit-demand buyer in Theorem~\ref{thm:log1/r-UD}. It directly follows from Lemmas~\ref{lem: fb-breaking terms},~\ref{lem:term3&5}, and \ref{lem:term4&6-UD}. 

\begin{theorem}\label{thm:log1/r-UD}
When the buyer is unit-demand, $\fb\leq O(\log(\frac{1}{r}))\cdot \pp$.
\end{theorem}

}
\subsection{Bounding Seller Surplus with Selectability}\label{subsec:seller-surplus}

In this subsection we bound terms $\circled{4}$ and $\circled{6}$ for a more general class of constraints $\cF$ using a variant of a fixed posted price (FPP) mechanism which we call \emph{constrained} FPP. In the variant, the mechanism determines a (randomized) subconstraint $\cF'\subseteq \cF$ upfront. 
Then the buyer is only allowed to take a feasible set in $\cF'$ (among all items that the sellers agree to sell at prices $\{\theta_i^S\}_{i \in [n]}$) and pays the price $\theta_i^B$ for each item she takes.\footnote{Throughout this paper, we assume for simplicity that the buyer will purchase item $i$ when $b_i=\theta_i^B$ as long as the bundle remains feasible after including $i$. Without this tie-breaking rule, one can simply decrease the posted price for each item by an arbitrarily small value $\epsilon$, and the loss of GFT will be arbitrarily small.} Let $\gcfpp$ denote the the optimal GFT among all \cfpp mechanisms.\footnote{Note that FPP is a subclass of \cfpp, and therefore $\pp\leq\gcfpp$.}
Since all of the posted prices as well as the subconstraint are independent from the agents' reported profiles, the mechanism is DSIC and ex-post IR. The mechanism is also ex-post WBB since $\theta_i^B\geq \theta_i^S$ for all $i\in [n]$.   


\notshow{     

\floatname{algorithm}{Mechanism}
\begin{algorithm}[ht]
\begin{algorithmic}[1]
\REQUIRE For every $i$, the posted prices $\theta_i^B, \theta_i^S$; subconstraint $\cF'\subseteq \cF$; input profile $(\Bb,\Bs)$ 
\STATE Let $R=\{\theta_i^S\geq s_i\}$ be the items on the market.

\STATE Offer set $R$ to the buyer together with the constraint $\cF'$.

\STATE The buyer picks her favorite bundle $S^*$ from the set $R$ subject to the feasiblility constraint $\cF'$. Formally $S^*=\argmax_{S\in \cF', S\subseteq R}\sum_i(b_i-\theta_i^B)$.
\STATE For every item $i\in S^*$, the buyer receives item $i$ by paying $\theta_i^B$. Seller $i$ sells her item and receives $\theta_i^S$.
\end{algorithmic}
\caption{{\sf \rppname~Mechanism}}
\label{alg:pp}
\end{algorithm}

}


To present our result, we introduce a concept for downward-closed constraints called $(\delta,\eta)$-\emph{selectability}~\cite{FeldmanSZ16}.
Feldman et al. introduce $(\delta,\eta)$-selectability in the study of Online Contention Resolution Schemes (OCRS)~\cite{FeldmanSZ16}. 
An OCRS is an algorithm defined for the following online selection problem: There is a ground set $I$, and the elements are revealed one by one, with item $i$ \emph{active} with probability $x_i$ independent of the other items. The algorithm is only allowed to accept active elements and has to irrevocably make a decision whether to accept an element before the next one is revealed. Moreover, the algorithm can only accept a set of elements subject to a feasibility constraint $\cF$. We use the vector $x$ to denote active probabilities for the elements and $R(x)$ to denote the random set of active elements. 

\begin{definition}[relaxation]\label{def-relaxation}
We say that a polytope $P\subseteq [0,1]^{|I|}$ is a relaxation of $P_{\cF}$ if it contains the same $\{0,1\}$-points, i.e., $P\cap \{0,1\}^{|I|} = P_{\cF}\cap \{0,1\}^{|I|}$.
\end{definition}

\begin{definition}[Online Contention Resolution Scheme]\label{def-ocrs}
An \emph{Online Contention Resolution Scheme} (OCRS) for a polytope $P\subseteq [0,1]^{|I|}$ and feasibility constraint $\cF$ is an online algorithm that selects a feasible and active set $S\subseteq R(x)$ and $S\in \cF$ for any $x\in P$. 
A \textbf{greedy OCRS} $\pi$ greedily decides whether or not to select an element in each iteration: given the vector $x\in P$, it first determines a sub-constraint $\cF_{\pi,x}\subseteq \cF$. When element $i$ is revealed, it accepts the element if and only if $i$ is active and $S\cup \{i\}\in \cF_{\pi,x}$, where $S$ is the set of elements accepted so far. In most cases, we choose $P$ to be $P_{\cF}$, the convex hull of all characteristic vectors of feasible sets in $\cF$: $P_{\cF}=conv(\ind_S~|~S\in \cF)$.
\end{definition}

\begin{definition}[$(\delta,\eta)$-selectability \cite{FeldmanSZ16}]\label{def-selectablity}
For any $\delta,\eta\in (0,1)$,  
a greedy OCRS $\pi$ for $P$ and $\cF$ is $(\delta,\eta)$-selectable if for every $x\in \delta\cdot P$ and $i\in I$,
$$\Pr[S\cup \{i\}\in \cF_{\pi,x}, \forall S\subseteq R(x), S\in\cF_{\pi,x}]\geq \eta.$$
The probability is taken over the randomness of $R(x)$ and the subconstraint $\cF_{\pi,x}$. We slightly abuse notation and say that $\cF$ is $(\delta,\eta)$-selectable if there exists a $(\delta,\eta)$-selectable greedy OCRS for $P_{\cF}$ and $\cF$.
\end{definition}

The following lemma is adapted from~\cite{FeldmanSZ16} and connects $(\delta,\eta)$-selectability to \cfpp~mechanisms. Once again, the OCRS gives us both a GFT guarantee and a mechanism: variables $v_i$ correspond to the bound on seller surplus, buyer item prices are $\{p_i\}_{i\in [n]}$, seller prices are $\{p_i - \xi_i\}_{i\in [n]}$, and the subconstraint is suggested by the OCRS.

\begin{lemma}\label{lem:feldman-selectable}
Suppose there exists a $(\delta,\eta)$-selectable greedy OCRS $\pi$ for the polytope $P_{\cF}$, for some $\delta,\eta\in (0,1)$. Fix any $\{p_i\}_{i\in [n]}\in \mathbb{R}_{+}^n$. For every $i\in [n]$, let $v_i=(p_i-s_i)^+\cdot \ind[b_i\geq p_i]$. For any $\Bq \in P_{\cF}$ that satisfies $q_i\leq \Pr_{b_i,s_i}[b_i\geq p_i>s_i] $ $\forall i$, let $\xi_i=p_i-G_i^{-1}(q_i/\Pr[b_i\geq p_i])$.\footnote{When $q_i\leq \Pr_{b_i,s_i}[b_i\geq p_i>s_i]$, $q_i/\Pr[b_i\geq p_i]\leq 1$. Thus $\xi_i$ is well-defined.} We have
$$\textstyle\sum_i\E_{b_i,s_i}\left[v_i\cdot \ind[v_i\geq \xi_i]\right]\leq \frac{1}{\delta\eta}\cdot \gcfpp.$$
Moreover, there exists a choice of $\Bq$ such that
$$\textstyle\E_{\Bb,\Bs}\left[\max_{S\in \cF}\sum_{i\in S}\left\{(p_i-s_i)^+\cdot\ind[b_i\geq p_i]\right\}\right]\leq \textstyle\sum_i\E_{b_i,s_i}\left[v_i\cdot \ind[v_i\geq \xi_i]\right]\leq \frac{1}{\delta\eta}\cdot \gcfpp.$$
\end{lemma}

\begin{prevproof}{Lemma}{lem:feldman-selectable}
Let $V_i$ denote the distribution of $v_i$ when $b_i\sim \DD_i^B, s_i\sim \DD_i^S$. Let $\Bv=\{v_i\}_{i\in [n]}$. Let $\hat{\textbf{q}}$ be a scaled-down vector of $\textbf{q}$ such that $\hat{q}_i=\delta\cdot q_i$ for every $i\in [n]$ and $\hat{\xi_i}=p_i-G_i^{-1}(\hat{q}_i/\Pr[b_i\geq p_i])$. This is also well-defined since $\hat{q}_i<q_i\leq\Pr[b_i\geq p_i]$. As $q \in P_{\cF}$, then $\hat{q}\in \delta\cdot P_{\cF}$. Consider the \cfpp~mechanism $\MM$ with buyer posted prices $\{p_i\}_{i\in [n]}$, seller posted prices $\{p_i-\hat{\xi}_i\}_{i\in [n]}$, and subconstraint $\cF_{\pi,\hat{q}}\in \cF$ stated in Definition~\ref{def-selectablity}. 

Fix any item $i\in [n]$. We say item $i$ as active if $v_i\geq \hat{\xi}_i$. Similarly to Section~\ref{subsec:seller-surplus-UD},  $v_i\geq \hat{\xi}_i$ if and only if $b_i\geq p_i\wedge s_i\leq p_i-\hat{\xi}_i$. That is, $i$ is active if and only if item $i$ is on the market and the buyer can afford it, which by choice of $\hat \xi_i$ happens independently across all $i$ with probability $\Pr_{v_i}[v_i\geq \hat{\xi}_i]=\Pr_{b_i,s_i}[p_i-s_i\geq \hat{\xi}_i\wedge b_i\geq p_i]=\hat{q}_i$. 

Then for any $\Bv$, the set of active items is $R(\Bv)=\{j\in [n]:~ v_j\geq \hat{\xi}_j\}$. By $(\delta,\eta)$-selectability (Definition~\ref{def-selectablity}) and the fact that $\hat{q}\in \delta\cdot P_{\cF}$, we have

\begin{equation}\label{equ:lem-selectable1}
\Pr_{\pi,\Bv}[S\cup \{i\}\in \cF_{\pi,\hat{q}}, \forall S\subseteq R(\Bv), S\in\cF_{\pi,\hat{q}}]\geq \eta.    
\end{equation}

Note that for the sets $S\in\cF_{\pi,\hat{q}}$ that have $i\in S$, then $S\cup \{i\}\in \cF_{\pi,\hat{q}}$ with probability 1.  Thus, if we require $S \subseteq R(\Bv)\backslash\{i\}$ instead, it can not be that $i \in S$, and so the following LHS occurs with equal probability, allowing us to rewrite inequality~\eqref{equ:lem-selectable1} as follows: 

\begin{equation}\label{equ:lem-selectable2}
\Pr_{\pi,\Bv}[S\cup \{i\}\in \cF_{\pi,\hat{q}}, \forall S\subseteq R(\Bv)\backslash \{i\}, S\in\cF_{\pi,\hat{q}}]\geq \eta.
\end{equation}

For any $\Bv_{-i}$, let $R_i(\Bv_{-i})=\{j\neq i:~ v_j\geq \hat{\xi}_j\}$. Then inequality~\eqref{equ:lem-selectable2} is equivalent to

$$\Pr_{\pi,v_{-i}}[S\cup \{i\}\in \cF_{\pi,\hat{q}}, \forall S\subseteq R_i(\Bv_{-i}), S\in\cF_{\pi,\hat{q}}]\geq \eta.$$


Define event $A_i=\left\{\Bv_{-i}:~S\cup \{i\}\in \cF_{\pi,\hat{q}}, \forall S\subseteq R_i(\Bv_{-i}), S\in\cF_{\pi,\hat{q}}\right\}$. We will argue that item $i$ must be in the buyer's favorite bundle $S^*$ 
when both of the following conditions are satisfied: (i) $v_i\geq \hat{\xi}_i$, and (ii) event $A_i$ happens. Note that in $\MM$, the set of items in the market is $T=\{j:s_j\leq p_j-\hat{\xi}_j\}$, thus $S^*=\argmax_{S\subseteq T, S\in \cF_{\pi,\hat{q}}}\sum_{j\in S}(b_j-p_j)$. Suppose by way of contradiction that both conditions are satisfied but $i\not\in S^*$. Clearly, for every $j\in S^*$, we have $b_j\geq p_j$, otherwise removing $j$ from $S^*$ will give the buyer greater utility. In addition, we have $s_j\leq p_j-\hat{\xi}_j$, so $S^*\subseteq R(\Bv)$. By definition, $S^*$ must lie in $\cF_{\pi,\hat{q}}$. Since event $A_i$ occurs, then $S^*\cup \{i\}\in \cF_{\pi,\hat{q}}$. As $v_i\geq \hat{\xi}_i$, this implies that $b_i\geq p_i$. Thus adding $i$ to $S^*$ keeps the set feasible and does not decrease the buyer's utility $\sum_{j\in S^*}(b_j-p_j)$. Thus $i\in S^*$ (see footnote \ref{footnote:tie-breaking}). 
This is a contradiction.    

Note that condition (i) and (ii) are independent. Thus for every $b_i$ and $s_i$ such that $b_i\geq p_i\wedge s_i\leq p_i-\hat{\xi}_i$ (or equivalently $v_i\geq\hat{\xi}_i$), the expected GFT of item $i$ over $b_{-i},s_{-i}$ is at least 
$$\Pr[A_i]\cdot (b_i-s_i)\geq \eta\cdot (p_i-s_i)=\eta\cdot v_i.$$

Thus

$$\gft(\MM)\geq \eta\cdot\sum_i\E_{v_i\sim V_i}[v_i\cdot \ind[v_i\geq \hat{\xi}_i]]\geq \delta\eta\cdot\sum_i\E_{v_i\sim V_i}[v_i\cdot \ind[v_i\geq \xi_i]],$$

where the last inequality is because for every $i$, we have $\E[v_i|v_i\geq \hat{\xi}_i]\geq \E[v_i|v_i\geq \xi_i]$ and $\Pr[v_i\geq \hat{\xi}_i]=\hat{q}_i=\delta\cdot \Pr[v_i\geq \xi_i]$.

For the second inequality stated in the lemma, note that 
$$\E_{\Bb,\Bs}\left[\max_{S\in \cF}\sum_{i\in S}\left\{(p_i-s_i)^+\cdot\ind[b_i\geq p_i]\right\}\right]=\E_{\Bv}\left[\max_{S\in \cF}\sum_{i\in S}v_i\right].$$
For every $\Bv$, let $\hat{S}(\Bv)=\argmax_{S\in \cF}\sum_{i\in S}v_i$, and break ties in favor of the set with smaller size. 
For every $i$, let $q_i=\Pr_{\Bv}[i\in \hat{S}(\Bv)]$ be the probability that $i$ is in the maximum weight feasible set. We have that $\Bq=\{q_i\}_{i\in [n]}\in P_{\cF}$. 
Also for every $i$, $q_i=\Pr_{\Bv}[i\in \hat{S}(\Bv)]\leq \Pr[v_i>0]=\Pr[b_i\geq p_i>s_i]$.
Moreover,

$$\E_{\Bv}\left[\max_{S\in \cF}\sum_{i\in S}v_i\right]=\sum_{i\in [n]}\E_{\Bv}\left[v_i\cdot \ind[i\in \hat{S}(\Bv)]\right]\leq \sum_{i\in [n]}\E_{v_i\sim V_i}\left[v_i\cdot \ind[v_i\geq \xi_i]\right]$$

The inequality follows from the fact that for every $i$, both sides integrate the random variable $v_i$ with a total probability mass $q_i$, while the right hand side integrates $v_i$ at the top $q_i$-quantile. 
\end{prevproof}

For each $j$ in the summation, choose $p_i$ from Lemma~\ref{lem:feldman-selectable} to be $\theta_{ij}$ (or $\theta_{ij}'$). Then both terms $\circled{4}$ and  $\circled{6}$ are bounded by   
$\frac{\log(1/r)}{\delta\eta}\cdot \gcfpp$. Theorem~\ref{thm:log1/r} then follows directly from Lemmas~\ref{lem:breaking-terms-1},~\ref{lem:breaking-terms-2},~\ref{lem:buyer-surplus-main-body}, and \ref{lem:feldman-selectable}.

Feldman et al.~\cite{FeldmanSZ16} show that many natural constraints---including matroids, matchings, knapsack, and their compositions---are $(\delta,\eta)$-selectable for some constants $\delta$ and $\eta$. For all of these, Theorem~\ref{thm:log1/r} implies that $\gcfpp$ is an $O(\log(1/r))$-approximation to $\fb$. See Appendix~\ref{sec:appx-constraints} for details.
\begin{theorem}\label{thm:log1/r}
Suppose the buyer's feasibility constraint $\cF$ is $(\delta,\eta)$-selectable for some $\delta,\eta\in (0,1)$. Then $\fb\leq O(\frac{\log(1/r)}{\delta\eta})\cdot \gcfpp$.
\end{theorem}


\subsection{General Constrained-Additive Buyer}\label{subsec:proof-downward-close}

In this section, we consider the case of a general constrained-additive buyer, and prove an $O(\log(n)\cdot \log(1/r))-$approximation to $\fb$ using \cfpp mechanisms. Note that Lemmas~\ref{lem:breaking-terms-1},~\ref{lem:breaking-terms-2}, and~\ref{lem:buyer-surplus-main-body} still hold in this setting. It is sufficient to bound the seller surplus term with $\gcfpp$. 

Throughout this section, we will use the following variant of FPP mechanisms: Other than posted prices, the mechanism also determines an integer $h>0$ upfront.
The buyer can purchase any set of items of size at least $h$ by paying the posted prices for each item in the set; otherwise, she leaves with nothing.
This is a subclass of \cfpp, with subconstraint $\cF'=\{S\mid S\in \cF\wedge |S|\geq h\}\subseteq \cF$.\footnote{If $h=1$, the mechanism becomes a standard FPP mechanism without any subconstraint $\cF'$.} 

\begin{lemma}\label{lem:log(n)-downward-close}
For any $\{p_i\}_{i\in [n]}\in \mathbb{R}_{+}^n$, $$\textstyle\cA=\E_{\Bb,\Bs}\left[\max_{T\in\cF}\sum_{i\in T}\{(p_i-s_i)^+\cdot\ind[b_i\geq p_i]\}\right]\leq O(\log(n))\cdot \gcfpp.$$
Hence terms \emph{\circled{4}} and \emph{\circled{6}} are both upper-bounded by $O(\log(n)\cdot\log(\frac{1}{r}))\cdot \gcfpp$.
\end{lemma}

For every $i\in [n]$, again construct random variables $v_i=(p_i-s_i)^+\cdot\ind[b_i\geq p_i]$. The main issue here is that in an FPP mechanism, say with posted prices $\theta_i^B=\theta_i^S=p_i$, the buyer will pick the maximum weight feasible set (among all items that sellers are willing to sell) according to weight $b_i-p_i$ (her utility). However, this might be far from the set used in the benchmark, i.e. the maximum weight feasible set according to weight $p_i-s_i$. In the previous section (when the constraint $\cF$ had selectability), by setting different prices for both sides and adding a more restrictive constraint, we guaranteed that if both the buyer and seller accept the posted prices for some item, then the buyer would purchase this item with at least constant probability. For general downward-closed $\cF$, it is unclear how to achieve this property with a \cfpp mechanism.

For every $\Bb,\Bs$, let $T^*(\Bb,\Bs)=\argmax_{T\in \cF}\sum_{i\in T}v_i$ be the optimal set used in the benchmark. We divide $\cA$ into three terms according to the value of $v_i$ when $i$ is in this optimal set: $v_i<\cA/2n$, $v_i\in [\cA/2n,2n\cA]$ and $v_i>2n\cA$. 
Denote the three terms $\cA_S, \cA_M, \cA_L$ accordingly. First we notice that the contribution of $\cA_S$ is at most a constant fraction of $\cA$, as
$\textstyle\cA_S=\E_{\Bb,\Bs}\left[\sum_iv_i\cdot\ind[i\in T^*(\Bb,\Bs)\wedge v_i<\frac{\cA}{2n}]\right]<\E_{\Bb,\Bs}\left[\frac{\cA}{2n}\cdot n\right]=\frac{\cA}{2}$. 

For $\cA_L$, in Lemma~\ref{lem:A_L}, we first prove that $\Pr_{b_i,s_i}[b_i\geq p_i\wedge p_i-s_i\geq 2n\cA]\leq \frac{1}{2n}$ holds for every $i$. This implies that in a standard FPP mechanism (where $h=1$) with $\theta_i^B=p_i,\theta_i^S=(p_i-2n\cA)^+$ for all $i$, the buyer purchases each item $i$ with probability at least $\frac{1}{2}$ if both the buyer and seller $i$ accept the posted prices.  First, we will need Lemma~\ref{lem:fpp}.

\begin{lemma}\label{lem:fpp}
Given any \cfpp mechanism $\MM$ with posted prices $\{\theta_i^B\}_{i\in [n]}$, $\{\theta_i^S\}_{i\in [n]}$ and $h=1$, suppose $\sum_i\Pr[b_i\geq \theta_i^B\wedge s_i\leq \theta_i^S]\leq \frac{1}{2}$. Then
$$\gft(\MM)\geq \frac{1}{2}\sum_i\E_{b_i,s_i}\left[(b_i-s_i)\cdot \ind[b_i\geq \theta_i^B\wedge s_i\leq \theta_i^S]\right].$$
\end{lemma}
\begin{proof}
For any item $i$, the buyer will purchase item $i$ if both of the following events happen: 
\begin{enumerate}
\item $b_i\geq \theta_i^B$ and $s_i\leq \theta_i^S$;
\item For all items $k\not=i$, either $s_k>\theta_k^S$ or $b_k<\theta_k^B$. 
\end{enumerate}

By the union bound, the second event happens with probability at least $1-\sum_{k\not=i}\Pr[b_i\geq \theta_i^B\wedge s_i\leq \theta_i^S]\geq\frac{1}{2}$. Since both events are independent, we have

$$\gft(\MM)\geq \frac{1}{2}\sum_i\E_{b_i,s_i}\left[(b_i-s_i)\cdot \ind[b_i\geq \theta_i^B\wedge s_i\leq \theta_i^S]\right].$$

\end{proof}

\begin{lemma}\label{lem:A_L}
$\cA_L=\E_{\Bb,\Bs}\left[\sum_iv_i\cdot\ind[i\in T^*(\Bb,\Bs)\wedge v_i>2n\cA]\right]\leq 2\cdot \pp$.
\end{lemma}

\begin{proof}
Consider the FPP mechanism with $\theta_i^B=p_i,\theta_i^S=(p_i-2n\cA)^+$ for all $i$ (and $h=1$). 

Note that for every $i\in [n]$, it must hold that $\Pr_{b_i,s_i}[b_i\geq p_i\wedge p_i-s_i\geq 2n\cA]\leq \frac{1}{2n}$. In fact,

$$\cA\geq \E_{b_i,s_i}[(p_i-s_i)^+\cdot \ind[b_i\geq p_i]]\geq 2n\cA\cdot \Pr_{b_i,s_i}[b_i\geq p_i\wedge p_i-s_i\geq 2n\cA].$$
Thus by Lemma~\ref{lem:fpp} and the fact that $b_i-s_i\geq p_i-s_i$ when $b_i\geq p_i$, we have

$$\pp\geq \frac{1}{2}\sum_i\E_{b_i,s_i}\left[(p_i-s_i)\cdot \ind[b_i\geq p_i\wedge p_i-s_i\geq 2n\cA]\right]\geq \frac{1}{2}\cdot \cA_L.$$
\end{proof}

In Lemma~\ref{lem:A_M} we bound $\cA_M$, which is the primary challenge for this approximation.

\begin{lemma}\label{lem:A_M}
$\cA_M\leq O(\log(n))\cdot \gcfpp$.
\end{lemma}

\begin{proof}

We further divide the interval $[\cA/2n,2n\cA]$ into $O(\log(n))$ buckets, where in each bucket $k$, $v_i$ falls in the range $[L_k,2L_k]$ for some $L_k$. Formally, for any $k\in \{1,2,...,\lceil 2\log(n)+2\rceil\}$, let $L_k=2^k\cdot \frac{\cA}{4n}$. We have
$$\cA_M\leq\sum_{k=1}^{\lceil 2\log(n)+2\rceil}\E_{\Bb,\Bs}\left[\sum_iv_i\cdot\ind[i\in T^*(\Bb,\Bs)\wedge L_k\leq v_i\leq 2L_k]\right].$$

In the rest of the proof, we will show that for any $k$, there exists some constant $c>0$ such that
$$\E_{\Bb,\Bs}\left[\sum_iv_i\cdot\ind[i\in T^*(\Bb,\Bs)\wedge L_k\leq v_i\leq 2L_k]\right]\leq c\cdot \gcfpp.$$

Fix any $k$. 
For every $i\in [n]$, let $t_i^{(k)}=\frac{v_i}{2L_k}\cdot \ind[L_k\leq v_i\leq 2L_k]$. This is a random variable in $[\frac{1}{2},1]$. Note that all random variables $t=\{t_i^{(k)}\}_{i\in [n]}$ are independent. Let $Z(t)=\max_{T\in \cF}\sum_{i\in T}t_i^{(k)}$. Then the contribution to $\cA_L$ from values in this range is bounded by the expectation of the random variable $Z(t)$:
$$\E_{\Bb,\Bs}\left[\sum_iv_i\cdot\ind[i\in T^*(\Bb,\Bs)\wedge L_k\leq v_i\leq 2L_k]\right]\leq \E_{\Bb,\Bs}\left[\max_{S\in \cF}\sum_{i\in S}v_i\cdot\ind[L_k\leq v_i\leq 2L_k]\right]=2L_k\E_t[Z(t)].$$


\notshow{
The proof of Lemma~\ref{lem:Z-concentrate} can be found in Section~\ref{subsec:Z-concentrate}.

\begin{lemma}\label{lem:Z-concentrate}
For any $c\in (0,1)$,
$$\Pr_t[Z(t)\geq c\cdot \E[Z(t)]]\geq \frac{(1-c)^2}{1+1/(2\cdot\E[Z(t)])}$$
\end{lemma}
}

Now consider the \cfpp mechanism with $\theta_i^B=p_i$ and $\theta_i^S=(p_i-L_k)^+$ for every $i$ (the threshold $h$ is determined later). Then in the mechanism, whenever the buyer purchases an item, the contributed GFT is at least $L_k$. Thus it is sufficient to show that the expected size of the purchasing set is at least a constant factor of $\E[Z(t)]$. Note that $Z(t)$ is a random variable on $t$, which is the maximum weight feasible set over $n$ independent random variables in $[0,1]$. In Lemma~\ref{lem:Z-concentrate}, we prove that $Z(t)$ concentrates near its mean. The proof is postponed to Section~\ref{subsec:Z-concentrate}.

\begin{lemma}\label{lem:Z-concentrate}
For any $c\in (0,1)$,
$$\Pr_t[Z(t)\geq c\cdot \E[Z(t)]]\geq \frac{(1-c)^2}{1+1/\E[Z(t)]}.$$
\end{lemma}

We first suppose that $\E[Z(t)]\geq \frac{1}{4}$. By applying Lemma~\ref{lem:Z-concentrate} with $c=\frac{1}{2}$, we get
$$\Pr_t\left[Z(t)\geq \frac{\E[Z(t)]}{2}\right]\geq \frac{1}{20}.$$
Let $h=\max\left\{\lfloor\frac{\E[Z(t)]}{2}\rfloor,1\right\}$. In mechanism $\MM_k$, note that for every $i$, $t_i^{(k)}>0$ implies that item $i$ is on the market and that the buyer can afford it. With probability at least $\frac{1}{20}$, $Z(t)\geq h$, which implies that the item set $\{i\mid i\in \argmax_{S\in \cF}\sum_{i\in S}t_i^{(k)}\wedge t_i^{(k)}>0\}$ is a feasible set of size at least $h$.  (Recall that all $t_i^{(k)}$ are in $[\frac{1}{2},1]$). In this scenario, the buyer will purchase a set of items of size at least $h$. For every item $i$ she purchases, the contributed GFT is $b_i-s_i\geq \theta_i^B-\theta_i^S=L_k$. Thus, $\gft(\MM_k)\geq \frac{1}{20}\cdot h\cdot L_k$. Readers who are familiar with mechanism design may notice that the role of the size threshold $h$ is similar to an ``entry fee'' in the posted price mechanism in auctions~\cite{BabaioffILW14,CaiDW16,CaiZ17,ChawlaM16,RubinsteinW15,Yao15}, though the buyer doesn't have to pay extra money to attend the auction. It guarantees that the buyer will purchase at least $h$ items when she can afford it, as otherwise she gets no utility.


When $\E[Z(t)]\geq \frac{1}{4}$, we have $h\geq \frac{\E[Z(t)]}{4}$. Thus
$$\E_{\Bb,\Bs}\left[\sum_iv_i\cdot\ind[i\in T^*(\Bb,\Bs)\wedge L_k\leq v_i\leq 2L_k]\right]\leq 160\cdot \gcfpp.$$

Now we consider the case where $\E[Z(t)]<\frac{1}{4}$. For every $i$, let $q_i=\Pr[t_i^{(k)}>0]=\Pr_{b_i,s_i}[b_i\geq p_i\wedge p_i-s_i\in [L_k,2L_k]]$. Then it holds that
$$\Pr[\forall i, t_i^{(k)}=0]=\prod_{i=1}^n(1-q_i)>\frac{1}{2}.$$
This is because if there exists $i$ such that $t_i^{(k)}>0$, then $Z(t)=\max_{T\in \cF}\sum_{i\in T}t_i^{(k)}\geq \frac{1}{2}$ as $t_i^{(k)}\in [\frac{1}{2},1]$ for every $i$. Thus if $\Pr[\forall i, t_i^{(k)}=0]\leq \frac{1}{2}$, then $\E[Z(t)]\geq \frac{1}{4}$, which leads to a contradiction.

Consider the \cfpp mechanism $\MM$ with $\theta_i^B=p_i$, $\theta_i^S=(p_i-L_k)^+$, and $h=1$. For every $i$, define event $E_i=\{t\mid t_i^{(k)}>0\wedge t_j^{(k)}=0,\forall j\not=i\}$. Note that $t_i^{(k)}>0$ implies that seller $i$ accepts price $\theta_i^S$ and also the buyer can afford item $i$. Under event $E_i$, there is at least one item on the market that the buyer can afford, i.e. item $i$. Thus the buyer must purchase \emph{some} item $j$ on the market that she can afford (possibly item $i$). For this item $j$, we have $b_j\geq \theta_j^B$ and $s_j\leq \theta_j^S$. Thus the contributed GFT is at least $b_j-s_j\geq p_j-s_j\geq L_k$. Since all $E_i$s are disjoint events, we have
$$\gft(\MM)\geq \sum_i \Pr[E_i]\cdot L_k=L_k\cdot \sum_i q_i\cdot \prod_{j\not=i}(1-q_j)\geq L_k\cdot \sum_i q_i\cdot \prod_{j}(1-q_j)>\frac{1}{2}L_k\cdot \sum_i q_i,$$
where the equality uses the fact that all $t_i^{(k)}$s are independent. On the other hand, since $t_i^{(k)}\leq 1$ for any $i$,
$$\E[Z(t)]\leq \E\left[\sum_i t_i^{(k)}\cdot\ind[t_i^{(k)}>0]\right]\leq \sum_i q_i.$$
Thus 

$$\E_{\Bb,\Bs}\left[\sum_iv_i\cdot\ind[i\in T^*(\Bb,\Bs)\wedge L_k\leq v_i\leq 2L_k]\right]\leq 2L_k\cdot \E[Z(t)]\leq 4\cdot \gcfpp.$$

Summing the inequality over all $k$ finishes the proof.

\end{proof}

\begin{prevproof}{Lemma}{lem:log(n)-downward-close}
By Lemmas~\ref{lem:A_L},~\ref{lem:A_M}, and the fact that $\cA_S\leq \frac{\cA}{2}$, we have that
$$\cA\leq 2(\cA_M+\cA_L)\leq O(\log(n))\cdot \gcfpp.$$
\end{prevproof}

Theorem~\ref{thm:downward-close} summarizes our result for a general constrained-additive buyer. It directly follows from Lemmas~\ref{lem:breaking-terms-1},~\ref{lem:breaking-terms-2},~\ref{lem:buyer-surplus-main-body}, and~\ref{lem:log(n)-downward-close}.  

\begin{theorem}\label{thm:downward-close}
For any downward-closed constraint $\cF$,
$\fb\leq O(\log(n)\cdot \log(\frac{1}{r}))\cdot \gcfpp$.
\end{theorem}

\subsubsection{Proof of Lemma~\ref{lem:Z-concentrate}}\label{subsec:Z-concentrate}

We recall the statement of Lemma~\ref{lem:Z-concentrate}: \emph{For any $c\in (0,1)$,}
$$\Pr_t[Z(t)\geq c\cdot \E[Z(t)]]\geq \frac{(1-c)^2}{1+1/\E[Z(t)]}.$$

Recall that $Z(t)=\max_{T\in \cF}\sum_{i\in T}t_i^{(k)}$. In the proof we will omit the superscript $k$ as it is fixed. The random seed $t$ is also omitted if clear from context.

\begin{lemma}\label{lem:P-Z inequ}
(Paley-Zygmund Inequality~\cite{paley1932note}) For any random variable $Z\geq 0$ with finite variance, for any $c\in [0,1]$,
$$\Pr[Z\geq c\cdot \E[Z]]\geq (1-c)^2\cdot\frac{\E[Z]^2}{\var[Z]+\E[Z]^2}.$$
\end{lemma}

To use Lemma~\ref{lem:P-Z inequ}, we only need to show an upper bound on $\var[Z(t)]$. 

\begin{lemma}
$\var[Z(t)]\leq \E[Z(t)]$.
\end{lemma}
\begin{proof}
By the Efron-Stein Inequality~\cite{efron1981jackknife},

$$\var[Z(t)]\leq \frac{1}{2}\sum_i \E_{t_i,t_i',t_{-i}}[(Z(t_i,t_{-i})-Z(t_i',t_{-i}))^2]=\sum_i \E_{t_{-i}}[\var[Z(t)|t_{-i}]].$$
Here $t_i'$ shares the same distribution with $t_i$ (a fresh sample). Note that for every fixed $t_{-i}$, $\var_{t_i}[Z(t_i,t_{-i})]\leq \E_{t_i}[(Z(t_i,t_{-i})-a)^2]$ for any real constant $a$. For every $i$, let $Z_i(t_{-i})=\max_{T\in \cF,i\not=T}\sum_{j\in T}t_j$, which only depends on $t_{-i}$. We have
$$\var[Z(t)]\leq \sum_i \E_{t_{-i}}[\var[Z(t)|t_{-i}]]\leq \sum_i \E[(Z(t)-Z_i(t_{-i}))^2]\leq \sum_i \E[Z(t)-Z_i(t_{-i})],$$
where the last inequality follows from the fact that $Z_i(t_{-i})\leq Z(t)\leq Z_i(t_{-i})+1$, as every random variable $t_j\in [\frac{1}{2},1]$.

Now fix any $t$. Let $T^*=\argmax_{T\in \cF}\sum_{j\in T}t_j$. Then for every $i$, by the definition of $Z_i$, $\sum_{j\in T^*\backslash\{i\}}t_j\leq Z_i(t_{-i})$. Thus $$\sum_iZ_i(t_{-i})\geq \sum_i\sum_{j\in T^*\backslash\{i\}}t_j=(n-1)\cdot \sum_{j\in T^*}t_j=(n-1)\cdot Z(t).$$

Hence,
$$\var[Z(t)]\leq \sum_i \E[Z(t)-Z_i(t_{-i})]\leq \E[Z(t)].$$
\end{proof}

\section{An Unconditional Approximation for a Single Constrained-Additive Buyer}\label{sec:log n approx}

In this section, we prove Theorem~\ref{thm:logn}, an unconditional $O(\log n)$-approximation when the buyer's feasibility constraint is selectable, and an unconditional 
$O(\log^2(n))$-approximation for a general constrained-additive buyer---without dependence on distributional parameters. The result combines the $\log (1/r)$-approximation and a novel mechanism---the \emph{seller adjusted posted price mechanism}.  


\begin{theorem}\label{thm:logn}
Suppose the buyer's feasibility constraint $\cF$ is $(\delta,\eta)$-selectable for some $\delta,\eta\in (0,1)$. Then there exists a DSIC, ex-post IR, ex-ante WBB mechanism $\MM$ such that
$\opt\leq O(\frac{\log n}{\delta\cdot \eta})\cdot \gft(\MM).$  Moreover, for a general constrained-additive buyer, there exists a DSIC, ex-post IR, ex-ante WBB mechanism $\MM$ such that
$\opt\leq O(\log^2 (n))\cdot \gft(\MM).$
\end{theorem}

\subsection{An Upper Bound of the Second-Best GFT}\label{sec:UB of SB}
Formally, we use $\opt$ to denote the optimal GFT attainable by any BIC, IR, ex-ante WBB mechanism. Notice that the GFT of any two-sided market mechanism can be broken down into the buyer's expected utility of this mechanism, plus the sum of all sellers' expected utilities (or profit), plus the difference between buyer's and sellers' expected payment. We show that the $\opt$ is upper bounded by the sum of the designers' utilities in two related \textbf{one-sided markets}: the \emph{super seller auction} and the \emph{super buyer procurement auction}. 

\vspace{-.15in}
\paragraph{Super Seller Auction.} Consider a one-sided market, where the designer is the super seller who owns all the items, replacing all the original sellers. 
The buyer is the same as in our two-sided market setting. The super seller designs a mechanism to sell the items to the buyer. The main difference between the super seller auction and the original two-sided market is that the mechanism only needs to be BIC and IR for the buyer, and does not have any incentive compatibility constraint for the super seller.  We use $\opts$ to denote the maximum profit (revenue minus her cost) achievable by any BIC and IR mechanism in the super seller auction. 

To avoid ambiguity in further proofs, for every subset $T\subseteq [n]$ and downward-closed feasibility constraint $\cJ$ with respect to $T$, we let $\opts(T,\mathcal{J})$ denote the optimal profit in the following super seller auction: the super seller owns the set of items in $T$ and  has cost $s_i\sim \DD_i^S$ for every item $i\in T$. The buyer has value $b_i\sim \DD_i^B$ for every item $i\in T$ and is additive subject to constraint $\cJ$. We slightly abuse notation and write $\opts(T,\text{ADD})$ if the buyer is additive ($\cJ=2^T$) and $\opts(T,\text{UD})$ if the buyer is unit-demand ($\cJ=\{\{i\}:~i\in T\}$). Clearly, $\opts=\opts([n],\mathcal{F})$.

\vspace{-.15in}
\paragraph{Super Buyer Procurement Auction.}

Similarly, let the \emph{super buyer procurement auction} be the one-sided market where the super buyer (same as the real buyer) designs the mechanism to procure items from the sellers. Here the mechanism only needs to be BIC and IR for all of the sellers, but not the buyer. 
We use $\optb$ to denote the maximum utility (value minus payment) of the super buyer attainable by any BIC and IR mechanism in the super buyer procurement auction. 

First, we extend the upper bound of Brustle et al.~\cite{BrustleCWZ17} to our multi-dimensional setting,
\begin{align*}
    \opt&\leq \optb+\opts  & \text{(Lemma~\ref{lem:ub-opt})}
\end{align*}
and then, we prove an analog of the ``Marginal Mechanism Lemma''~\cite{CaiH13,HartN12} for the optimal profit (Lemma~\ref{lem:separation}). The proofs of both extensions appear in Appendix~\ref{apx:UB of SB}. We partition the items into the set of ``likely to trade'' items, that is, items with trade probability $r_i=\Pr_{b_i,s_i}[b_i\geq s_i]\geq 1/n$, and the ``unlikely to trade'' items. We can bound $\opts$ by the first-best GFT of the ``likely to trade'' items and the optimal profit of the super seller auction with the ``unlikely to trade'' items, and then use this to to decompose $\opts$ further, giving
\begin{align*}
    \opt&\leq \optb+\opts(L,\cF\big|_L)+\fb(H,\cF\big|_H) & \text{(Lemma~\ref{lem:likely trade and unlikely trade})}\\ 
    &\leq \optb+ \opts(L,\cF\big|_L)+O\left(\frac{\log n}{\delta\cdot \eta}\right)\cdot \gcfpp. & \text{(Theorem~\ref{thm:log1/r})}
\end{align*}


Of course, we can  Theorem~\ref{thm:downward-close} to instead use an $O\left(\log^2(n)\right)$-factor for a general constrained-additive buyer. 
It is well known that in multi-item auctions, the revenue of selling the items separately is a $O(\log n)$-approximation to the optimal revenue when there is a single additive buyer~\cite{LiY13}. Cai and Zhao~\cite{CaiZ19} provide an extension of this $O(\log n)$-approximation to profit maximization.  We build on this in Section~\ref{sec:UB of unlikely to trade items} to upper bound the $\opts(L,\cF\big|_L)$ term, where with $|L|$ items, we get a $\log(|L|)$ factor (Lemma~\ref{lem:unlikley trade item upper bound}).

All together, this gives the following upper bound on the second-best GFT.

\begin{lemma}[Upper Bound on Second-Best GFT]\label{lem:UB of second best GFT}
Define $H=\{i\in [n]:r_i\geq \frac{1}{n}\}$ and $L=[n]\backslash H=\{i\in [n]:r_i<\frac{1}{n}\}$. Suppose the buyer's feasibility constraint $\cF$ is 
$(\delta,\eta)$-selectable for some $\delta,\eta\in (0,1)$. Then 
\begin{align*}
    \opt \leq \optb+ O\left(\log(|L|)\cdot \sum_{i\in L}\E_{b_i,s_i}\left[(\tilde{\varphi}_i(b_i)-s_i)^+\right]\right)+O\left(\frac{\log n}{\delta\cdot \eta}\right)\cdot \gcfpp.
\end{align*}
For a general constrained-additive buyer,
the $O\left(\frac{\log n}{\delta\cdot \eta}\right)$ factor above becomes $O\left(\log^2(n)\right)$.
\end{lemma}

Next, Section~\ref{sec:bounding super buyer} gives details on constructing a mechanism for a two-sided market whose GFT is at least $\optb$.  In Section~\ref{sec:SAPP}, we show how to use a generalization of posted price mechanisms to approximate the second term in the upper bound by the GFT of the 
Seller Adjusted Posted Price mechanism. The approximation heavily relies on the fact that in expectation, only one item can trade, so it is crucial that $L$ only contains the ``unlikely to trade'' items.


\subsection{Bounding the Optimal Buyer Utility in the Super Buyer Procurement Auction}\label{sec:bounding super buyer}
In this section, we construct a two-sided market to bound $\optb$ for any constrained additive buyer. 




\begin{lemma}\label{lem:mechanism for OPT-B}
Consider the mechanism $\MM^*=(x,p^B,p^S)$ where for every item $i$, buyer profile $\Bb$, and seller profile $\Bs$, $$x_i(\Bb,\Bs)=\ind[b_i-\tilde{\tau}_i(s_i)\geq 0\wedge i\in \argmax_{S\in \cF}\sum_{i\in S}(b_i-\widetilde{\tau}_i(s_i))^+].$$ 
Here $\widetilde{\tau}_i(s_i)$ is Myerson's ironed virtual value function\footnote{The seller's unironed virtual value function is $\tau_i(s_i) = s_i + \frac{G_i(s_i)}{g_i(s_i)}$.} for seller $i$'s distribution $\DD_i^S$. For every seller $i$, since $\tilde{\tau}_i(s_i)$ is non-decreasing in $s_i$, $x_i(\Bb,\Bs)$ is non-increasing in $s_i$. Define $p^S_i(\Bb,\Bs)$ as the threshold payment for seller $i$, i.e., the largest cost $s_i$ such that $x_i(\Bb,s_i,s_{-i})=1$. Define the buyer's payment $p^B(\Bb,\Bs)=\sum_i x_i(\Bb,\Bs)\cdot \tilde{\tau}_i(s_i)$. $\MM^*$ is DSIC, ex-post IR, ex-ante SBB~\footnote{One can make the mechanism IR and ex-post SBB by defining $p^B(\Bb,\Bs)=\sum_i p_i^S(\Bb,\Bs)$. The mechanism is still DSIC for all sellers. It is only BIC for the buyer, as the sellers' gains only equal to the virtual welfare when taking expectation over sellers' profile.} and  
$$\gft(\MM^*)\geq \optb=\E_{\Bb,\Bs}[\max_{S\in \cF}\sum_{i\in S}(b_i-\widetilde{\tau}_j(s_j))^+].$$ 
\end{lemma}
\begin{proof}

Since the seller's allocation rule is monotone and we use the threshold payment, $\MM^*$ is DSIC and ex-post IR for each seller. 

Note that for any seller profile $\Bs$, when the buyer has true type $\Bb$, her expected utility from reporting $\Bb'$ is $\sum_i x_i(\Bb',\Bs)\cdot (b_i-\widetilde{\tau}_i(s_i))$.
According to the definition of $x$, the buyer's utility is maximized when $\Bb'=\Bb$. Hence, $\MM$ is DSIC for the buyer. Moreover we have ex-post IR, as the buyer's expected utility when reporting truthfully is  $\max_{S\in \cF}\sum_{i\in S}(b_i-\widetilde{\tau}_i(s_i))^+\geq 0$.

It only remains to prove that the mechanism is ex-ante SBB and to lower bound its GFT. By Myerson's lemma\footnote{This lemma is used several times, and is formally stated as Lemma~\ref{lem:mye} in Appendix~\ref{apx:md}.} (Lemma~\ref{lem:mye}), for every $\Bb$ we have
$$\E_{\Bs}\left[\sum_i p_i^S(\Bb,\Bs)\right]=\E_{\Bs}\left[\sum_i x_i(\Bb,\Bs)\cdot\widetilde{\tau}_i(s_i)\right]=\E_{\Bs}[p^B(\Bb,\Bs)].$$

Thus the mechanism is ex-ante SBB.

Why is $\optb = \E_{\Bb,\Bs}[\max_{S\in \cF}\sum_{i\in S}(b_i-\widetilde{\tau}_j(s_j))^+]$? Notice that only the sellers are strategic in a super buyer procurement auction, and their types are all single-dimensional. One can apply the standard Myersonian analysis to the super buyer procurement auction and show that the optimal buyer utility is exactly $\E_{\Bb,\Bs}[\max_{S\in \cF}\sum_{i\in S}(b_i-\widetilde{\tau}_j(s_j))^+]$.

Note that the buyer's expected utility in $\MM^*$ is exactly $\optb$. As $\MM^*$ is an ex-ante SBB mechanism, the expected GFT of $\MM^*$ is equal to the buyer's expected utility plus the sum of all sellers' expected utility, and the latter is non-negative since $\MM^*$ is ex-post IR for every seller.
\end{proof}

\subsection{The Seller Adjusted Posted Price Mechanism}\label{sec:SAPP}
In this section, we introduce a new mechanism---the \emph{Seller Adjusted Posted Price} (\Msapp) Mechanism. We define an adjusted price mechanism to first elicit each seller's cost $s_i$, and then produce posted prices $\{\theta_i(\Bs)\}_{i\in [n]}$ as a function of the reported profile $\Bs$; thus the mechanism is a collection of posted prices depending on the reported seller cost profile. The items are offered to the buyer at each posted price $\theta_i(\Bs)$, with the buyer only allowed to purchase at most one item by paying the posted price. 
See Mechanism~\ref{alg:sapp} for a complete description of the SAPP mechanism. We show that for a properly selected mapping $\{\theta_i(\cdot)\}_{i\in[n]}$, the \Msapp~mechanism is DSIC, ex-post IR, and ex-ante WBB. Moreover, its GFT is at least  $\Theta\left( \sum_{i\in L}\E_{b_i,s_i}\left[(\tilde{\varphi}_i(b_i)-s_i)^+\right]\right)$.

Since the posted prices depend on the reported seller cost profile, we need to be careful to ensure that there is no incentive for any seller to misreport the cost. We identify a sufficient condition for the posted prices, called \emph{bi-monotonicity}, to make sure the corresponding mechanism is DSIC and ex-post IR.

\begin{definition}[Bi-monotonic Prices]\label{def:bi-monotonicity}
We say the posted prices $\{\theta_i(\Bs)\}_{i\in [n]}$ are bi-monotonic, if (i) $\theta_i(\Bs)\geq s_i$ for all seller profile $\Bs$ and seller $i$; (ii) $\theta_i(\Bs)$ is \textbf{non-decreasing in $s_i$} and \textbf{non-increasing in $s_j$ for all $j\neq i$}.
\end{definition}

In Lemma~\ref{sapp:monotone}, we prove that bi-monotonic posted prices induce a monotone allocation rule for every seller, enabling \emph{threshold payments}~\cite{Myerson81,MyersonS83}. Formally, for every seller $i$ let $\hat{x}_i(\Bb,\Bs)$ denote the probability that the buyer trades with seller $i$ under profile $(\Bb,\Bs)$. This is either 0 or 1 since all $\theta_i(\Bs)$s are fixed values when $\Bs$ is fixed. If $\hat{x}_i(\Bb,\Bs)=1$, $p_i^S(\Bb,\Bs)$ is defined as the maximum value $s_i'$ such that $\hat{x}_i(\Bb,s_i',s_{-i})=1$. Otherwise $p_i^S(\Bb,\Bs)=0$. This makes the \Msapp~mechanism DSIC and ex-post IR.

\begin{lemma}\label{sapp:monotone}
Let $\MM$ be a \Msapp~mechanism with bi-monotonic posted prices $\{\theta_i(\Bs)\}_{i\in[n]}$. Then the allocation of the mechanism $\hat{x}_i(\Bb,\Bs)$ is non-increasing in $s_i$ for all sellers $i$, and $\MM$ is DSIC and ex-post IR for the buyer and the sellers.
\end{lemma}

\begin{proof}
Notice that for every type $\Bb$, the buyer chooses the item that maximizes $b_i-\theta_i(\Bs)$ (and does not choose any item if she cannot afford any of the items). For every $i$, by bi-monotonicity, when $s_i$ decreases, $b_i-\theta_i(\Bs)$ does not decrease while $b_j-\theta_j(\Bs)$ does not increase for all $j\not=i$. Thus if the buyer chooses item $i$ under the original $s_i$, she must continue to choose item $i$ for smaller reports $s_i'$. Thus $\hat{x}_i(\Bb,\Bs)$ is non-increasing in $s_i$. 
Since every seller receives the threshold payment, she is DSIC and ex-post IR. As the buyer simply faces a posted price mechanism, the mechanism is DSIC and ex-post IR for the buyer.
\end{proof}


\floatname{algorithm}{Mechanism}\label{mech:SAPP}
\begin{algorithm}[ht]
\begin{algorithmic}[1]
\REQUIRE $\forall i\in [n]$, function $\theta_i(\cdot)$ that maps each seller cost profile to a price for item $i$. Input $(\Bb, \Bs)$.
\STATE Given the sellers' reported cost profile $\Bs$, offer each item $i$ to the buyer at price $\theta_i(\Bs)$.
\STATE The buyer is allowed to purchase at most one item by paying the corresponding posted price.
\STATE If no item is picked, then no trade happens and payment is 0 for every agent. Otherwise, if the buyer chooses item $i$, she receives item $i$ and pays $\theta_i(\Bs)$. Seller $i$ sells her item and receives threshold payment.
\end{algorithmic}
\caption{{\sf Seller Adjusted Posted Price Mechanism}}
\label{alg:sapp}
\end{algorithm}


The main challenge we face here is establishing the budget balance condition. Unfortunately, having bi-monotonic posted prices is not sufficient. Consider the $n=1$ case: the posted price $p(s)=s$ is trivially bi-monotonic. Clearly, the corresponding \Msapp~mechanism achieves \fb. However, due to the impossibility result by Myerson and Satterthwaite~\cite{MyersonS83}, no BIC, IR, and ex-ante WBB mechanism can always achieve \fb, so the \Msapp~mechanism must sometimes violate the budget balance constraint. In Lemma~\ref{sapp:main}, we show that even though bi-monotonic posted prices do not imply budget balance, there is indeed a wide range of bi-monotonic posted prices that induce budget balanced \Msapp~mechanisms. Our budget balance proof heavily relies on an allocation coupling argument (Lemma~\ref{lem:sapp-main1}) that simultaneously provides a lower bound on the buyer's payment, as well as an upper bound on the payment to the seller.


\begin{lemma}\label{sapp:main}
Let $x=\{x_i(\Bb,\Bs)\}_{i\in [n]}$ be an arbitrary allocation rule that satisfies (i) the buyer never purchases more than one item in expectation under each profile $(\Bb,\Bs)$, i.e., $\sum_{i\in[n]} x_i(\Bb,\Bs)\leq 1$, and (ii) for every buyer type $\Bb$ and seller $i$, $x_i(\Bb,\Bs)$ is non-increasing in $s_i$, and non-decreasing in $s_j$ for all $j\not=i$. We define $q_i(\Bs)=\E_{\Bb}[x_i(\Bb,\Bs)\cdot \ind[\tilde{\varphi}_i(b_i)\geq s_i]]$, where $\tilde{\varphi}_i(b_i)$ is Myerson's ironed virtual value for $\DD_i^B$, and $\theta_i(\Bs)=F_i^{-1}(1-\frac{q_i(\Bs)}{2})$. The posted prices $\{\theta_i(\Bs)\}_{i\in [n]}$ are bi-monotonic, and the corresponding SAPP mechanism $\MM$ is DSIC, ex-post IR, and ex-ante WBB. Moreover, $\gft(\MM)\geq \frac{1}{4}\E_{\Bb,\Bs}\left[\sum_i (\widetilde{\varphi}_i(b_i)-s_i)\cdot x_i(\Bb,\Bs)\right]$.
\end{lemma}




\begin{proof}
It is not hard to verify that  $\{\theta_i(\Bs)=F_i^{-1}(1-\frac{q_i(\Bs)}{2})\}_{i\in[n]}$ is bi-monotonic. 
Now we proceed to prove that the \Msapp ~mechanism $\MM$ is ex-ante WBB.  We require the following lemma.

\begin{lemma}\label{lem:sapp-main1}
For every seller $i$ and every seller profile $\Bs$, $\hat{x}_i(\Bs)\in \left[\frac{q_i(\Bs)+q_i(\Bs)^2}{4},\frac{q_i(\Bs)}{2}\right]$. 
\end{lemma}

\begin{proof}
Note that the buyer will purchase item $i$ if both of the following conditions are satisfied: 
\begin{enumerate}
    \item The buyer can afford item $i$, i.e., $b_i\geq \theta_i(\Bs)$.
    \item The buyer can not afford any other items, i.e., $b_j<\theta_j(\Bs),\forall j\not=i$.
\end{enumerate}

By choice of $\theta_i(\Bs)$, the first event happens with probability $\Pr[b_i\geq \theta_i(\Bs)]=q_i(\Bs)/2$. 

Note that $\sum_{i\in[n]} q_i(\Bs)\leq \E_{\Bb}[\sum_{i\in [n]} x_i(\Bb,\Bs)]\leq 1$. For each $j\neq i$, $\Pr[b_j<\theta_j(\Bs)]=1-\frac{q_j(\Bs)}{2}$. Thus $\sum_{j\neq i} \left(1-\frac{q_j(\Bs)}{2}\right)\geq n-\frac{3}{2}+\frac{q_i(\Bs)}{2}$.
The second event happens with probability $$\prod_{j\neq i} \left(1-\frac{q_j(\Bs)}{2}\right)\geq \frac{1}{2}+\frac{q_i(\Bs)}{2}.$$ The equality holds when one out of the $n-1$ $q_j(\Bs)$'s equals  
$1-q_i(\Bs)$ and the rest are all equal to $0$.
Notice that the two events are independent, so we have the upper and lower bound on $\hat{x}_i(\Bs)$.
\end{proof}

We return to the proof of Lemma~\ref{sapp:main}. For easy reference, we list our notation again:
\begin{itemize}
    \item $x=\{x_i(\Bb,\Bs)\}_{i\in [n]}$ is an arbitrary allocation.
    \item 
    $\hat{x}_i(\Bb,\Bs)$ is the probability that item $i$ trades in $\MM$ under profile $(\Bb,\Bs)$.
    \item $\hat{x}_i(\Bs) = \E_\Bb[\hat{x}_i(\Bb,\Bs)]$ is the probability that item $i$ trades over the randomness of buyer valuations, i.e. the interim trade probability.
    \item $q_i(\Bs)=\E_{\Bb}[x_i(\Bb,\Bs)\cdot \ind[\tilde{\varphi}_i(b_i)\geq s_i]]$ is the probability that item $i$ trades in allocation $x$ \emph{and} the buyer's ironed virtual value for item $i$ is above the seller's cost.
    \item $\theta_i(\Bs)=F_i^{-1}(1-\frac{q_i(\Bs)}{2})$ is the buyer's posted price set such that $\Pr[b_i\geq \theta_i(\Bs)]=q_i(\Bs)/2$.
\end{itemize}
 Fix any seller profile $\Bs$. For simplicity, we slightly abuse notation and use $\hat{x}_i(z)$ and $q_i(z)$ to denote $\hat{x}_i(z,s_{-i})$ and $q_i(z,s_{-i})$. The expected payment from the buyer under cost profile $\Bs$ is $\sum_{i\in[n]} \hat{x}_i(s_i)\cdot \theta_i(\Bs)$. For every seller $i$, denote $p_i^S(\Bs)=\E_{\Bb}[p_i^S(\Bb,\Bs)]$ her expected payment under cost profile $\Bs$. 

Note that for every $\Bb,\Bs$, the threshold payment $p_i^S(\Bb,\Bs)$ can be rewritten as the quantity $\int_{s_i}^\infty \hat{x}_i(\Bb,t,s_{-i})dt+s_i\cdot \hat{x}_i(\Bb,s_i,s_{-i})$: When $\hat{x}_i(\Bb,\Bs)=0$, then $\hat{x}_i(\Bb,t,s_{-i})$ for all $t\geq s_i$ since $\hat{x}_i(\Bb,\Bs)$ is non-increasing in $s_i$. Thus the above quantity is 0. When $\hat{x}_i(\Bb,\Bs)=1$, let $s_i'$ be the maximum value such that $\hat{x}_i(\Bb,s_i',s_i)=1$. Then the above quantity is equal to $\int_{s_i}^{s_i'}1dt+s_i=s_i'=p_i^S(\Bb,\Bs)$. Thus 

$$p_i^S(\Bs)=\E_{\Bb}[p_i^S(\Bb,\Bs)]=\int_{s_i}^\infty \hat{x}_i(z,s_{-i})dz+s_i\cdot \hat{x}_i(s_i,s_{-i}).$$
We will show that $p_i^S(\Bs)\leq \hat{x}_i(s_i)\cdot \theta_i(\Bs)$. By definition,
\begin{align*}
    p_i^S(\Bs)&=\int_{s_i}^\infty \hat{x}_i(z)dz+s_i\cdot \hat{x}_i(s_i)\\
    &=\int_{s_i}^\infty\int_{0}^\infty \ind[\hat{x}_i(z)\geq t]dtdz+s_i\cdot \hat{x}_i(s_i) & \text{(~$\hat{x}_i(z)=\int_{0}^\infty \ind[\hat{x}_i(z)\geq t]dt,\forall z$~)}\\
    &=\int_{s_i}^\infty\int_{0}^{\hat{x}_i(s_i)} \ind[\hat{x}_i(z)\geq t]dtdz+s_i\cdot \hat{x}_i(s_i) & \quad \text{(~$\ind[\hat{x}_i(z)\geq t]=0,\forall z\geq s_i {\textrm{ and }} t>\hat{x}_i(s_i)$)}\\
    &=\int_{0}^{\hat{x}_i(s_i)}\int_{s_i}^\infty\ind[\hat{x}_i(z)\geq t]dzdt+s_i\cdot \hat{x}_i(s_i).
\end{align*}

The last equality follows from {Fubini's Theorem,  as the integral is finite due to the monotonicity of $\hat{x}_i(\cdot)$.}

Moreover, since $\hat{x}_i(\cdot)$ is non-increasing, for every $z\leq s_i,t\leq \hat{x}_i(s_i)$, we have $\hat{x}_i(z)\geq \hat{x}_i(s_i)\geq t$. Thus
$$\int_{0}^{\hat{x}_i(s_i)}\int_{0}^{s_i} \ind[\hat{x}_i(z)\geq t]dzdt=\int_{0}^{\hat{x}_i(s_i)}\int_{0}^{s_i} 1dzdt=s_i\cdot \hat{x}_i(s_i).$$

Combining the two equations, we have
\begin{align*}
p_i^S(\Bs) &=\int_{0}^{\hat{x}_i(s_i)}\int_{0}^\infty\ind[\hat{x}_i(z)\geq t]dzdt &\\ &\leq \int_{0}^{\hat{x}_i(s_i)}\int_{0}^\infty\ind[q_i(z)\geq 2t]dzdt   & \text{(Lemma~\ref{lem:sapp-main1})}\\
&\leq\int_{0}^{\hat{x}_i(s_i)}\int_{0}^\infty\ind\left[\Pr_{b_i}[\widetilde{\varphi}_i(b_i)\geq z]\geq 2t\right]dzdt   & \text{(Definition of $q_i(\cdot)$)}
\end{align*}

For every $t$, we prove that $\int_{0}^\infty\ind\left[\Pr_{b_i}[\widetilde{\varphi}_i(b_i)\geq z]\geq 2t\right]dz\leq \widetilde{\varphi}_i(F_i^{-1}(1-2t+\epsilon))$ for any $\epsilon>0$. In fact, let $z^*=\widetilde{\varphi}_i(F_i^{-1}(1-2t+\epsilon))$. For every $z>z^*$, $\Pr[\widetilde{\varphi}_i(b_i)\geq z]\leq \Pr[\widetilde{\varphi}_i(b_i)>z^*]=\Pr[b_i>F_i^{-1}(1-2t+\epsilon)]\leq 2t-\epsilon$. So $\ind\left[\Pr[\widetilde{\varphi}_i(b_i)\geq z]\geq 2t \right]=0$ for all $z>z^*$.

Therefore, for any $\epsilon>0$, we have the following.  We will change variables.
\begin{align*}
p_i^S(\Bs)&\leq \int_{0}^{\hat{x}_i(s_i)}\widetilde{\varphi}_i(F_i^{-1}(1-2t+\epsilon))dt\\
&=\int_{\infty}^{F_i^{-1}(1-2\hat{x}_i(s_i)+\epsilon)}\widetilde{\varphi}_i(y)d\frac{1+\epsilon-F_i(y)}{2}   & \text{($y=F_i^{-1}(1-2t+\epsilon)$)}\\
&=-\frac{1}{2}\int_{\infty}^{F_i^{-1}(1-2\hat{x}_i(s_i)+\epsilon)}\widetilde{\varphi}_i(y)f_i(y)dy\\
&=\frac{1}{2}\int_{F_i^{-1}(1-2\hat{x}_i(s_i)+\epsilon)}^{\infty}\widetilde{\varphi}_i(y)f_i(y)dy\\
&=\frac{1}{2}F_i^{-1}(1-2\hat{x}_i(s_i)+\epsilon)\cdot [1 - F_i(F_i^{-1}(1-2\hat{x}_i(s_i)+\epsilon))]   & \text{(Myerson's Lemma (\ref{lem:mye}))}\\
&= F_i^{-1}(1-2\hat{x}_i(s_i)+\epsilon) \cdot  (\hat{x}_i(s_i)-\epsilon/2) \\
&\leq\hat{x}_i(s_i)\cdot F_i^{-1}(1-2\hat{x}_i(s_i)+\epsilon)
\end{align*}

If $q_i(s_i)=0$, then $\hat{x}_i(s_i)\cdot F_i^{-1}(1-2\hat{x}_i(s_i)+\epsilon)=0=\hat{x}_i(s_i)\cdot \theta_i(\Bs)$. 
Otherwise, choose $\epsilon$ to be any number in $(0,\frac{q_i(s_i)^2}{4})$. Then, according to Lemma~\ref{lem:sapp-main1} and our choice of $\epsilon$, $$1-2\hat{x}_i(s_i)+\epsilon\leq 1-\frac{q_i(s_i)}{2}-\frac{q_i(s_i)^2}{4}< 1-\frac{q_i(s_i)}{2}.$$
Hence, $F_i^{-1}(1-2\hat{x}_i(s_i)+\epsilon)< \theta_i(\Bs)$. 
Thus $p_i^S(\Bs)\leq \hat{x}_i(\Bs)\cdot \theta_i(\Bs)$ for every $i$ and $\Bs$, which implies that $\E_{\Bs}\left[\sum_i\theta_i(\Bs)\cdot \hat{x}_i(\Bs)\right]\geq \E_{\Bs}\left[\sum_i p_i^S(s_i,s_{-i})\right]$. Hence $\MM$ is ex-ante WBB. 

\vspace{.2in}

We now need to lower bound the GFT from mechanism $\MM$.

\begin{align*}
    \gft(\MM)=&\E_{\Bb,\Bs}\left[\sum_i(b_i-s_i)\cdot \hat{x}_i(\Bb,\Bs)\right]\\
    \geq&\E_{\Bs}\left[\sum_i(\theta_i(\Bs)-s_i)\cdot \hat{x}_i(\Bs)\right]  & \text{($\hat{x}_i(\Bb,\Bs)=0$ if $b_i<\theta_i(\Bs)$)}\\
    \geq &\frac{1}{2}\E_{\Bs}\left[\sum_i\left(F_i^{-1}\left(1-\frac{q_i(\Bs)}{2}\right)-s_i\right)\cdot \frac{q_i(\Bs)}{2}\right]   & \text{(Definition of  $\theta_i(\Bs)$, $q_i(\Bs)$ and Lemma~\ref{lem:sapp-main1})}\\
    =&\frac{1}{2}\E_{\Bb,\Bs}\left[\sum_i(\widetilde{\varphi}_i(b_i)-s_i)\cdot\ind\left[b_i\geq F_i^{-1}\left(1-\frac{q_i(\Bs)}{2}\right)\right]\right]  & \text{(Myerson's Lemma (\ref{lem:mye}))}\\
    \geq& \frac{1}{4}\E_{\Bb,\Bs}\left[\sum_i(\widetilde{\varphi}_i(b_i)-s_i)\cdot x_i(\Bb,\Bs)\cdot \ind[\tilde{\varphi}_i(b_i)\geq s_i]\right]\\
    \geq& \frac{1}{4}\E_{\Bb,\Bs}\left[\sum_i(\widetilde{\varphi}_i(b_i)-s_i)\cdot x_i(\Bb,\Bs)]\right]   
\end{align*}

Here the second-to-last inequality uses the fact that $$\E_{b_i}\left[\widetilde{\varphi}_i(b_i)\cdot\ind[b_i\geq F_i^{-1}\left(1-\frac{q_i(\Bs)}{2}\right)\right]\geq \frac{1}{2}\cdot\E_{\Bb}\left[\widetilde{\varphi}_i(b_i)\cdot x_i(\Bb,\Bs)\cdot \ind[\tilde{\varphi}_i(b_i)\geq s_i]\right]$$
holds for every $\Bs$ and $i$. This is because the right hand side $$\frac{1}{2}\cdot\E_{\Bb}[\widetilde{\varphi}_i(b_i)\cdot x_i(\Bb,\Bs) \cdot \ind[\tilde{\varphi}_i(b_i)\geq s_i]]=\E_{b_i}\left[\widetilde{\varphi}_i(b_i)\cdot \frac{1}{2}\E_{b_{-i}}[x_i(\Bb,\Bs) \cdot \ind[\tilde{\varphi}_i(b_i)\geq s_i]]\right]$$ can be viewed as the expectation of $\widetilde{\varphi}_i(b_i)$ on an event of $b_i$ with a total probability mass $$\E_{b_i}\left[\frac{1}{2}\E_{b_{-i}}[x_i(\Bb,\Bs) \cdot \ind\left[\tilde{\varphi}_i(b_i)\geq s_i]\right]\right]=\frac{q_i(\Bs)}{2},$$ while the left hand side is the maximum expectation of $\widetilde{\varphi}_i(b_i)$ on any event of $b_i$ with  total probability mass $\frac{q_i(\Bs)}{2}$, as $\widetilde{\varphi}_i(b_i)$ is non-decreasing on $b_i$. 
\end{proof}

Lemma~\ref{lem:SAPP bounding unlikely trade items} shows how to choose an allocation rule $x$ so that the induced \Msapp~mechanism (using Lemma~\ref{sapp:main}) has GFT at least  $\Omega\left(\sum_{i\in L}\E_{b_i,s_i}\left[(\tilde{\varphi}_i(b_i)-s_i)^+\right]\right)$. Note that the existence of such an $x$ heavily relies on the fact that in expectation there is only one  item that can trade among the ``unlikely to trade'' items. 



\begin{lemma}\label{lem:SAPP bounding unlikely trade items}
We let  $\sapp(S)$ denote the optimal GFT attainable by any DSIC, ex-post IR, and ex-ante WBB \Msapp~mechanisms over items in $S$ for any subset $S\subseteq [n]$. $\sapp(L)\geq \frac{1}{4e}\cdot \sum_{i\in L}\E_{b_i,s_i}[(\widetilde{\varphi}_i(b_i)-s_i)^+]$.
\end{lemma}

\begin{proof}
Let $\Bb_L=\{b_i\}_{i\in L}$ and $\Bs_L=\{s_i\}_{i\in L}$. For every $i\in L$, define the event that only $i$ is tradeable: $$A_i=\left\{(\Bb_L,\Bs_L):b_i\geq s_i\wedge b_j<s_j,\forall j\in L\backslash\{i\}\right\}.$$ We consider the following allocation rule:
\begin{align*}
x_i(\Bb_L,\Bs_L)=
\begin{cases}
\ind[\widetilde{\varphi}_i(b_i)\geq s_i] &,~ \text{if }  (\Bb,\Bs)\in A_i\\
0 &,~\text{otherwise}
\end{cases}
\end{align*}

Notice that $(\Bb_L,\Bs_L)\in A_i$ implies that $(\Bb_L,s_i',\Bs_{L\backslash\{i\}})\in A_i$ for any $s_i'\leq s_i$. Thus, $x_i(\Bb_L,\Bs_L)$ is non-increasing in $s_i$. Similarly, it is easy to verify that $x_i(\Bb_L,\Bs_L)$ is non-decreasing in all $s_j$ where $j\in L\backslash\{i\}$. Furthermore, $\sum_{i\in L} x_i(\Bb_L,\Bs_L)\leq 1$ for all $\Bb_L,\Bs_L$. If we choose the posted prices according to Lemma~\ref{sapp:main}, then the corresponding mechanism has GFT that is at least $\frac{1}{4}\E_{\Bb,\Bs}\left[\sum_i (\widetilde{\varphi}_i(b_i)-s_i)\cdot x_i(\Bb,\Bs)\right]$.

Moreover, by the definition of $x_i(\Bb,\Bs)$, 
\begin{align*}
\E_{\Bb,\Bs}\left[\sum_{i\in L} (\widetilde{\varphi}_i(b_i)-s_i)\cdot x_i(\Bb,\Bs)\right]=&\sum_{i\in L} \E_{b_i,s_i}[(\widetilde{\varphi}_i(b_i)-s_i)^+)]\cdot \prod_{j\in L\backslash\{i\}}\Pr_{b_j,s_j}[b_j<s_j]\\
\geq&\sum_{i\in L} \E_{b_i,s_i}[(\widetilde{\varphi}_i(b_i)-s_i)^+)]\cdot (1-\frac{1}{n})^{|L|}\\
\geq& \frac{1}{e}\cdot\sum_{i\in L}\E_{b_i,s_i}[(\widetilde{\varphi}_i(b_i)-s_i)^+)]
\end{align*}
The first inequality holds because for each item $j\in L$, $\Pr_{b_j,s_j}[b_j<s_j]\geq 1-1/n$. Hence, $$\sapp(L)\geq \frac{1}{4e}\cdot \sum_{i\in L}\E_{b_i,s_i}[(\widetilde{\varphi}_i(b_i)-s_i)^+].$$
\end{proof}

\subsection{Bounding the Optimal Profit from the Unlikely to Trade Items}\label{sec:UB of unlikely to trade items}

In this section, 
we provide an upper bound for the optimal super seller profit from items in $L$. It is well known that in multi-item auctions the revenue of selling the items separately is a $O(\log n)$-approximation to the optimal revenue when there is a single additive buyer~\cite{LiY13}. Cai and Zhao~\cite{CaiZ19} provide a extension of this $O(\log n)$-approximation to profit maximization. Combining this approximation with some basic observations based on the Cai-Devanur-Weinberg duality framework~\cite{CaiDW16}, we derive the following upper bound of $\opts(L,\cF\big|_L)$.
\begin{lemma}\label{lem:unlikley trade item upper bound}
$$\opts(L,\cF\big|_L)\leq O\left(\log(|L|)\cdot \sum_{i\in L}\E_{b_i,s_i}\left[(\tilde{\varphi}_i(b_i)-s_i)^+\right]\right).$$
Here $\widetilde{\varphi}_i(b_i)$ is Myerson's ironed virtual value function\footnote{The buyer's unironed virtual value function is $\varphi_i(b_i) = b_i - \frac{1-F_i(b_i)}{f_i(b_i)}$. These values are averaged to be made monotonic in the quantile space, which creates $\widetilde{\varphi}_i(b_i)$.}  for the buyer's distribution for item $i$, $\DD_i^B$.
\end{lemma}


To bound $\opts(L,\cF\big|_L)$, we need the following result from \cite{CaiZ19}. It provides a benchmark of the optimal profit using the Cai-Devanur-Weinberg duality framework~\cite{CaiDW16}: The profit of any BIC, IR mechanism is upper bounded by the buyer's virtual welfare with respect to some virtual value function, minus the sellers' total cost for the same allocation.

A sketch of the framework is as follows: We first formulate the profit maximization problem as an LP. Then we Lagrangify the BIC and IR constraints to get a partial Lagrangian dual of the LP. Since the buyer’s payment is unconstrained in the partial Lagrangian, one can argue that the corresponding dual variables must form a flow to ensure that the benchmark is finite. By weak duality, any choice of the dual variables/flow derives a benchmark for the optimal profit. In \cite{CaiZ19}, they also construct a canonical flow and prove that there exists a BIC and IR mechanism whose profit is within a constant factor times the benchmark w.r.t. the flow for any single constrained-additive buyer.     
\begin{lemma}\label{lem:mz19}
\cite{CaiZ19} For any $T\subseteq [n]$ and feasibility constraint $\cJ$ with respect to $T$, consider the super seller auction with item set $T$ and to $\cJ$-constrained buyer. Any flow $\lambda_T$ induces a finite benchmark for the optimal profit, that is,

$$\opts(T,\cJ)\leq \max_{x\in P_{\cJ}}\E_{\Bb,\Bs}\left[\sum_{i\in T} x_i(\Bb,\Bs)\cdot (\Phi_i^T(\Bb)-s_i)\right]$$

where

$$\Phi_i^{T}(\Bb)=b_i-\frac{1}{f_i(b_i)}\sum_{\Bb'}\lambda_T(\Bb',\Bb)\cdot (b_i'-b_i)$$
can be viewed as buyer $i$'s virtual value function, and $P_{\cJ}$ is the set of all feasible allocation rules. More specifically, $\lambda_T(\Bb',\Bb)$ is the Lagrangian multiplier for the BIC/IR constraint that says when the buyer has true type $\Bb$ she does not want to misreport $\Bb'$. The equality sign is achieved when the optimal dual $\lambda_T^*$ is chosen.


\end{lemma}

\vspace{.2in}

Next, we show that $\opts(L,\cF\big|_L)$ is no more than $\opts(L,\text{ADD})$ using Lemma~\ref{lem:mz19}.

\begin{lemma}\label{lem:compare-to-additive}
$\opts(L,\cF\big|_L)\leq \opts(L,\text{ADD}).$
\end{lemma}
\begin{proof}
Let $\hat{\lambda}_L$ be the optimal dual in Lemma~\ref{lem:mz19} when the buyer is additive without any feasibility constraint, and $\hat{\Phi}_i^L(\cdot)$ be the induced virtual value function. We have that 
\begin{align*}
\opts(L,\cF\big|_L)&\leq \max_{x\in P_{\cF|_L}}\E_{\Bb,\Bs}\left[\sum_{i\in L}x_i(\Bb,\Bs)\cdot (\hat{\Phi}_i^L(\Bb)-s_i)\right]\\
&\leq \E_{\Bb,\Bs}\left[\sum_{i \in L}(\hat{\Phi}_i^L(\Bb)-s_i)^+\right]\\
&=\max_{x_i(\Bb,\Bs)\in [0,1]}\E_{\Bb,\Bs}\left[\sum_ix_i(\Bb,\Bs)\cdot (\hat{\Phi}_i^L(\Bb)-s_i)\right]\\
&=\opts(L,\text{ADD}). 
\end{align*}
\end{proof}




Cai and Zhao~\cite{CaiZ19} also give a logarithmic upper bound of the optimal profit for a single additive buyer, using the sum of optimal profit for each individual item.
\begin{lemma}\label{lem:mz19-additive}
\cite{CaiZ19} $$\opts(L,\text{ADD})\leq\log(|L|)\cdot \sum_{i\in L}\opts(\{i\})= \log(|L|)\cdot \sum_{i\in L}\E_{b_i,s_i}[(\varphi_i(b_i)-s_i)^+].$$
\end{lemma}

Together, Lemmas~\ref{lem:compare-to-additive} and \ref{lem:mz19-additive} conclude the proof of Lemma~\ref{lem:unlikley trade item upper bound}:
$$\opts(L,\cF\big|_L)\leq O\left(\log(|L|)\cdot \sum_{i\in L}\E_{b_i,s_i}\left[(\tilde{\varphi}_i(b_i)-s_i)^+\right]\right).$$


\vspace{.4in}
\begin{prevproof}{Theorem}{thm:logn}
The theorem follows directly from Lemmas~\ref{lem:UB of second best GFT},~\ref{lem:mechanism for OPT-B}, ~\ref{sapp:main}, and~\ref{lem:SAPP bounding unlikely trade items}.
\end{prevproof}
\section{Lower Bounds and the First-Best--Second-Best Gap}\label{sec:reduction}

In the unconditional approximation results stated in Section~\ref{sec:log n approx}, we compare the GFT of our mechanism to $\opt$. Readers may be interested in whether our mechanism is also an approximation to $\fb$. In fact, this question is related to one of the major open problems in two-sided markets: \emph{How large is the gap between the second-best and the first-best GFT?} In this section, we consider a unit-demand buyer and present a reduction from achieving a $\fb$ approximation in our multi-dimensional setting to the open problem regarding the gap in single-dimensional two-sided markets.

\vspace{-.1in}
\paragraph{Matching Markets.} This setting has a two-sided market with $n$ buyers, $n$ sellers, and $n$ identical items. Each seller owns one item and each buyer is interested in buying at most one item. Thus the value (or cost) for every agent is a scalar. Here we consider a special case where for every $i\in [n]$, buyer $i$ and seller $i$ can only trade with each other, 
and at most one pair of agents in the market can trade. This is bilateral trade when $n=1$.

\begin{theorem}\label{thm:reduction}
Suppose the buyer is unit-demand in the multi-dimensional setting, and define $\fb$, $\optb$, $\sapp$ as in the previous section. 
Consider the following matching market with $n$ buyers and $n$ sellers: for every $i\in [n]$, buyer $i$ has value drawn from $\DD_i^B$ and seller $i$ has cost drawn from $\DD_i^S$. 
Let $\fbsd=\E_{\Bb,\Bs}[\max_i(b_i-s_i)]$ be the first-best GFT of the matching market defined above (which is the same as $\fb$ in the multi-dimensional unit-demand setting) and $\optsd$ be the second-best GFT. 
For any $c>1$, suppose $\optsd\geq 1/c\cdot\fbsd$, then
$$\max\{\optb,\sapp\}\geq \frac{1}{2c}\cdot \fb.$$
\end{theorem}

The proof of Theorem~\ref{thm:reduction} is straightforward and can be found in Appendix~\ref{appx:reduction}; it directly follows from Lemmas~\ref{lem:mechanism for OPT-B}, \ref{sapp:main}, and an upper bound of $\optsd$ by Brustle et al.~\cite{BrustleCWZ17}.  The main takeaway of Theorem~\ref{thm:reduction} is that, if the largest gap between $\fbsd$ and $\optsd$ is at most (i.e. a constant) $c$ for matching markets, then our mechanism is a $2c$-approximation to $\fb$. Note that if the buyer is additive, such a reduction clearly exists: In the additive case, items can be treated separately without impacting the IC constraint. Then performing a Buyer (or Seller) Offering mechanism\footnote{In bilateral trade, a Buyer Offering mechanism lets the buyer choose a take-it-or-leave-it price for the seller according to her value. And in the Seller Offering mechanism, the seller is asked to pick the price for the buyer.} for every item separately obtains GFT at least $\optsd$~\cite{BrustleCWZ17}, thus approximating $\fb$ by the assumption. 
Theorem~\ref{thm:reduction} shows that for a unit-demand buyer, a similar reduction also exists using the SAPP mechanism.

On the other hand, finding a lower bound for our result (compared to $\opt$) is at least as hard as finding a lower bound 
for the approximation ratio w.r.t.
$\fb$, and thus is \emph{at least as hard} as finding an instance in the matching market that separates $\fbsd$ from $\optsd$---a problem that has long remained open. Indeed,  even in bilateral trade, deciding whether the gap is finite or not is still open. 

\section*{Acknowledgements}
The authors would like to thank Anna Karlin for helpful discussions in the early stages of the paper.

\newpage 
\bibliographystyle{plain}
\bibliography{Yang.bib}
\newpage
\appendix

\section{Examples}\label{sec:lower bounds}
\notshow{\noindent\textbf{Example for Knapsack Constraints}

We have proved the approximation for intersection of matroids and matching constraints. Reader may wonder if the same class of mechanisms applies to more general downward-closed feasibility constraints. At the end of this section, we briefly discuss the performance of \ksppabbrev~mechanisms for general downward-closed $\cF$. As shown in Lemma~\ref{lem:kspp-intersect-matroid}, selectability and \propnoun~are the two crucial properties that guarantee the GFT of our \ksppabbrev~mechanism is comparable to the seller surplus term. Unfortunately, for general downward-closed constraint $\cF$, it remains open whether $\cF$ has $(\delta,\eta)$-selectability for constant $\delta$ and $\eta$. For other constraints proved to have selectability, i.e. knapsack constraints~\cite{FeldmanSZ16}, we show in Example~\ref{example-knapsack} that general knapsack constraints can be $\Theta(n)$-exchangeable. Moreover,  the second statement in Lemma~\ref{lem:kspp-intersect-matroid} can be violated by up to a factor of $\Theta(n)$ if we use  Mechanism~\ref{alg:kspp-matroid}. The evidence suggests that a new mechanism and different analytic techniques are needed to solve more general constraints.   

\begin{definition}[Knapsack Constraint]
A \emph{knapsack constraint} $\cF$ with respect to the ground set $I$ is defined as: $\cF=\{S\subseteq I:~\sum_{i\in S}c_i\leq 1\}$. Here $c_i\in [0,1]$ is the size of element $i$. The knapsack polytope $P_{\cF}=\{x\in [0,1]^{|I|}:~\sum_{i\in I}x_ic_i\leq 1\}$.
\end{definition}
\begin{example}\label{example-knapsack}
For any constant $C,\gamma>0$ such that $\gamma<C$, let $k=\lceil C/\gamma\rceil>1$. Consider the following knapsack constraint $\cF$ with respect to $[n]$. Let $A=[k]$ and $B=\{k+1,\ldots,n\}$. For every item $i\in A$, $c_i=1$; for every item $i\in B$, $c_i=\frac{1}{n}$. In other words, any feasible set in $\cF$ can either include one item from $A$, or any subset of $B$. Consider $S=\{1\}$ and $S'=B$. Let $H$ be the mapping $H:B\mapsto \{\{1\},\emptyset\}$ in Definition~\ref{def:exchangeable}. Since item 1 is not compatible to any item in $B$, $H(i)=\{1\}$ for any $i\in B$. Thus $\beta\geq |\{i\in B~\mid~1\in H(i)\}|=|B|=n-k$. Thus $\cF$ is $(n-k)$-\propname.

Now consider the packing procedure under the scenario in Table~\ref{table:example}.
\begin{table}[h]
    \centering
    \begin{tabular}{|c|c|c|c|}
    \hline
    Item $i$ & $c_i$ & $w_i$ & $v_i=\theta_i^S-s_i$\\
    \hline
     $i<k$  & 1 & $\gamma$ & $1+\epsilon$  \\
    \hline
    $i=k$   & 1 & $C-(k-1)\gamma$ & $1+\epsilon$\\
    \hline
    $i>k$   & $1/n$ & $\gamma$ & $1$\\
    \hline
    \end{tabular}
    \caption{The scenario for some reported sellers' proifile $\Bs$}
    \label{table:example}
\end{table}
Here $\epsilon>0$ is any positive constant. The packing procedure in Mechanism~\ref{alg:kspp-matroid} will pack all items in $A$. On the other hand, among all items that are not packed (set $B$), the maximum weight feasible set is also $B$. We have $\sum_{i\in A}w_i\cdot v_i=C(1+\epsilon)$ while $\sum_{i\in B}v_i=n-k$.

Note that $k$ is a constant. We have
$$\sum_{i\in B}v_i\geq \Omega(n)\cdot \sum_{i\in A}w_i\cdot v_i$$
\end{example}
}

\paragraph{Tight Example of the $\log\left(\frac{1}{r}\right)$-Approximation.}

Consider the case when $n=1$ (bilateral trade). We introduce an example provided by Blumrosen and  Dobzinski~\cite{BlumrosenD16}. They prove that in this example, no fixed posted price mechanisms can achieve an approximation ratio better than $\Omega(\log(1/r))$ compared to the first-best GFT. In addition, we will verify that the statement also holds for the second-best GFT for the same example. It implies that our $\log(\frac{1}{r})$-approximation is tight even compared to the second-best GFT.

\begin{example}[Example in Bilateral Trading~\cite{BlumrosenD16}]\label{example-log(1/r)-tight}
For any $t>0$, consider a buyer and a seller with values on the support $[0,t]$. Let $\lambda=\frac{1}{1-e^{-t}}$. Let $F(b)=\lambda(1-e^{-b})$ with $f(b)=\lambda e^{-b}$ and $G(s)=\lambda(e^{s-t}-e^{-t})$ with $g(s)=\lambda e^{s-t}$. Then 
$$r=\Pr[b\geq s]=\int_{0}^t\int_{0}^b\lambda e^{-b}\cdot \lambda e^{s-t}dsdb=\lambda^2\cdot e^{-t}(t-1+e^{-t})=\frac{t-1}{e^t-1}+\frac{t}{(e^t-1)^2}$$ 
$$\fb=\int_{0}^t\int_{0}^b(b-s)\lambda e^{-b}\cdot \lambda e^{s-t}dsdb=\lambda^2\cdot (\frac{t-2}{e^t}+\frac{t+2}{e^{2t}})$$
In any fixed posted price mechanism, note that the mechanism always achieves a larger GFT by choosing the same price for both agents. The gains from trade from posting at price $p$ is
$$\gft(p)=\int_{0}^p\int_{p}^v(b-s)\lambda e^{-b}\cdot \lambda e^{s-t}dbds=\lambda^2(\frac{t+2}{e^{2t}}+\frac{2}{e^t}-\frac{p+2}{e^{p+t}}-\frac{e^p(t+2-p)}{e^{2t}})<\lambda^2(\frac{t+2}{e^{2t}}+\frac{2}{e^t})$$

When $t$ is sufficiently large, $\fb$ is about $\lambda^2\cdot \frac{t-2}{e^t}$ while $\gft(p)$ is at most $\lambda^2\cdot \frac{2}{e^t}$, as $\frac{t+2}{e^{2t}}$ is negligible. Thus $\gft(p)=O(1/t)\cdot \fb$. On the other hand, $r=\Theta(\frac{t}{e^t})$ for large $t$, $\log(\frac{1}{r})=\Theta(t)$. Thus $\gft(p)=O(1/\log(\frac{1}{r}))\cdot \fb$.

We now verify that $\gft(p)=O(1/\log(\frac{1}{r}))\cdot \opt$ for any $p\in [0,t]$ and sufficiently large $t$. By~\cite{BrustleCWZ17},
$$\opt\geq \E_{b,s}[(b-s)\cdot\ind[\tilde{\varphi}(b)-s\geq 0]]$$

For the above distribution, $\varphi(b)=b-\frac{1-F(b)}{f(b)}=b-1+e^{b-t}$ is monotonic increasing in $b$. Thus $\tilde{\varphi}(b)=\varphi(b)$.
\begin{align*}
\frac{1}{\lambda^2}\cdot\E_{b,s}[(b-s)\cdot\ind[\tilde{\varphi}(b)-s\geq 0]]&=\int_{0}^t\int_0^{b-1+e^{b-t}}(b-s)\cdot e^{s-b-t}dsdb\\
&\geq\int_{0}^t\int_0^{b-1}(b-s)\cdot e^{s-b-t}dsdb\\
&=e^{-t}\cdot\int_{0}^t\int_{1}^{b}k\cdot e^{-k}dkdb &\text{($k=b-s$)}\\
&=e^{-t}\cdot\int_{0}^t (-e^{-k}(k+1)\big|_{1}^{b})ds\\
&=e^{-t}\cdot\int_{0}^t \left[\frac{2}{e}-e^{-b}(b+1)\right]ds\\
&=e^{-t}\cdot (\frac{2t}{e}+\frac{t+2}{e^t}-2)
\end{align*}


Thus when $t$ is sufficiently large, $\opt=\Omega(\lambda^2\cdot\frac{t}{e^t})$ and we have $\gft(p)=O(1/t)\cdot \opt=O(1/\log(\frac{1}{r}))\cdot \opt$. 
\end{example}


\begin{example}[$\sapp$ vs. $\pp$]
\label{example-sapp-beat-fp}
For any fixed $n$, consider the following instance for an additive buyer. $\DD_1^B$ and $\DD_1^S$ are distributions in Example~\ref{example-log(1/r)-tight} for some sufficiently large $t$. Pick any $C>0$. For every $i=2,\ldots,n$, $\DD_i^B$ is a degenerate distribution at $C$, i.e. the value is $C$ with probability 1. Distribution $\DD_i^S$ takes value $C+\epsilon$ with probability $1-\frac{1}{2n}$ and $C$ with probability $\frac{1}{2n}$, for some small $\epsilon>0$.
As shown in  Example~\ref{example-sapp-beat-fp}, when $t$ is large, $r_1=\Theta(\frac{t}{e^t})<\frac{1}{n}$. For $i\geq 2$, $r_i=\frac{1}{2n}$. Thus all items are ``unlikely to trade'' items ($r_i<\frac{1}{n}$).  

Note that for $i\geq 2$, $b_i$ is always no more than $s_i$. By Lemma~\ref{lem:SAPP bounding unlikely trade items}, $$\sapp=\Omega\left(\sum_i\E_{b_i,s_i}[(\tilde{\varphi}_i(b_i)-s_i)^+]\right)=\Omega(\E_{b_1,s_1}[(\tilde{\varphi}_1(b_1)-s_1)^+])$$

In Example~\ref{example-log(1/r)-tight}, when $t$ is sufficiently large, $\E_{b_1,s_1}[(\tilde{\varphi}_1(b_1)-s_1)^+]=\Omega(\lambda^2\cdot\frac{t}{e^t})$. On the other hand, any fixed price mechanism can only gain positive GFT from item 1. Thus $\pp=O(\lambda^2\cdot \frac{2}{e^t})$, which can be arbitrarily far from $\sapp$ as $t$ goes to infinity.
\end{example}

\paragraph{Dependence on $r$ is Necessary.}

We show that the dependence on $r=\min_i{r_i}$ is necessary for the approximation result of fixed posted price mechanisms. More formally, suppose fixed posted price mechanisms achieves an approximation ratio of $f(r_1,\ldots,r_n)$, for some $n$-ary function $f$. We will show that $f(r_1,\ldots,r_n)=\Omega(\log(1/r))$.
Consider the instance shown in Example~\ref{example-sapp-beat-fp}. Clearly $\fb=\E\left[(b_1-s_1)^+\right]$. 
Since all items other than item 1 always contribute 0 gains from trade, no fixed posted price mechanism can achieve better than $\Omega(\log(1/r_1))$-approximation to the first-best. Thus $f(r_1,\ldots,r_n)=\Omega(\log(1/r_1))$. Similarly we have $f(r_1,\ldots,r_n)=\Omega(\log(1/r_i))$ for all $i=1,\ldots,n$. Thus $f(r_1,\ldots,r_n)=\Omega(\log(1/r))$.

\paragraph{SAPP Mechanism is Necessary.} We provide the following example (Example~\ref{example:sapp-necessary}) to show that the class of SAPP mechanisms defined in Mechanism~\ref{mech:SAPP} is necessary to obtain any finite approximation ratio to $\opt$. 

Our example is constructed in the bilateral trade setting. By Lemma 4 of \cite{BrustleCWZ17}, the mechanism used in Lemma~\ref{lem:mechanism for OPT-B} shares the same allocation rule with the Buyer Offering (BO) mechanism, where the buyer picks a take-it or leave-it price according to her value and the seller can choose whether to sell at that price. Let the Seller Offering (SO) mechanism \cite{BlumrosenM16,BrustleCWZ17} be the mechanism analogous to the BO mechanism by switching the role of the buyer and the  seller. Denote $\bo$ and $\so$ the GFT of the BO and SO mechanism respectively. We show that in \Cref{example:sapp-necessary}, both $\pp$ and $\bo$ are arbitrarily far away from $\opt$. 

Following a recent breakthrough by Deng, Mao, Sivan and Wang which shows that $\max\{\so, \bo\}\leq 8.23\cdot\fb$~\cite{deng2021approximately}, we show that for Example~\ref{example:sapp-necessary}, $\so+\bo< c\cdot \fb$ for some absolute constant $c<1$.\footnote{The example was constructed in the early version of the paper~\cite{cai2021multi}. Here we bound $\so$ with a more careful analysis.} Example~\ref{example:sapp-necessary} also shows that for any integer $k$, there exists a bilateral trade instance such that $\so+k\cdot\bo$ is strictly less than $\fb$.\footnote{We can construct a similar family of instances such that for any integer $k$, there is an bilateral trade instance with $\bo+k\cdot \so$  strictly less than $\fb$.} Alternatively, a recent work by Babaioff, Dobzinski and Kupfer~\cite{babaioff2021note} independently obtain a similar result {that constructs a bilateral trade instance where } $\so+\bo<0.99\cdot \fb$.


\begin{lemma}\label{lem:example-3}
Given any integer $m\geq 3$. There exists a bilateral trade instance such that:
\begin{enumerate}
    \item $\bo\leq 1$, and $\pp\leq \log(m)$.
    \item $\frac{1}{4}\cdot \lfloor\log m\rfloor(\log m-\log\log m-1)\leq\fb\leq \log m\cdot (\log m+1)$
    \item $\fb-\so\in \left(\frac{\log\log m-1}{4}, \frac{\log(m+2)}{2}\right]$
\end{enumerate}
Hence for any sufficiently large integer $m$, both of the following inequalities hold: 
\begin{itemize}
    \item $\max\{\pp,\bo\}\leq O(\frac{1}{\log(m)})\cdot\opt$
    \item $\so+\frac{\log\log(m)}{5}\cdot \bo<\fb$
\end{itemize}
Moreover, there exists an integer $m$ and some absolute constant $c<1$ such that $\so+\bo<c\cdot \fb$.
\end{lemma}

\begin{example}\label{example:sapp-necessary}
For every positive integer $m\geq 2$, consider the bilateral trade instance where the seller's and buyer's (discrete) distributions are shown in the following tables. In the table, $g(s)$ (or $f(b)$) represents the density of the corresponding value in the support. 

\begin{table}[h]
    \centering
    \begin{tabular}{|c|c|c|c|c|c|c|}
    \hline
     $s$ & 0 & $2^m-2^{m-1}$ & \ldots & $2^m-2^k$ & \ldots & $2^m-1$\\
    \hline
    $g(s)$ & $\frac{1}{2^m}$ & $\frac{1}{2^m}$ & \ldots & $\frac{1}{2^{k+1}}$  & \ldots & $\frac{1}{2}$\\
    \hline
    $\tau(s)$ & 0 & $2^m$ & \ldots & $2^m$  & \ldots & $2^m$\\
    \hline
    \end{tabular}
    \caption{Seller's Distribution}
\end{table}

\begin{table}[h]
    \centering
    \begin{tabular}{|c|c|c|c|c|c|}
    \hline
     $b$ & $2^m-2^{L}$ & \ldots & $2^m-2^k$ & \ldots & $2^m-1$\\
    \hline
    $f(b)$ & $p_L$ & \ldots & $p_k$  & \ldots & $p_0$\\
    \hline
    \end{tabular}
    \caption{Buyer's Distribution}
\end{table}

For the seller's distribution, one can verify that the virtual value $\tau(s)$ is 0 if $s=0$ and $2^m$ elsewhere.\footnote{For discrete distributions, the virtual value for the seller's distribution is defined as $\tau(s)=s+\frac{\sum_{t<s}g(t)\cdot (s-s')}{g(s)}$, where $s'$ is the largest type in the support that is smaller than $s$.} 

For the buyer's distribution, choose $L=\lceil m-\log(m)\rceil$. Then define the sequence $\{p_k\}_{k=0}^L$ as follows: Construct the sequence $\{q_k\}_{k=0}^L$ with $q_0=1$, $q_1=\frac{1}{m-1}$, and for every $k=2,\ldots,L$, $q_{k}=\frac{m-k+2}{m-k}\cdot q_{k-1}$. Then for every $k=0,\ldots,L$, define $p_k=q_k/\sum_{j=0}^Lq_j$ for every $k$. 

By induction, we have $\sum_{j=0}^kq_j=q_{k+1}\cdot (m-k-1)$. Thus 
\begin{equation}\label{equ:p-relation}
\sum_{j=0}^kp_j=p_{k+1}\cdot (m-k-1).
\end{equation}
\end{example}


\begin{prevproof}{Lemma}{lem:example-3}

In the BO mechanism (as well as the mechanism in Lemma~\ref{lem:mechanism for OPT-B}), the buyer trades with the seller if and only if $b\geq \tau(s)$. Thus in Example~\ref{example:sapp-necessary}, the item trades only if $s=0$.

$$\bo=\sum_{k=0}^L (2^m-2^k)\cdot p_k\cdot\frac{1}{2^m}\leq \sum_{k=0}^L p_k.$$

For $\fb$, we have \begin{align*}
\fb&=\sum_{k=0}^L\left(\sum_{j=k+1}^{m-1}(2^j-2^k)\cdot p_k\cdot\frac{1}{2^{j+1}}+\frac{(2^m-2^k)p_k}{2^m}\right)\\
&\geq \sum_{k=0}^L p_k\cdot \left(\sum_{j=k+1}^{m-1}\frac{2^{j-1}}{2^{j+1}}+\frac{2^{m-1}}{2^m}\right)\geq\frac{1}{4}\cdot \sum_{k=0}^Lp_k\cdot (m-k),
\end{align*}
where the first inequality follows from the fact that $2^j-2^k\geq 2^{j-1}$ for any $j>k$.

Now consider any fixed posted price mechanism. Clearly the largest GFT is achieved when the posted prices are same for both the buyer and the seller. Without loss of generality we can assume the posted price $p$ lies in the support of distributions, i.e., $p=2^m-2^k$ for $k=0,\ldots,L$. For any $k$, the mechanism with posted price $p=2^m-2^k$ achieves GFT 
$$\sum_{j=0}^k\left(\sum_{i=k+1}^{m-1}(2^i-2^j)\cdot p_j\cdot\frac{1}{2^{i+1}}+\frac{(2^m-2^j)p_j}{2^m}\right)\leq \sum_{j=0}^kp_j\left(\sum_{i=k+1}^{m-1}\frac{2^i}{2^{i+1}}+1\right)\leq \sum_{j=0}^kp_j\cdot (m-k).$$

Let $Q_k=\sum_{j=0}^kp_j\cdot (m-k)$. Note that by the choice of sequence $\{p_k\}_{k=0}^L$ in Example~\ref{example:sapp-necessary}, we have for any $k=0,\ldots,L-1$ that
$$Q_{k+1}-Q_k=\sum_{j=0}^{k+1}p_j\cdot (m-k-1)-\sum_{j=0}^kp_j\cdot (m-k)=p_{k+1}\cdot (m-k-1)-\sum_{j=0}^k p_j=0.$$

Thus all $Q_k$s share the same value. Let this value be $Q$. Then $Q=Q_L=\sum_{j=0}^Lp_j\cdot (m-L)=m-L\leq \log(m)$. Moreover, $\pp\leq \max_kQ_k=Q\leq \log(m)$. $\bo\leq \sum_{k=0}^Lp_k=\frac{Q_L}{m-L}\leq\frac{1}{\log(m)}\cdot Q\leq 1$. 
On the other hand, 
\begin{equation}\label{equ:sapp is necessary FB}
\begin{aligned}
\fb&\geq\frac{1}{4}\cdot \sum_{k=0}^Lp_k\cdot (m-k)=\frac{1}{4}\cdot \sum_{k=0}^L\sum_{j=0}^{k-1}p_j~~~~~\text{(Equation~\ref{equ:p-relation})}\\
&=\frac{1}{4}\sum_{k=0}^L \frac{Q_{k-1}}{m-k+1}~~~~~\text{(Definition of $Q_k$)}\\
&\geq \frac{Q}{4}\cdot \int_{m-L+1}^{m+1}\frac{1}{x}dx=\frac{Q}{4}\cdot\log\left(\frac{m+1}{m-L+1}\right)\\
&\geq \frac{Q}{4}\cdot (\log(m)-\log\log(m)-1)~~~~~\text{(Definition of $L$)}
\end{aligned}
\end{equation}

When $m$ is sufficiently large, we have $\fb\geq \frac{Q}{5}\cdot \log(m)$. Thus
$$\frac{\log(m)}{5}\cdot\max\{\pp,\bo\}\leq \fb.$$


We derive the upper bound of $\fb$ similar to \Cref{equ:sapp is necessary FB}:

\begin{equation}\label{equ:SAPP is necessary FB 2}
\begin{aligned}
\fb&=\sum_{k=0}^L\left(\sum_{j=k+1}^{m-1}(2^j-2^k)\cdot p_k\cdot\frac{1}{2^{j+1}}+\frac{(2^m-2^k)p_k}{2^m}\right)\\
&\leq\sum_{k=0}^L p_k\cdot \left(\sum_{j=k+1}^{m-1}\frac{2^j}{2^{j+1}}+1\right)\\
&\leq\sum_{k=0}^Lp_k\cdot (m-k)\\
&=\sum_{k=0}^L \frac{Q_{k-1}}{m-k+1}~~~~\text{(\Cref{equ:p-relation} and definition of $Q_k$)}\\
&=Q\cdot \sum_{k=0}^L \frac{1}{m-k+1}\\
&\leq Q\cdot \int_{m-L}^{m+1}\frac{1}{x}dx=Q\cdot \log\left(\frac{m+1}{m-L}\right)\\
&\leq \log m\cdot (\log m+1)
\end{aligned}
\end{equation}

Next we analyze the GFT of the SO mechanism. By \cite{BrustleCWZ17},

$$\so=\E_{b,s}[(b-s)\cdot \ind[\tilde{\varphi}(b)-s\geq 0]].$$

We calculate the buyer's virtual value.\footnote{For discrete distributions, the buyer's virtual value is defined as $\varphi(b)=b-\frac{\sum_{t>b}f(t)\cdot (b'-b)}{f(b)}$, where $b'$ is the smallest type in the support that is larger than $b$.} We have that $\varphi(2^m-1)=2^m-1$, and for every $k=1,\ldots,L$, $$\varphi(2^m-2^k)=(2^m-2^k)-\frac{\sum_{j=0}^{k-1}p_j\cdot 2^{k-1}}{p_k}=2^m-2^k-2^{k-1}(m-k).$$

{
Thus $\varphi(\cdot)$ is monotone increasing and $\tilde{\varphi}(b)=\varphi(b)$ for every $b$ in the support. Let $k,j$ be the numbers such that $b=2^m-2^k$ and $s=2^m-2^j$. Then $\varphi(b)\geq s$ if and only if $j\geq k+\log(m-k+2)$.
We have  
\begin{equation}\label{equ:sapp is necessary SO}
\begin{aligned}
\so&=\sum_{k=0}^L\left(\sum_{j=k+\lceil\log(m-k+2)\rceil}^{m-1}(2^j-2^k)\cdot p_k\cdot\frac{1}{2^{j+1}}+\frac{(2^m-2^k)p_k}{2^m}\right)\\
&=\fb-\sum_{k=0}^L\sum_{j=k+1}^{k+\lceil\log(m-k+2)\rceil-1}(2^j-2^k)\cdot p_k\cdot\frac{1}{2^{j+1}}
\end{aligned}
\end{equation}
}
By \Cref{equ:sapp is necessary SO},
\begin{align*}
\fb-\so
&\leq \sum_{k=0}^L\sum_{j=k+1}^{k+\lceil\log(m+2)\rceil-1} p_k\cdot\frac{2^j}{2^{j+1}}\\
&\leq \frac{\log(m+2)}{2}\sum_{k=0}^Lp_k=\frac{\log(m+2)}{2}.
\end{align*}
Note that $\fb\geq \frac{Q}{5}\cdot \log(m)=\Omega(\log^2(m))$. When $m\to \infty$, we have $\frac{\opt}{\fb}\to 1$ since $\opt\geq \so$. Thus $$\max\{\pp,\bo\}\leq O(\frac{1}{\log(m)})\cdot\opt.$$

By \Cref{equ:sapp is necessary SO},

\begin{equation}\label{equ:sapp is necessary SO 2}
\begin{aligned}
\fb-\so
&\geq \sum_{k=0}^L\sum_{j=k+1}^{k+\lceil\log(m-L+2)\rceil-1}(2^j-2^{j-1})\cdot p_k\cdot\frac{1}{2^{j+1}}\\
&=\frac{\lceil\log(m-L+2)\rceil-1}{4}\sum_{k=0}^Lp_k>\frac{\log\log(m)-1}{4},    
\end{aligned}
\end{equation}
where the last inequality follows from $\sum_{k=0}^Lp_k=1$ and $L=\lceil m-\log(m)\rceil$. Since $\bo\leq \sum_{k=0}^Lp_k=1$ and $\frac{\log\log (m)-1}{4}>\frac{\log\log (m)}{5}$ for sufficiently large $m$, we have $\so+\frac{\log\log(m)}{5}\cdot \bo<\fb$.

To prove the last part of the statement, choose $m$ such that $\log\log(m)=6$. Then by \Cref{equ:sapp is necessary SO 2}, $\fb-\so>\frac{5}{4}$. $\bo\leq 1$. Moreover by \Cref{equ:SAPP is necessary FB 2}, $\fb\leq e^6\cdot (e^6+1)$. Thus by choosing $c=1-\frac{1}{4e^6(e^6+1)}<1$, we have
$$\bo+\so<\fb-\frac{1}{4}\leq c\cdot \fb$$
\end{prevproof}
\section{Mechanism Design Background}\label{apx:md}

\notshow{
\noindent\textbf{IC, IR constraints}

\begin{itemize}
    \item BIC: For every agent, reporting her true value (or cost) maximizes her expected utility over the profile of other agents.
    \item DSIC: For every agent, reporting her true value (or cost) maximizes her expected utility, no matter what other agents report.
    \item (Bayesian) IR: For every agent, reporting her true value (or cost) derives non-negative utility over the profile of other agents.
    \item Ex-post IR: For every agent, reporting her true value (or cost) derives non-negative utility, no matter what other agents report.
\end{itemize}
}

\noindent\textbf{Myerson's Lemma}

For reference, we formally state Myerson's Lemma.


\begin{lemma}[Myerson's Lemma~\cite{Myerson81}] \label{lem:mye}
In a setting with single-dimensional preferences, where the buyer's distribution for item $i$ is $\DD_i^B$ and seller $i$'s distribution is $\DD_i^S$, let the probability of trade for each item be $\hat x=\{\hat x_i(\Bb,\Bs)\}_{i\in [n]}$.
Denote the interim allocation rule for the sellers as $\hat{x}_i(\Bs) = \E_\Bb[\hat{x}_i(\Bb,\Bs)]$. 
In order to be BIC, sellers' payments must meet the following payment identity:
$$p_i^S(\Bs)=s_i\cdot \hat{x}_i(s_i,s_{-i})+\int_{s}^\infty \hat{x}_i(t,s_{-i})dt.$$
Then let $\tau_i(s_i) = s_i + \frac{G_i(s_i)}{g_i(s_i)}$ be seller $i$'s Myerson virtual value, and $\widetilde{\tau}_i(s_i)$ the \emph{ironed} virtual value, obtained by averaging the virtual values in quantile space to enforce that $\widetilde{\tau}_i(\cdot)$ is monotone non-decreasing in $s_i$.  Then expected payment equals expected virtual welfare:
$$\E_{\Bs}\left[\sum_i p_i^S(\Bb,\Bs)\right]=\E_{\Bs}\left[\sum_i \hat x_i(\Bb, \Bs) \cdot \widetilde{\tau}_i(s_i)\right].$$

Further, let $\varphi_i(b_i) = b_i - \frac{1-F_i(b_i)}{f_i(b_i)}$ be a single-dimensional buyer's Myerson virtual value, and $\widetilde{\varphi}_i(b_i)$ the ironed virtual value.  Then
$$\int _p ^\infty \widetilde{\varphi}_i(x) dx = p \cdot [1-F_i(p)].$$

Similarly, for IC single-dimensional payments $p_i^B(b_i)$, then
$$\E_{\Bs}\left[\sum_i p_i^B(\Bb,\Bs)\right]=\E_{\Bs}\left[\sum_i \hat x_i(\Bb, \Bs) \cdot \widetilde{\varphi}_i(b_i)\right].$$

\end{lemma}

\section{Missing Details from 
Section~\ref{sec:log1/r}}\label{sec:appx_log1/r}

\subsection{Missing Proofs from the Upper Bound of $\fb$ in Section~\ref{subsec:breaking terms}}

\notshow{
\begin{prevproof}{Lemma}{lem:x_i-vs-y_i}
Note that for every $i\in [n]$, $b_i<x_i\wedge s_i>x_i$ implies that $b_i<s_i$.  We have
\begin{align*}
    1-r_i&=\Pr_{b_i\sim \DD_i^B,s_i\sim \DD_i^S}[b_i<s_i]
\leq \Pr_{b_i,s_i}[b_i<x_i\wedge s_i>x_i]=(1-\frac{r_i}{2})\cdot (1-\Pr_{s_i}[s_i\leq x_i]). 
\end{align*}

Suppose $x_i<y_i$. Then 
$(1-\frac{r_i}{2})\cdot (1-\Pr_{s_i}[s_i\leq x_i])\geq (1-\frac{r_i}{2})^2>1-r_i.$
This is a contradiction. Thus $x_i\geq y_i$.
\end{prevproof}
}

\notshow{
\begin{figure}[h!] 
	\centering
	\includegraphics[scale=.5]{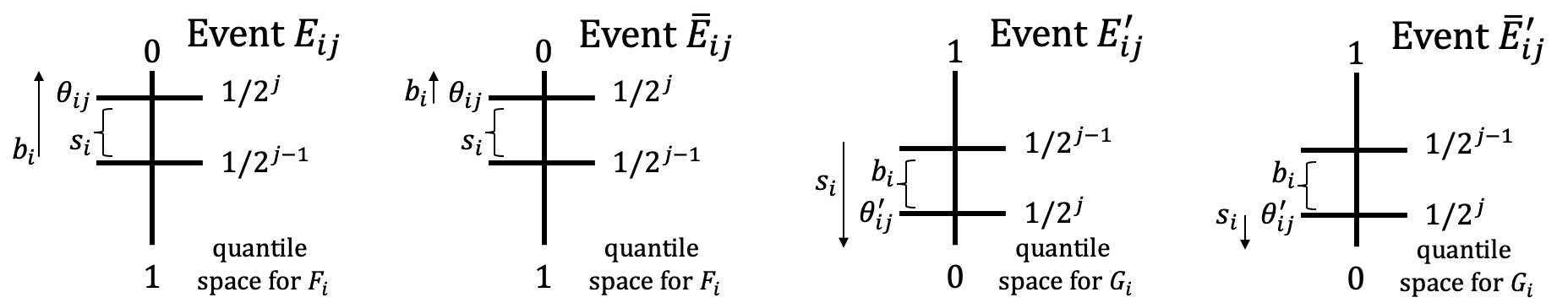} 
	\caption{Events $E_{ij}$ and $E'_{ij}$ used in the upper bound of GFT, and the restricted events $\overline{E}_{ij}$ and $\overline{E}'_{ij}$.}
	\label{fig:events} 
\end{figure}
}

\begin{prevproof}{Lemma}{lem:breaking-terms-1}

For every $i,\Bs,b_i$, define $$q_i(b_i,\Bs)=(b_i-s_i)^+\cdot\Pr_{b_{-i}}[i\in S^*(\Bb,\Bs)]\cdot \ind\left[ \overline{F_i}^{-1}(\frac{1}{2^{j-1}})\leq s_i\leq \overline{F_i}^{-1}(\frac{1}{2^{j}})\right].$$ Then we have that $q_i(b_i,\Bs)\geq 0$ is non-decreasing in $b_i$, as both $b_i-s_i$ and the probability $\Pr_{b_{-i}}[i\in S^*(\Bb,\Bs)]$ is non-decreasing in $b_i$.

Since $\theta_{ij}=\overline{F_i}^{-1}(\frac{1}{2^{j}})$ and $\Pr_{b_i}\left[b_i\geq \overline{F_i}^{-1}(\frac{1}{2^{j}})\right] = \frac{1}{2} \Pr_{b_i}\left[b_i\geq \overline{F_i}^{-1}(\frac{1}{2^{j-1}})\right]$, we have

\begin{equation}\label{equ:proof-term4-2}
\begin{aligned}
\E_{b_i}\left[q_i(b_i,\Bs)\cdot\ind\left[b_i\geq\theta_{ij}\right]\right]
\geq&q_i(\theta_{ij},\Bs)\cdot\Pr_{b_i}\left[b_i\geq\theta_{ij}\right]\\
= & q_i(\theta_{ij},\Bs)\cdot\Pr_{b_i}\left[\overline{F_i}^{-1}(\frac{1}{2^{j-1}})\leq b_i< \theta_{ij}\right]\\
\geq & \E_{b_i}\left[q_i(b_i,\Bs)\cdot\ind\left[\overline{F_i}^{-1}(\frac{1}{2^{j-1}})\leq b_i< \theta_{ij}\right]\right].
\end{aligned}
\end{equation}


Thus we have

\begin{align*}
    &\E_{\Bb,\Bs}\left[\sum_i(b_i-s_i)^+\cdot\ind[i\in S^*(\Bb,\Bs)\wedge E_{ij}] \right] \\
    =&\sum_i\E_{b_i,\Bs}\left[q_i(b_i,\Bs)\cdot\ind[b_i\geq \overline{F_i}^{-1}(\frac{1}{2^{j-1}})]\right]\\
    \leq& 2\cdot\sum_i\E_{b_i,\Bs}\left[q_i(b_i,\Bs)\cdot\ind[b_i\geq \theta_{ij}]\right]~~~~~~~\text{(Inequality~\eqref{equ:proof-term4-2})}\\
    =&2\cdot\E_{\Bb,\Bs}\left[\sum_i(b_i-s_i)^+\cdot\ind[i\in S^*(\Bb,\Bs)\wedge \overline{E}_{ij}] \right]\\
    \leq & 2\cdot\E_{\Bb,\Bs}\left[\sum_i(b_i-\theta_{ij})^+\cdot\ind[i\in S^*(\Bb,\Bs)\wedge \overline{E}_{ij}]\right]+2\cdot\E_{\Bb,\Bs}\left[\sum_i(\theta_{ij}-s_i)^+\cdot\ind[i\in S^*(\Bb,\Bs)\wedge \overline{E}_{ij}]\right]
\end{align*}

Moreover, we have
\begin{equation}
\begin{aligned}
    &\E_{\Bb,\Bs}\left[\sum_i(b_i-\theta_{ij})^+\cdot \ind[i\in S^*(\Bb,\Bs)]\cdot \ind[\overline{E}_{ij}]\right]\\ 
    \leq& \E_{\Bb,\Bs}\left[\sum_i(b_i-\theta_{ij})^+\cdot \ind[s_i\leq \theta_{ij}]\cdot \ind[i\in S^*(\Bb,\Bs)]\right]\\
    \leq& \E_{\Bb,\Bs}\left[\max_{S\in \cF}\sum_{i\in S}\left\{(b_i-\theta_{ij})^+\cdot \ind[s_i\leq \theta_{ij}]\right\}\right]
\end{aligned}
\end{equation}
Similarly,
$$\E_{\Bb,\Bs}\left[\sum_i(\theta_{ij}-s_i)^+\cdot\ind[i\in S^*(\Bb,\Bs)\wedge \overline{E}_{ij}]\right]\leq \E_{\Bb,\Bs}\left[\max_{S\in \cF}\sum_{i\in S}\left\{(\theta_{ij}-s_i)^+\cdot \ind[b_i\geq\theta_{ij}]\right\}\right]$$

\end{prevproof}

\begin{prevproof}{Lemma}{lem:breaking-terms-2}
For every $i\in [n]$ and $j=1,...,\lceil \log(2/r)\rceil$, let $E_{ij}'$ be the event that $G_i^{-1}(\frac{1}{2^j})\leq b_i\leq G_i^{-1}(\frac{1}{2^{j-1}})\wedge s_i\leq G_i^{-1}(\frac{1}{2^{j-1}})$ and $\overline{E}_{ij}'$ be the event that $G_i^{-1}(\frac{1}{2^j})\leq b_i\leq G_i^{-1}(\frac{1}{2^{j-1}})\wedge s_i\leq G_i^{-1}(\frac{1}{2^{j}})$. We have
$$\circled{2}\leq \textstyle\sum\limits_{j=1}^{\lceil \log(\frac{2}{r})\rceil}\E_{\Bb,\Bs}\left[\sum_i(b_i-s_i)^+\cdot\ind[i\in S^*(\Bb,\Bs)\wedge E_{ij}'] \right].$$

Fix any $j$. For every $i,\Bb,s_i$, define $$q_i(\Bb,s_i)=(b_i-s_i)^+\cdot\Pr_{s_{-i}}[i\in S^*(\Bb,\Bs)]\cdot \ind\left[G_i^{-1}(\frac{1}{2^j})\leq b_i\leq G_i^{-1}(\frac{1}{2^{j-1}})\right].$$ Then we have that $q_i(\Bb,s_i)>0$ is non-increasing in $s_i$. Since $\theta_{ij}'=G_i^{-1}(\frac{1}{2^j})$ and $\Pr_{s_i}\left[s_i\leq G_i^{-1}(\frac{1}{2^{j}})\right] = \frac{1}{2} \Pr_{s_i}\left[s_i\leq G_i^{-1}(\frac{1}{2^{j-1}})\right]$, we have

\begin{equation}\label{equ:breaking-terms-2}
\begin{aligned}
\E_{s_i}\left[q_i(\Bb,s_i)\cdot\ind\left[s_i\leq\theta_{ij}'\right]\right]
\geq&q_i(\Bb,\theta_{ij}')\cdot\Pr_{s_i}\left[s_i\leq\theta_{ij}'\right]\\
= & \frac{1}{2}q_i(\Bb,\theta_{ij}')\cdot \Pr_{s_i}\left[s_i\leq G_i^{-1}(\frac{1}{2^{j-1}})\right]\\
\geq & \frac{1}{2} \E_{s_i}\left[q_i(\Bb,s_i)\cdot\ind\left[s_i\leq G_i^{-1}(\frac{1}{2^{j-1}})\right]\right].
\end{aligned}
\end{equation}


Thus we have
\begin{align*}
    &\E_{\Bb,\Bs}\left[\sum_i(b_i-s_i)^+\cdot\ind[i\in S^*(\Bb,\Bs)\wedge E_{ij}']\right]\\
    =&\sum_i\E_{\Bb,s_i}\left[q_i(\Bb,s_i)\cdot\ind[s_i\leq G_i^{-1}(\frac{1}{2^{j-1}})]\right]\\
    \leq& 2\cdot \sum_i\E_{\Bb,s_i}\left[q_i(\Bb,s_i)\cdot\ind[s_i\leq \theta_{ij}']\right]~~~~~~~~~\text{(Inequality~\ref{equ:breaking-terms-2})}\\
    =&2\cdot \E_{\Bb,\Bs}\left[\sum_i(b_i-s_i)^+\cdot\ind[i\in S^*(\Bb,\Bs)\wedge \overline{E}_{ij}']\right]\\
    \leq & 2\cdot \E_{\Bb,\Bs}\left[\sum_i(b_i-\theta_{ij}')^+\cdot\ind[i\in S^*(\Bb,\Bs)\wedge \overline{E}_{ij}']\right]+2\cdot \E_{\Bb,\Bs}\left[\sum_i(\theta_{ij}'-s_i)^+\cdot\ind[i\in S^*(\Bb,\Bs)\wedge \overline{E}_{ij}']\right]
\end{align*}

Moreover, we have
\begin{equation}
\begin{aligned}
    &\E_{\Bb,\Bs}\left[\sum_i(b_i-\theta_{ij}')^+\cdot \ind[i\in S^*(\Bb,\Bs)]\cdot \ind[\overline{E}_{ij}']\right]\\ 
    \leq& \E_{\Bb,\Bs}\left[\sum_i(b_i-\theta_{ij}')^+\cdot \ind[s_i\leq \theta_{ij}']\cdot \ind[i\in S^*(\Bb,\Bs)]\right]\\
    \leq& \E_{\Bb,\Bs}\left[\max_{S\in \cF}\sum_{i\in S}\left\{(b_i-\theta_{ij}')^+\cdot \ind[s_i\leq \theta_{ij}']\right\}\right]
\end{aligned}
\end{equation}
Similarly,
$$\E_{\Bb,\Bs}\left[\sum_i(\theta_{ij}'-s_i)^+\cdot\ind[i\in S^*(\Bb,\Bs)\wedge \overline{E}_{ij}']\right]\leq \E_{\Bb,\Bs}\left[\max_{S\in \cF}\sum_{i\in S}\left\{(\theta_{ij}'-s_i)^+\cdot \ind[b_i\geq\theta_{ij}']\right\}\right]$$

\end{prevproof}

\notshow{

\subsection{Bounding Buyer Surplus} \label{apx:buyer surplus}

We bound terms $\circled{3}$ and $\circled{5}$ using fixed posted price mechanisms. The result holds for a general constrained-additive buyer. 

\begin{lemma}\label{lem:term3&5}(Restatement of Lemma~\ref{lem:buyer-surplus-main-body})
For any $\{p_i\}_{i\in [n]}\in \mathbb{R}_{+}^n$,
$$\E_{\Bb,\Bs}\left[\max_{S\in \cF}\sum_{i\in S}\{(b_i-p_i)^+\cdot\ind[s_i\leq p_i]\}\right]\leq \pp$$

Thus both $\emph{\circled{3}}$ and $\emph{\circled{5}}$ are upper bounded by $O(\log(\frac{1}{r}))\cdot \pp$.
\end{lemma}


\begin{proof}
Consider the fixed posted price mechanism $\MM$ with $\theta_i^S=\theta_i^B=p_i$. For every $\Bs$, let $A(\Bs)=\{i\in [n]~|~s_i\leq p_i\}$ be the set of available items. Then the buyer will choose the best set $S\subseteq A(\Bs),S\in \cF$ that maximizes $\sum_{i\in S}(b_i-p_i)^+$ (and not buy any item if $b_i - p_i \leq 0$ for all $i$). Thus the gains from trade $\sum_{i\in S}(b_i-s_i)$ are at least $\sum_{i\in S}(b_i-p_i)^+\geq 0$. We have

$$\gft(\MM)\geq \E_{\Bb,\Bs}\left[\max_{S\subseteq A(\Bs),S\in \cF}\sum_{i\in S}(b_i-p_i)^+\right]= \E_{\Bb,\Bs}\left[\max_{S\in \cF}\sum_{i\in S}\left\{(b_i-p_i)^+\cdot \ind[s_i\leq p_i]\right\}\right]$$

To bound terms \circled{3} and \circled{5}, just apply the above inequality with $p_i=\theta_{ij}$ (or $\theta_{ij}'$).
\end{proof}

}

\notshow{
\begin{proof}
Fix any $j$. Consider the fixed posted price mechanism $\MM$ with $\theta_i^S=\theta_i^B=\theta_{ij} = \overline{F_i}^{-1}(\frac{1}{2^{j}})$. For every $\Bs$, let $A(\Bs)=\{i\in [n]~|~s_i\leq \theta_{ij}\}$ be the set of available items. Then the buyer will choose the best set $S\subseteq A(\Bs),S\in \cF$ that maximizes $\sum_{i\in S}(b_i-\theta_{ij})^+$ (and not buy any item if $b_i - \theta_{ij} \leq 0$ for all $i$). Thus the gains from trade $\sum_{i\in S}(b_i-s_i)$ are at least $\sum_{i\in S}(b_i-\theta_{ij})^+\geq 0$. Recall that $E_{ij}$ is the event that $s_i \in [\overline{F_i}^{-1}(\frac{1}{2^{j-1}}), \overline{F_i}^{-1}(\frac{1}{2^{j}})]$ and $b_i \geq \overline{F_i}^{-1}(\frac{1}{2^{j-1}})$, and that $S^*(\Bb,\Bs)=\argmax_{S\in \cF}\sum_{k\in S}(b_k-s_k)$.  
We have
\begin{equation}
\begin{aligned}
    \gft(\MM)&\geq \E_{\Bb,\Bs}\left[\max_{S\subseteq A(\Bs),S\in \cF}\sum_{i\in S}(b_i-\theta_{ij})^+\right]\\
    &= \E_{\Bb,\Bs}\left[\max_{S\in \cF}\sum_{i\in S}\left\{(b_i-\theta_{ij})^+\cdot \ind[s_i\leq \theta_{ij}]\right\}\right]\\
    &\geq \E_{\Bb,\Bs}\left[\sum_i(b_i-\theta_{ij})^+\cdot \ind[s_i\leq \theta_{ij}]\cdot \ind[i\in S^*(\Bb,\Bs)]\right]\\
    &\geq \E_{\Bb,\Bs}\left[\sum_i(b_i-\theta_{ij})^+\cdot \ind[i\in S^*(\Bb,\Bs)]\cdot \ind[\overline{E}_{ij}] \right]
\end{aligned}
\end{equation}

The last inequality uses the fact that every profile in the event $\overline{E}_{ij}$ satisfies that $s_i \leq  \overline{F_i}^{-1}(\frac{1}{2^{j}})$. 

The proof of the other term is analogous, by considering the fixed posted price mechanism with $\theta_i^B=\theta_i^S=\theta_{ij}'$.

For the other term, consider the fixed posted price mechanism $\MM_j'$ with $\theta_i^S=\theta_i^B=\theta_{ij}'$. For every $\Bs$, let $A'(\Bs)=\{i\in [n]~|~s_i\leq \theta_{ij}'\} = G_i^{-1}(\frac{1}{2^{j}})$ be the set of available items. We have

\begin{equation}\label{equ:gft-pp}
\begin{aligned}
\pp&\geq \E_{\Bb,\Bs}\left[\max_{S\subseteq A'(\Bs),S\in \cF}\sum_{i\in S}(b_i-\theta_{ij}')^+\right]\\
    &= \E_{\Bb,\Bs}\left[\max_{S\in \cF}\sum_{i\in S}\left\{(b_i-\theta_{ij}')^+\cdot \ind[s_i\leq \theta_{ij}']\right\}\right]\\
    &\geq \E_{\Bb,\Bs}\left[\sum_i(b_i-\theta_{ij}')^+\cdot \ind[s_i\leq \theta_{ij}']\cdot \ind[i\in S^*(\Bb,\Bs)]\right]\\
\end{aligned}
\end{equation}

For every $i,\Bb,s_i$, define $q_i(\Bb,s_i)=\Pr_{s_{-i}}[i\in S^*(\Bb,\Bs)]$. Then we have that $q_i(\Bb,s_i)$ is non-increasing in $s_i$. Since $\Pr_{s_i}\left[s_i\leq G_i^{-1}(\frac{1}{2^{j}})\right] = \frac{1}{2} \Pr_{s_i}\left[s_i\leq G_i^{-1}(\frac{1}{2^{j-1}})\right]$, we have

\begin{align*}
\E_{s_i}\left[q_i(\Bb,s_i)\cdot\ind\left[s_i\leq\theta_{ij}'\right]\right]
\geq&q_i(\Bb,\theta_{ij}')\cdot\Pr_{s_i}\left[s_i\leq\theta_{ij}'\right]\\
= & \frac{1}{2}q_i(\Bb,\theta_{ij}')\cdot \Pr_{s_i}\left[s_i\leq G_i^{-1}(\frac{1}{2^{j-1}})\right]\\
\geq & \frac{1}{2} \E_{s_i}\left[q_i(\Bb,s_i)\cdot\ind\left[s_i\leq G_i^{-1}(\frac{1}{2^{j-1}})\right]\right].
\end{align*}


Continuing the RHS of \eqref{equ:gft-pp}, by moving the expectation over the sellers profile inside and then back out, we get:
\begin{align*}
    &\E_{\Bb,\Bs}\left[\sum_i(b_i-\theta_{ij}')^+\cdot \ind[s_i\leq \theta_{ij}']\cdot \ind[i\in S^*(\Bb,\Bs)]\right]\\
    =&\E_{\Bb}\left[\sum_i(b_i-\theta_{ij}')^+\cdot \E_{s_i}[q_i(\Bb,s_i)\cdot\ind[s_i\leq \theta_{ij}']]\right]\\
    \geq& \frac{1}{2}\cdot\E_{\Bb}\left[\sum_i(b_i-\theta_{ij}')^+\cdot\E_{s_i}\left[q_i(\Bb,s_i)\cdot\ind\left[s_i\leq G_i^{-1}(\frac{1}{2^{j-1}})\right]\right]\right]\\
    =& \frac{1}{2}\E_{\Bb,\Bs}\left[\sum_i(b_i-\theta_{ij}')^+\cdot \ind[i\in S^*(\Bb,\Bs)]\cdot\ind\left[s_i\leq G_i^{-1}(\frac{1}{2^{j-1}})\right]\right]\\
    \geq& \frac{1}{2}\E_{\Bb,\Bs}\left[\sum_i(b_i-\theta_{ij}')^+\cdot \ind[i\in S^*(\Bb,\Bs)]\cdot\ind[E_{ij}']\right]
\end{align*}

\end{proof}

}

\notshow{

\subsection{Missing Details from Bounding Seller Surplus for a Unit-Demand Buyer in Section~\ref{subsec:seller-surplus-UD}} \label{apx:seller-surplus-UD}

\begin{proof}[Proof of Lemma~\ref{lem:term4&6-UD}]
For every $i$, let $$v_i=(p_i-s_i)^+\cdot\ind[b_i\geq p_i]$$ be a random variable that depends on $b_i$ and $s_i$. Let $\Bv=\{v_i\}_{i\in [n]}$. Let $V_i$ be the distribution of $v_i$ as $b_i\sim \DD_i^B, s_i\sim \DD_i^S$, and $V=\times_{i=1}^nV_i$ be the distribution of $\Bv$. Then the last line of the inequality above is equal to $\E_{\Bv\sim V}[\max_iv_i]$.

Consider any threshold $\xi>0$.  Observe that  $v_i\geq\xi$ if and only if $b_i\geq p_i\wedge p_i-s_i\geq \xi$. Consider the fixed posted price mechanism $\MM$ with $\theta_i^B=p_i$ and $\theta_i^S=p_i-\xi$ for every $i\in [n]$.  
Whenever the buyer purchases some item $i$, we must have $b_i\geq p_i$ (the buyer buys) and $s_i\leq p_i-\xi$ (the seller sells), and the contributed GFT satisfies $b_i-s_i\geq p_i-s_i\geq \xi$. In addition, the buyer will purchase \emph{some} item if and only if there exists some $i$ such that $v_i\geq \xi$. 
Therefore we can apply the prophet inequality~\cite{krengel1978semiamarts,samuel1984comparison,KleinbergW12} with threshold $\xi=\frac{1}{2}\cdot \E_{\Bv\sim V}[\max_iv_i]$ to ensure that the GFT of mechanism $\MM$ is at least $\frac{1}{2}\E_{\Bv\sim V}[\max_iv_i]$. 

\end{proof}

}

\notshow{

An OCRS is an algorithm defined for the following online selection problem: There is a ground set $I$, and the elements are revealed one by one, with item $i$ \emph{active} with probability $x_i$ independent of the other items. The algorithm is only allowed to accept active elements and has to irrevocably make a decision whether to accept an element before the next one is revealed. Moreover, the algorithm can only accept a set of elements subject to a feasibility constraint $\cF$. We use the vector $x$ to denote active probabilities for the elements and $R(x)$ to denote the random set of active elements.

\begin{definition}[Online Contention Resolution Scheme]\label{def-ocrs}
An \emph{Online Contention Resolution Scheme} (OCRS) for a polytope $P\subseteq [0,1]^{|I|}$ and feasibility constraint $\cF$ is an online algorithm that selects a feasible and active set $S\subseteq R(x)$ and $S\in \cF$ for any $x\in P$. 
A \textbf{greedy OCRS} $\pi$ greedily decides whether or not to select an element in each iteration: given the vector $x\in P$, it first determines a sub-constraint $\cF_{\pi,x}\subseteq \cF$. When element $i$ is revealed, it accepts the element if and only if $i$ is active and $S\cup \{i\}\in \cF_{\pi,x}$, where $S$ is the set of elements accepted so far. In most cases, we choose $P$ to be $P_{\cF}$, the convex hull of all characteristic vectors of feasible sets in $\cF$: $P_{\cF}=conv(\ind_S~|~S\in \cF)$.
\end{definition}

\begin{definition}[$(\delta,\eta)$-selectability \cite{FeldmanSZ16}]\label{def-selectablity}
For any $\delta,\eta\in (0,1)$,  
a greedy OCRS $\pi$ for $P$ and $\cF$ is $(\delta,\eta)$-selectable if for every $x\in \delta\cdot P$ and $i\in I$
$$\Pr[S\cup \{i\}\in \cF_{\pi,x}, \forall S\subseteq R(x), S\in\cF_{\pi,x}]\geq \eta.$$
The probability is taken over the randomness of $R(x)$ and the subconstraint $\cF_{\pi,x}$. We slightly abuse notation and say that $\cF$ is $(\delta,\eta)$-selectable if there exists a $(\delta,\eta)$-selectable greedy OCRS for $P_{\cF}$ and $\cF$.
\end{definition}

The following lemma is adapted from~\cite{FeldmanSZ16} and connects $(\delta,\eta)$-selectability to \rppabbrev~mechanisms. Once again, the OCRS gives us both a GFT guarantee and a mechanism: variables $v_i$ correspond to the bound on seller surplus, buyer item prices are $\{p_i\}_{i\in [n]}$, and seller prices are $\{p_i - \xi_i\}_{i\in [n]}$.

\begin{lemma}\label{lem:feldman-selectable}
Suppose there exists a $(\delta,\eta)$-selectable greedy OCRS $\pi$ for the polytope $P_{\cF}$, for some $\delta,\eta\in (0,1)$. Fix any $\{p_i\}_{i\in [n]}\in \mathbb{R}_{+}^n$. For every $i\in [n]$, let $v_i=(p_i-s_i)^+\cdot \ind[b_i\geq p_i]$. For any $\Bq \in P_{\cF}$ that satisfies $q_i\leq \Pr_{b_i,s_i}[b_i\geq p_i>s_i] $ $\forall i$, let $\xi_i=p_i-G_i^{-1}(q_i/\Pr[b_i\geq p_i])$.\footnote{When $q_i\leq \Pr_{b_i,s_i}[b_i\geq p_i>s_i]$, $q_i/\Pr[b_i\geq p_i]\leq 1$. Thus $\xi_i$ is well-defined.} We have
$$\textstyle\sum_i\E_{b_i,s_i}\left[v_i\cdot \ind[v_i\geq \xi_i]\right]\leq \frac{1}{\delta\eta}\cdot \gcfpp.$$

Moreover, there exists a choice of $\Bq$ such that
$$\textstyle\E_{\Bb,\Bs}\left[\max_{S\in \cF}\sum_{i\in S}\left\{(p_i-s_i)^+\cdot\ind[b_i\geq p_i]\right\}\right]\leq \textstyle\sum_i\E_{b_i,s_i}\left[v_i\cdot \ind[v_i\geq \xi_i]\right]\leq \frac{1}{\delta\eta}\cdot \gcfpp.$$
\end{lemma}

\begin{prevproof}{Lemma}{lem:feldman-selectable}
Let $V_i$ denote the distribution of $v_i$ when $b_i\sim \DD_i^B, s_i\sim \DD_i^S$. Let $\Bv=\{v_i\}_{i\in [n]}$. Let $\hat{\textbf{q}}$ be a scaled-down vector of $\textbf{q}$ such that $\hat{q}_i=\delta\cdot q_i$ for every $i\in [n]$ and $\hat{\xi_i}=p_i-G_i^{-1}(\hat{q}_i/\Pr[b_i\geq p_i])$. This is also well-defined since $\hat{q}_i<q_i\leq\Pr[b_i\geq p_i]$. As $q \in P_{\cF}$, then $\hat{q}\in \delta\cdot P_{\cF}$. Consider the \cfpp~mechanism $\MM$ with buyer posted prices $\{p_i\}_{i\in [n]}$, seller posted prices $\{p_i-\hat{\xi}_i\}_{i\in [n]}$, and subconstraint $\cF_{\pi,\hat{q}}\in \cF$ stated in Definition~\ref{def-selectablity}. 

Fix any item $i\in [n]$. We say item $i$ as active if $v_i\geq \hat{\xi}_i$. Similarly to Section~\ref{subsec:seller-surplus-UD},  $v_i\geq \hat{\xi}_i$ if and only if $b_i\geq p_i\wedge s_i\leq p_i-\hat{\xi}_i$. That is, $i$ is active if and only if item $i$ is on the market and the buyer can afford it, which by choice of $\hat \xi_i$ happens independently across all $i$ with probability $\Pr_{v_i}[v_i\geq \hat{\xi}_i]=\Pr_{b_i,s_i}[p_i-s_i\geq \hat{\xi}_i\wedge b_i\geq p_i]=\hat{q}_i$. 

Then for any $\Bv$, the set of active items is $R(\Bv)=\{j\in [n]:~ v_j\geq \hat{\xi}_j\}$. By $(\delta,\eta)$-selectability (Definition~\ref{def-selectablity}) and the fact that $\hat{q}\in \delta\cdot P_{\cF}$, we have

\begin{equation}\label{equ:lem-selectable1}
\Pr_{\pi,\Bv}[S\cup \{i\}\in \cF_{\pi,\hat{q}}, \forall S\subseteq R(\Bv), S\in\cF_{\pi,\hat{q}}]\geq \eta.    
\end{equation}

Note that for the sets $S\in\cF_{\pi,\hat{q}}$ that have $i\in S$, then $S\cup \{i\}\in \cF_{\pi,\hat{q}}$ with probability 1.  Thus, if we require $S \subseteq R(\Bv)\backslash\{i\}$ instead, it can not be that $i \in S$, and so the following LHS occurs with equal probability, allowing us to rewrite inequality~\eqref{equ:lem-selectable1} as follows: 

\begin{equation}\label{equ:lem-selectable2}
\Pr_{\pi,\Bv}[S\cup \{i\}\in \cF_{\pi,\hat{q}}, \forall S\subseteq R(\Bv)\backslash \{i\}, S\in\cF_{\pi,\hat{q}}]\geq \eta.
\end{equation}

For any $\Bv_{-i}$, let $R_i(\Bv_{-i})=\{j\neq i:~ v_j\geq \hat{\xi}_j\}$. Then inequality~\eqref{equ:lem-selectable2} is equivalent to

$$\Pr_{\pi,v_{-i}}[S\cup \{i\}\in \cF_{\pi,\hat{q}}, \forall S\subseteq R_i(\Bv_{-i}), S\in\cF_{\pi,\hat{q}}]\geq \eta.$$


Define event $A_i=\left\{\Bv_{-i}:~S\cup \{i\}\in \cF_{\pi,\hat{q}}, \forall S\subseteq R_i(\Bv_{-i}), S\in\cF_{\pi,\hat{q}}\right\}$. We will argue that item $i$ must be in the buyer's favorite bundle $S^*$ 
when both of the following conditions are satisfied: (i) $v_i\geq \hat{\xi}_i$, and (ii) event $A_i$ happens. Note that in $\MM$, the set of items in the market is $T=\{j:s_j\leq p_j-\hat{\xi}_j\}$, thus $S^*=\argmax_{S\subseteq T, S\in \cF_{\pi,\hat{q}}}\sum_{j\in S}(b_j-p_j)$. Suppose by way of contradiction that both conditions are satisfied but $i\not\in S^*$. Clearly, for every $j\in S^*$, we have $b_j\geq p_j$, otherwise removing $j$ from $S^*$ will give the buyer greater utility. In addition, we have $s_j\leq p_j-\hat{\xi}_j$, so $S^*\subseteq R(\Bv)$. By definition, $S^*$ must lie in $\cF_{\pi,\hat{q}}$. Since event $A_i$ occurs, then $S^*\cup \{i\}\in \cF_{\pi,\hat{q}}$. As $v_i\geq \hat{\xi}_i$, this implies that $b_i\geq p_i$. Thus adding $i$ to $S^*$ keeps the set feasible and does not decrease the buyer's utility $\sum_{j\in S^*}(b_j-p_j)$. Thus $i\in S^*$ (see footnote 6). 
This is a contradiction.    

Note that condition (i) and (ii) are independent. Thus for every $b_i$ and $s_i$ such that $b_i\geq p_i\wedge s_i\leq p_i-\hat{\xi}_i$ (or equivalently $v_i\geq\hat{\xi}_i$), the expected GFT of item $i$ over $b_{-i},s_{-i}$ is at least 
$$\Pr[A_i]\cdot (b_i-s_i)\geq \eta\cdot (p_i-s_i)=\eta\cdot v_i.$$

Thus

$$\gft(\MM)\geq \eta\cdot\sum_i\E_{v_i\sim V_i}[v_i\cdot \ind[v_i\geq \hat{\xi}_i]]\geq \delta\eta\cdot\sum_i\E_{v_i\sim V_i}[v_i\cdot \ind[v_i\geq \xi_i]],$$

where the last inequality is because for every $i$, we have $\E[v_i|v_i\geq \hat{\xi}_i]\geq \E[v_i|v_i\geq \xi_i]$ and $\Pr[v_i\geq \hat{\xi}_i]=\hat{q}_i=\delta\cdot \Pr[v_i\geq \xi_i]$.

For the second inequality stated in the lemma, note that 
$$\E_{\Bb,\Bs}\left[\max_{S\in \cF}\sum_{i\in S}\left\{(p_i-s_i)^+\cdot\ind[b_i\geq p_i]\right\}\right]=\E_{\Bv}\left[\max_{S\in \cF}\sum_{i\in S}v_i\right].$$
For every $\Bv$, let $\hat{S}(\Bv)=\argmax_{S\in \cF}\sum_{i\in S}v_i$, and break ties in favor of the set with smaller size. 
For every $i$, let $q_i=\Pr_{\Bv}[i\in \hat{S}(\Bv)]$ be the probability that $i$ is in the maximum weight independent set. We have that $\Bq=\{q_i\}_{i\in [n]}\in P_{\cF}$. 
Also for every $i$, $q_i=\Pr_{\Bv}[i\in \hat{S}(\Bv)]\leq \Pr[v_i>0]=\Pr[b_i\geq p_i>s_i]$.
Moreover,

$$\E_{\Bv}\left[\max_{S\in \cF}\sum_{i\in S}v_i\right]=\sum_{i\in [n]}\E_{\Bv}\left[v_i\cdot \ind[i\in \hat{S}(\Bv)]\right]\leq \sum_{i\in [n]}\E_{v_i\sim V_i}\left[v_i\cdot \ind[v_i\geq \xi_i]\right]$$

The inequality follows from the fact that for every $i$, both sides integrate random variable $v_i$ with a total probability mass $q_i$, while right hand side integrates $v_i$ at the top $q_i$-quantile. 
\end{prevproof}

For each $j$ in the summation, choose $p_i$ from Lemma~\ref{lem:feldman-selectable} to be $\theta_{ij}$ (or $\theta_{ij}'$). Then both terms $\circled{4}$ and  $\circled{6}$ are bounded by   
$\frac{\log(1/r)}{\delta\eta}\cdot \gcfpp$. Theorem~\ref{thm:log1/r} then follows directly from Lemmas~\ref{lem:breaking-terms-1},~\ref{lem:breaking-terms-2},~\ref{lem:buyer-surplus-main-body}, and \ref{lem:feldman-selectable}.

}

\notshow{

\begin{prevproof}{Lemma}{lem:term4&6}
For the first term, we notice that when $\overline{E}_{ij}$ happens, it must hold that $b_i\geq \overline{F_i}^{-1}(\frac{1}{2^{j}})$. Thus

\begin{align*}
    &\E_{\Bb,\Bs}\left[\sum_i(\theta_{ij}-s_i)^+\cdot\ind[i\in S^*(\Bb,\Bs)]\cdot \ind[\overline{E}_{ij}] \right]\\
    \leq&\E_{\Bb,\Bs}\left[\sum_i(\theta_{ij}-s_i)^+\cdot\ind[i\in S^*(\Bb,\Bs)]\cdot \ind[b_i\geq \overline{F_i}^{-1}(\frac{1}{2^j})]\right]\\
    \leq &\E_{\Bb,\Bs}\left[\max_{S\in \cF}\sum_{i\in S}\{(\theta_{ij}-s_i)^+\cdot\ind[b_i\geq \overline{F_i}^{-1}(\frac{1}{2^j})]\}\right]&\text{($S^*(\Bb,\Bs)\in \cF$)}\\
    \leq &\frac{1}{\delta\eta}\cdot \gcfpp&\text{(Lemma~\ref{lem:feldman-selectable} and $\theta_{ij}=\overline{F_i}^{-1}(\frac{1}{2^j}),\forall i$)}
\end{align*}

Similarly for the second term, note that when $\overline{E}_{ij}'$ happens, it must hold that $b_i\geq G_i^{-1}(\frac{1}{2^{j}})$. Thus
\begin{align*}
    &\E_{\Bb,\Bs}\left[\sum_i(\theta_{ij}'-s_i)^+\cdot\ind[i\in S^*(\Bb,\Bs)]\cdot \ind[\overline{E}_{ij}'] \right]\\
    \leq&\E_{\Bb,\Bs}\left[\sum_i(\theta_{ij}'-s_i)^+\cdot\ind[i\in S^*(\Bb,\Bs)]\cdot \ind[b_i\geq G_i^{-1}(\frac{1}{2^j})]\right]\\
    \leq &\E_{\Bb,\Bs}\left[\max_{S\in \cF}\sum_{i\in S}\{(\theta_{ij}'-s_i)^+\cdot\ind[b_i\geq G_i^{-1}(\frac{1}{2^j})]\}\right]&\text{($S^*(\Bb,\Bs)\in \cF$)}\\
    \leq &\frac{1}{\delta\eta}\cdot \gcfpp&\text{(Lemma~\ref{lem:feldman-selectable} and $\theta_{ij}'=G_i^{-1}(\frac{1}{2^j}),\forall i$)}
\end{align*}
\end{prevproof}

}

\subsection{Applications of Theorem~\ref{thm:log1/r} for Natural Constraints}\label{sec:appx-constraints}

\begin{definition}[Matroid Constraint]
A matroid is specified by a pair $(I,\cF)$, where $I$ is a finite ground set and $\cF\subseteq 2^I$ is a family of subsets of $I$. $(I,\cF)$ satisfies all of the following properties:

\begin{itemize}
    \item $\emptyset\in \cF$.
    \item $\cF$ is downward-closed: For every $S\in \cF$, we have $S'\in \cF, \forall S'\subseteq S$.
    \item $\cF$ has \emph{exchange property}: For every $S,S'\in \cF$ and $|S'|>|S|$, there exists $e\in S'\backslash S$ such that $S\cup \{e\}\in \cF$.
\end{itemize}

In the paper we say $\cF$ is a matroid constraint with respect to $I$ if $(I,\cF)$ forms a matroid.
\end{definition}

\begin{definition}[Matching, Knapsack Constraint]
Given an undirected graph $G=(V,E)$. $\cF\subseteq 2^E$ is a \emph{matching constraint} with respect to the ground set $E$ if $\cF=\{M\subseteq E: \text{$M$ is a matching in $G$}\}$.
A \emph{knapsack constraint} $\cF$ with respect to the ground set $I$ is defined as: $\cF=\{S\subseteq I:~\sum_{i\in S}c_i\leq 1\}$. Here $c_i\in [0,1]$ is the weight of element $i$. 
\end{definition}

Feldman et al.~\cite{FeldmanSZ16} prove that matroids, matching constraints and knapsack constraints are all $(\delta,\eta)$-selectable for some constant $\delta,\eta\in (0,1)$. Moreover, they prove that $(\delta,\eta)$-selectability has nice composability.

\begin{lemma}[Selectability of Natural Constraints]\label{lem:fsz1}\cite{FeldmanSZ16}
\begin{itemize}
\item For any matroid constraint $\cF$ and any $\delta\in (0,1)$, there exists a $(\delta,1-\delta)$-selectable greedy OCRS for $P_{\cF}$. Moreover, for any $\epsilon\in (0,1-\delta)$, there exists a $(\delta,1-\delta-\epsilon)$-selectable greedy OCRS $\pi$ for $P_{\cF}$, and the running time of $\pi$ is polynomial on the input size and $1/\epsilon$.
\item For any matching constraint $\cF$ and any $\delta\in (0,1)$, there exists an efficient $(\delta,e^{-2\delta})$-selectable greedy OCRS for $P_{\cF}$. 
\item For any knapsack constraint $\cF$ and any $\delta\in (0,\frac{1}{2})$, there exists an efficient $(\delta,\frac{1-2\delta}{2-2\delta})$-selectable greedy OCRS for $P_{\cF}$.
\end{itemize}
\end{lemma}


\begin{lemma}[Composability of Selectability]\label{lem:fsz2}
\cite{FeldmanSZ16} Given two downward-closed constraints $\cF_1$ and $\cF_2$ with respect to the same ground set $I$. Let $\cF=\cF_1\cap \cF_2$. Suppose there exist a $(\delta,\eta_1)$-selectable greedy OCRS $\pi_1$ for $P_1$ and $\cF_1$, and a $(\delta,\eta_2)$-selectable greedy OCRS $\pi_2$ for $P_2$ and $\cF_2$. Then there exists a $(\delta,\eta_1\cdot\eta_2)$-selectable greedy OCRS $\pi$ for $P_{1}\cap P_{ 2}$ and $\cF$. When $P_1=P_{\cF_1}$ and $P_2 = P_{\cF_2}$, as $P_{\cF_1}\cap P_{\cF_2}\subseteq P_{\cF}$, $\pi$ is also $(\delta,\eta_1\cdot\eta_2)$-selectable for $P_{\cF}$ and $\cF$. Moreover, $\pi$ is efficient computable given $\pi_1$ and $\pi_2$.
\end{lemma}

\begin{corollary}\label{cor:log1/r-natural constraints}
Suppose the buyer's feasibility constraint is $\cF=\bigcap_{t=1}^d \cF_t$ for some constant $d$, where each $\cF_t$ is a matroid, matching constraint, or knapsack constraint. Then $\fb\leq O(\log(\frac{1}{r}))\cdot \gcfpp$.
\end{corollary}

\begin{prevproof}{Corollary}{cor:log1/r-natural constraints}
Pick any constant $\delta\in (0,\frac{1}{2})$. By Lemma~\ref{lem:fsz1} and~\ref{lem:fsz2}, there exists some constant $\eta\in (0,1)$ such that there exists an efficient, $(\delta,\eta)$-selectable greedy OCRS for $P_{\cF}$. Then the result follows from Theorem~\ref{thm:log1/r}.
\end{prevproof}

\subsection{Computing the Approximately-Optimal Mechanism Efficiently}

By Lemma~\ref{lem:term3&5}, the fixed posted price mechanisms to bound term \circled{3} and \circled{5} are efficiently computable. In Lemma~\ref{lem:feldman-selectable}, the buyer posted prices are chosen as $p_i=\theta_{ij}$ (or $\theta_{ij}'$), which can be computed efficiently. In order to find seller posted prices, it's sufficient to find the optimal $q$ in Lemma~\ref{lem:feldman-selectable} that maximizes $\sum_i\E_{b_i,s_i}\left[v_i\cdot \ind[v_i\geq \xi_i]\right]$. For every $i$ and $q_i\leq \Pr[b_i\geq p_i>s_i]$, let $$h_i(q_i)=\E_{b_i,s_i}[v_i\cdot \ind[v_i\geq \xi_i]]=\Pr[b_i\geq p_i]\cdot \int_{0}^{p_i-\xi_i}(p_i-s_i)dG_i(s_i).$$

Then it's equivalent to solve the following maximization problem over the polytope $P_{\cF}$:
\begin{align*}
    \text{max}~~~&\sum_ih_i(q_i)\\
    \text{s.t.}~~~& q\in P_{\cF}\\
               & q_i\leq \Pr[b_i\geq p_i]\cdot \Pr[s_i<p_i], \forall i\in [n]. 
\end{align*}

{
Observe that every $h_i$ is a concave function. In many settings one can efficiently obtain a near-optimal solution using convex optimization techniques.
}

Moreover, the subconstraint $\cF'$ can be efficiently computed if there exists an efficient greedy OCRS for $P_{\cF}$. For the constraints in Corollary~\ref{cor:log1/r-natural constraints}, by Lemma~\ref{lem:fsz1} and~\ref{lem:fsz2}, the subconstraint can be computed efficiently.
\section{Missing Details from 
Section~\ref{sec:log n approx}}\label{apx:logn}

\label{apx:UB of SB}

In this section, we use the notions of the ``super seller auction'' and ``super buyer procurement auction'' from Section~\ref{sec:UB of SB}.

We claim that the GFT of any IR, BIC, ex-ante WBB mechanism $\MM=(x,p^B,p^S)$ is upper bounded by $\opts+\optb$. The proof is adapted from~\cite{BrustleCWZ17}.

\begin{lemma}\label{lem:ub-opt}\cite{BrustleCWZ17}
$\opt\leq \opts+\optb.$
\end{lemma}
\begin{proof}
Take any BIC, IR, ex-ante WBB mechanism $\MM=(x,p^B,p^S)$. Since every seller $i$ is BIC and IR, we have for any $s_i,s_i'$, 
$$\E_{\Bb,s_{-i}}\left[p_i^S(\Bb,\Bs)-s_i\cdot x_i(\Bb,\Bs)\right]\geq \max\left\{\E_{\Bb,s_{-i}}[p_i^S(\Bb,s_i',s_{-i})]-s_i\cdot x_i(\Bb,s_i',s_{-i}), 0\right\}$$

Observe that $\MM'=(x,p^S)$ is a valid super buyer procurement auction. The above inequalities are exactly the BIC and IR constraints for seller $i$. Thus $\MM'$ is BIC and IR. Similarly, $\MM''=(x,p^B)$ is BIC and IR, so it is a valid super seller auction. Since $\MM$ is ex-ante WBB, $\E_{\Bb,\Bs}[p^B(\Bb,\Bs)-\sum_ip^S(\Bb,\Bs)]\geq 0$. Thus we have
\begin{align*}
\gft(\MM)&=\E_{\Bb,\Bs}\left[\sum_{i\in[n]} x_i(\Bb,\Bs)(b_i-s_i)\right]\\
&\leq \E_{\Bb,\Bs}\left[p^B(\Bb,\Bs)-\sum_{i\in[n]} x_i(\Bb,\Bs)\cdot s_i\right]+\E_{\Bb,\Bs}\left[\sum_{i\in [n]} (x_i(\Bb,\Bs)\cdot b_i-p_i^S(\Bb,\Bs))\right]\\
&\leq \opts+\optb
\end{align*}

Taking $\MM$ to be the GFT-maximizing mechanism completes the proof.
\end{proof}


We next prove an analog of the ``Marginal Mechanism Lemma''~\cite{CaiH13,HartN12} for the optimal profit. Namely, let $(T,R)$ be a partition of the items in $[n]$, then the optimal profit in a super seller auction with items in $[n]$ is upper bounded by the first-best GFT for items in $T$ plus the optimal profit in a super seller auction with items in $R$. 

\begin{lemma}[Marginal Mechanism for Profit]\label{lem:separation}
For any subset $T\in [n]$, we let $\cF\big|_T=\{S\subseteq T:S\in\cF\}$ denote the restriction of $\cF$ to $T$. We use $\fb(T,\cF\big|_T)$ to denote the first-best GFT obtainable between sellers in $T$ and the $\cF\big|_T$-constrained additive buyer, that is, $$\fb(T,\cF\big|_T)=\E_{\Bb_T,\Bs_T}\left[\max_{S\in \cF|_T}\sum_{i\in S}(b_i-s_i)^+\right],$$ where $\Bb_T=\{b_i\}_{i\in T}$, $\Bs_T=\{s_i\}_{i\in T}$. 
 Let $(R,T)$ be any partition of the items in $[n]$. Then $$\opts([n],\cF)\leq \opts(R,\cF\big|_R)+\fb(T,\cF\big|_T).$$
\end{lemma}

\begin{proof}
Consider the optimal BIC and IR mechanism $\MM=(x,p)$ in the super seller auction with item set $[n]$. We will construct a BIC and IR mechanism $\MM'=(x',p')$ in the super seller auction with item set $R$ as follows. The mechanism only sells items in $R$ using the same allocation $x$. The payment for the buyer is defined as the payment $p$ in $\MM$ minus the buyer's expected total value for all items in $T$. Formally, for every $\Bb_R=\{b_j\}_{j\in R}$, $\Bs_R=\{s_j\}_{j\in R}$ and $i\in R$, let $$x_i'(\Bb_R,\Bs_R)=\E_{\Bb_{T},\Bs_{T}}[x_i(\Bb,\Bs)]$$ $$p'(\Bb_R,\Bs_R)=\E_{\Bb_{T},\Bs_{T}}\left[p(\Bb,\Bs)-\sum_{j\in T}b_j\cdot x_j(\Bb,\Bs)\right].$$ 

Notice that in $\MM'$, the expected utility of the buyer with type $\Bb_R$ when reporting $\Bb_R'$ is
$$\E_{\Bs_R}\left[\sum_{i\in T}b_i\cdot x_i'(\Bb_R',\Bs_R)-p'(\Bb_R',\Bs_R)\right] =\E_{\Bb_{T},\Bs}\left[\sum_{i\in [n]}b_i\cdot x_i(\Bb'_R,\Bb_T,\Bs)-p(\Bb'_R,\Bb_T,\Bs)\right],$$


Since $\MM$ is BIC and IR, $\MM'$ is also BIC and IR. Thus
\begin{align*}
\opts([n],\cF)=&\E_{\Bb,\Bs}\left[p(\Bb,\Bs)-\sum_{i\in [n]}s_i\cdot x_i(\Bb,\Bs)\right]\\
=&\E_{\Bb_R,\Bs_R}\left[p'(\Bb_R,\Bs_R)-\sum_{i\in R}s_i\cdot x_i'(b_R,s_R)\right]+\E_{\Bb,\Bs}\left[\sum_{i\in T}(b_i-s_i)\cdot x_i(\Bb,\Bs)\right]\\
\leq& \opts(R,\cF\big|_R)+\fb(T,\cF\big|_T).
\end{align*}
\end{proof}

We partition the items into the set of ``likely to trade'' items, that is, items with trade probability $r_i=\Pr_{b_i,s_i}[b_i\geq s_i]\geq 1/n$, and the ``unlikely to trade'' items. We can bound the $\opts$ by the first-best GFT of the ``likely to trade'' items and the optimal profit of the super seller auction with the ``unlikely to trade'' items. We can further replace the first-best GFT of the ``likely to trade'' by  $O(\log n)\cdot \gcfpp$ according to Theorem~\ref{thm:log1/r} or by  $O(\log^2(n))\cdot \gcfpp$ according to Theorem~\ref{thm:downward-close} depending on the buyer's feasibility constraint. Formally,
\begin{lemma}\label{lem:likely trade and unlikely trade}
Define $H=\{i\in [n]:r_i\geq \frac{1}{n}\}$ and $L=[n]\backslash H=\{i\in [n]:r_i<\frac{1}{n}\}$. Suppose the buyer's feasibility constraint $\cF$ is 
$(\delta,\eta)$-selectable for some $\delta,\eta\in (0,1)$. Then 
\begin{align*}
    \opt&\leq \optb+\opts(L,\cF\big|_L)+\fb(H,\cF\big|_H)\\
    &\leq \optb+ \opts(L,\cF\big|_L)+O\left(\frac{\log n}{\delta\cdot \eta}\right)\cdot \gcfpp.
\end{align*}
For general constrained-additive buyer,
$$\opt\leq \optb+ \opts(L,\cF\big|_L)+O\left(\log^2( n)\right)\cdot \gcfpp.$$
\end{lemma}
\begin{proof}
The first inequality follows from Lemma~\ref{lem:ub-opt} and~\ref{lem:separation}. Since $\cF$ is 
$(\delta,\eta)$-selectable, $\cF\big|_H$ is also 
$(\delta,\eta)$-selectable. We derive the second inequality by applying Theorem~\ref{thm:log1/r} on the items in $H$. For a general constrained-additive buyer, we derive the inequality by applying Theorem~\ref{thm:downward-close} on the items in $H$.
\end{proof}

\begin{proof}[Proof of Lemma~\ref{lem:UB of second best GFT}]
It directly follows from Lemmas~\ref{lem:likely trade and unlikely trade} and~\ref{lem:unlikley trade item upper bound}.
\end{proof}

\notshow{

Next, we provide an upper bound of the optimal super seller profit from items in $L$. It is well known that in multi-item auctions the revenue of selling the items separately is a $O(\log n)$-approximation to the optimal revenue when there is a single additive buyer~\cite{LiY13}. Cai and Zhao~\cite{CaiZ19} provide a extension of this $O(\log n)$-approximation to profit maximization. Combining this approximation with some basic observations based on the Cai-Devanur-Weinberg duality framework~\cite{CaiDW16}, we derive the following upper bound of $\opts(L,\cF\big|_L)$.
\begin{lemma}\label{lem:unlikley trade item upper bound}
$$\opts(L,\cF\big|_L)\leq O\left(\log(|L|)\cdot \sum_{i\in L}\E_{b_i,s_i}\left[(\tilde{\varphi}_i(b_i)-s_i)^+\right]\right).$$
Here $\widetilde{\varphi}_i(b_i)$ is Myerson's ironed virtual value function\footnote{The buyer's unironed virtual value function is $\varphi_i(b_i) = b_i - \frac{1-F_i(b_i)}{f_i(b_i)}$, and then these are averaged to be made monotonic in quantile space to create $\widetilde{\varphi}_i(b_i)$.}  for the buyer's distribution for item $i$, $\DD_i^B$.
\end{lemma}

Section~\ref{sec:UB of unlikely to trade items} is dedicated to its proof. Combining Lemmas~\ref{lem:likely trade and unlikely trade} and~\ref{lem:unlikley trade item upper bound} gives the proof of Lemma~\ref{lem:UB of second best GFT}, stated in Section~\ref{sec:UB of SB}.

}

\notshow{

\subsection{Bounding the Optimal Buyer Utility in the Super Buyer Procurement Auction}\label{sec:bounding super buyer}
In this section, we construct a two-sided market to bound $\optb$ for any constrained additive buyer. 




\begin{lemma}\label{lem:mechanism for OPT-B}
Consider the mechanism $\MM^*=(x,p^B,p^S)$ where for every item $i$, buyer profile $\Bb$, and seller profile $\Bs$, $$x_i(\Bb,\Bs)=\ind[b_i-\tilde{\tau}_i(s_i)\geq 0\wedge i\in \argmax_{S\in \cF}\sum_{i\in S}(b_i-\widetilde{\tau}_i(s_i))^+].$$ 
Here $\widetilde{\tau}_i(s_i)$ is Myerson's ironed virtual value function\footnote{The seller's unironed virtual value function is $\tau_i(s_i) = s_i + \frac{G_i(s_i)}{g_i(s_i)}$.} for seller $i$'s distribution $\DD_i^S$. For every seller $i$, since $\tilde{\tau}_i(s_i)$ is non-decreasing in $s_i$, $x_i(\Bb,\Bs)$ is non-increasing in $s_i$. Define $p^S_i(\Bb,\Bs)$ as the threshold payment for seller $i$, i.e., the largest cost $s_i$ such that $x_i(\Bb,s_i,s_{-i})=1$. Define the buyer's payment $p^B(\Bb,\Bs)=\sum_i x_i(\Bb,\Bs)\cdot \tilde{\tau}_i(s_i)$. $\MM^*$ is DSIC, ex-post IR, ex-ante SBB~\footnote{One can make the mechanism IR and ex-post SBB by defining $p^B(\Bb,\Bs)=\sum_i p_i^S(\Bb,\Bs)$. The mechanism is still DSIC for all sellers. But it's only BIC for the buyer, as the seller's gains equals to the virtual welfare only when taking expectation over sellers' profile.} and  
$$\gft(\MM^*)\geq \optb=\E_{\Bb,\Bs}[\max_{S\in \cF}\sum_{i\in S}(b_i-\widetilde{\tau}_j(s_j))^+].$$ 
\end{lemma}
\begin{proof}

Since the seller's allocation rule is monotone and we use the threshold payment, $\MM^*$ is DSIC and ex-post IR for each seller. 

Note that for any seller profile $\Bs$, when the buyer's has true type $\Bb$, her expected utility by reporting $\Bb'$ is $\sum_i x_i(\Bb',\Bs)\cdot (b_i-\widetilde{\tau}_i(s_i))$.
According to the definition of $x$, the buyer's utility is maximized when $\Bb'=\Bb$. Hence, $\MM$ is DSIC for the buyer. Moreover we have ex-post IR, as the buyer's expected utility when reporting truthfully is  $\max_{S\in \cF}\sum_{i\in S}(b_i-\widetilde{\tau}_i(s_i))^+\geq 0$.

It only remains to prove that the mechanism is ex-ante SBB and to lower bound its GFT. By Myerson's lemma\footnote{This lemma is used several times, and is formally stated as Lemma~\ref{lem:mye} in Appendix~\ref{apx:md}.} (Lemma~\ref{lem:mye}), for every $\Bb$ we have
$$\E_{\Bs}\left[\sum_i p_i^S(\Bb,\Bs)\right]=\E_{\Bs}\left[\sum_i x_i(\Bb,\Bs)\cdot\widetilde{\tau}_i(s_i)\right]=\E_{\Bs}[p^B(\Bb,\Bs)].$$

Thus the mechanism is ex-ante SBB.

Why is $\optb = \E_{\Bb,\Bs}[\max_{S\in \cF}\sum_{i\in S}(b_i-\widetilde{\tau}_j(s_j))^+]$? Notice that only the sellers are strategic in a super buyer procurement auction, and their types are all single-dimensional. One can apply the standard Myersonian analysis to the super buyer procurement auction and show that the optimal buyer utility is exactly $\E_{\Bb,\Bs}[\max_{S\in \cF}\sum_{i\in S}(b_i-\widetilde{\tau}_j(s_j))^+]$.

Note that the buyer's expected utility in $\MM^*$ is exactly $\optb$. As $\MM^*$ is an ex-ante SBB mechanism, the expected GFT of $\MM^*$ is equal to the buyer's expected utility plus the sum of all seller's expected utility, and the latter is non-negative since $\MM^*$ is ex-post IR for every seller.
\end{proof}


}

\notshow{
\subsection{Missing Details on the Seller Adjusted Posted Price Mechanism}\label{apx:SAPP}

In this section, we first argue that any \Msapp~mechanism with bi-monotonic posted prices is DSIC and ex-post IR in Lemma~\ref{sapp:monotone}.  Then, in Lemma~\ref{sapp:main}, we showed how to use a class of allocation rules to set posted prices for the buyer such that the corresponding SAPP mechanism is DSIC, IR, ex-ante WBB.  In Lemma~\ref{lem:lb on sapp gft}, we argue that it obtains good GFT.  Finally, in Lemma~\ref{lem:SAPP bounding unlikely trade items}, we use these lemmas to show that the optimal SAPP mechanism on the ``unlikely to trade'' items obtains good GFT.

\begin{lemma}\label{sapp:monotone}
Let $\MM$ be a \Msapp~mechanism with bi-monotonic posted prices $\{\theta_i(\Bs)\}_{i\in[n]}$. Then the allocation of the mechanism $\hat{x}_i(\Bb,\Bs)$ is non-increasing in $s_i$ for all sellers $i$, and $\MM$ is DSIC and ex-post IR for the buyer and the sellers.
\end{lemma}

\begin{proof}
Notice that for every type $\Bb$, the buyer chooses the item that maximizes $b_i-\theta_i(\Bs)$ (and does not choose any item if she cannot afford any of the items). For every $i$, by bi-monotonicity, when $s_i$ decreases, $b_i-\theta_i(\Bs)$ does not decrease while $b_j-\theta_j(\Bs)$ does not increase for all $j\not=i$. Thus if the buyer chooses item $i$ under the original $s_i$, she must continue to choose item $i$ for smaller reports $s_i'$. Thus $\hat{x}_i(\Bb,\Bs)$ is non-increasing in $s_i$. 
Since every seller receives threshold payment, she is DSIC and ex-post IR. As the buyer simply faces a posted price mechanism, the mechanism is DSIC and ex-post IR for the buyer.
\end{proof}

We proved that the bi-monotonic posted prices induce a monotone allocation rule for every seller. Thus, in the corresponding \Msapp~mechanism, every seller receives \emph{threshold payment}~\cite{Myerson81,MyersonS83}.
Formally, for every seller $i$ let $\hat{x}_i(\Bb,\Bs)$ denote the probability that the buyer trades with seller $i$ under profile $(\Bb,\Bs)$. This is either 0 or 1 since all $\theta_i(\Bs)$s are fixed value when $\Bs$ is fixed. If $\hat{x}_i(\Bb,\Bs)=1$, $p_i^S(\Bb,\Bs)$ is defined as the maximum value $s_i'$ such that $\hat{x}_i(\Bb,s_i',s_{-i})=1$. Otherwise $p_i^S(\Bb,\Bs)=0$. Hence any \Msapp~allocation rule is DSIC and ex-post IR.

Now we have the proof of Lemma~\ref{lem:sapp-main1}, which is used in the proof of Lemma~\ref{sapp:main}.

\begin{proof}[Proof of Lemma~\ref{lem:sapp-main1}]
Note that the buyer will purchase item $i$ if both of the following conditions are satisfied: 
\begin{enumerate}
    \item The buyer can afford item $i$, i.e., $b_i\geq \theta_i(\Bs)$.
    \item The buyer can not afford any other items, i.e., $b_j<\theta_j(\Bs),\forall j\not=i$.
\end{enumerate}

By choice of $\theta_i(\Bs)$, the first event happens with probability $\Pr[b_i\geq \theta_i(\Bs)]=q_i(\Bs)/2$. 

Note that $\sum_{i\in[n]} q_i(\Bs)\leq \E_{\Bb}[\sum_{i\in [n]} x_i(\Bb,\Bs)]\leq 1$. For each $j\neq i$, $\Pr[b_j<\theta_j(\Bs)]=1-\frac{q_j(\Bs)}{2}$. Thus $\sum_{j\neq i} \left(1-\frac{q_j(\Bs)}{2}\right)\geq n-\frac{3}{2}+\frac{q_i(\Bs)}{2}$.
The second event happens with probability $$\prod_{j\neq i} \left(1-\frac{q_j(\Bs)}{2}\right)\geq \frac{1}{2}+\frac{q_i(\Bs)}{2}.$$ The equality holds when one out of the $n-1$ $q_j(\Bs)$'s equals  
$1-q_i(\Bs)$ and the rest all equal to $0$.
Notice that the two events are independent, so we have the upper and lower bound on $\hat{x}_i(\Bs)$.
\end{proof}

\notshow{

We proceed to lower bound the GFT from the mechanism defined in Section~\ref{sec:SAPP}.

\begin{lemma}\label{lem:lb on sapp gft}
Let $\MM$ be the mechanism defined in Lemma~\ref{sapp:main}.  Then
$$\gft(\MM)\geq \frac{1}{4}\E_{\Bb,\Bs}\left[\sum_i (\widetilde{\varphi}_i(b_i)-s_i)\cdot x_i(\Bb,\Bs)\right].$$
\end{lemma}

\begin{proof}
\begin{align*}
    \gft(\MM)=&\E_{\Bb,\Bs}\left[\sum_i(b_i-s_i)\cdot \hat{x}_i(\Bb,\Bs)\right]\\
    \geq&\E_{\Bs}\left[\sum_i(\theta_i(\Bs)-s_i)\cdot \hat{x}_i(\Bs)\right]  & \text{($\hat{x}_i(\Bb,\Bs)=0$ if $b_i<\theta_i(\Bs)$)}\\
    \geq &\frac{1}{2}\E_{\Bs}\left[\sum_i\left(F_i^{-1}\left(1-\frac{q_i(\Bs)}{2}\right)-s_i\right)\cdot \frac{q_i(\Bs)}{2}\right]   & \text{(Definition of  $\theta_i(\Bs)$, $q_i(\Bs)$ and Lemma~\ref{lem:sapp-main1})}\\
    =&\frac{1}{2}\E_{\Bb,\Bs}\left[\sum_i(\widetilde{\varphi}_i(b_i)-s_i)\cdot\ind\left[b_i\geq F_i^{-1}\left(1-\frac{q_i(\Bs)}{2}\right)\right]\right]  & \text{(Myerson's Lemma (\ref{lem:mye}))}\\
    \geq& \frac{1}{4}\E_{\Bb,\Bs}\left[\sum_i(\widetilde{\varphi}_i(b_i)-s_i)\cdot x_i(\Bb,\Bs)\cdot \ind[\tilde{\varphi}_i(b_i)\geq s_i]]]\right]\\
    \geq& \frac{1}{4}\E_{\Bb,\Bs}\left[\sum_i(\widetilde{\varphi}_i(b_i)-s_i)\cdot x_i(\Bb,\Bs)]\right]   
\end{align*}

Here the second-to-last inequality uses the fact that $$\E_{b_i}\left[\widetilde{\varphi}_i(b_i)\cdot\ind[b_i\geq F_i^{-1}\left(1-\frac{q_i(\Bs)}{2}\right)\right]\geq \frac{1}{2}\cdot\E_{\Bb}\left[\widetilde{\varphi}_i(b_i)\cdot x_i(\Bb,\Bs)\cdot \ind[\tilde{\varphi}_i(b_i)\geq s_i]]]\right]$$
holds for every $\Bs$ and $i$. This is because the right hand side $$\frac{1}{2}\cdot\E_{\Bb}[\widetilde{\varphi}_i(b_i)\cdot x_i(\Bb,\Bs) \cdot \ind[\tilde{\varphi}_i(b_i)\geq s_i]]=\E_{b_i}\left[\widetilde{\varphi}_i(b_i)\cdot \frac{1}{2}\E_{b_{-i}}[x_i(\Bb,\Bs) \cdot \ind[\tilde{\varphi}_i(b_i)\geq s_i]]\right]$$ can be viewed as the expectation of $\widetilde{\varphi}_i(b_i)$ on an event of $b_i$ with a total probability mass $$\E_{b_i}\left[\frac{1}{2}\E_{b_{-i}}[x_i(\Bb,\Bs) \cdot \ind\left[\tilde{\varphi}_i(b_i)\geq s_i]\right]\right]=\frac{q_i(\Bs)}{2},$$ while the left hand side is the maximum expectation of $\widetilde{\varphi}_i(b_i)$ on any event of $b_i$ with  total probability mass $\frac{q_i(\Bs)}{2}$, as $\widetilde{\varphi}_i(b_i)$ is non-decreasing on $b_i$. 
\end{proof}

Finally, we prove Lemma~\ref{lem:SAPP bounding unlikely trade items} which shows that there exists a DSIC, ex-post IR, and ex-ante WBB \Msapp~mechanism over items in $L$ whose GFT is $\Theta\left(\sum_{i\in L}\E_{b_i,s_i}[(\widetilde{\varphi}_i(b_i)-s_i)^+]\right)$.  $\E_{b_i,s_i}[(\widetilde{\varphi}_i(b_i)-s_i)^+]$ can be viewed as the optimal profit that seller $i$ can achieve by selling her item to the buyer. In other words, we need to compare to the sum of the optimal profit for selling all items in $L$. However, in a \Msapp~mechanism, the buyer is only allowed to purchase up to one item. Since trade happens less frequently, it is not clear how we can use the GFT of a \Msapp~mechanism to approximate the total profit. Notice, however, we are not dealing with arbitrary distributions. The items in $L$ are unlikely to trade. Indeed, with constant probability, there is at most one item that can be traded. In other words, with constant probability, the restriction that the buyer can purchase no more than one item does not affect the number of trades. 

\begin{lemma}\label{lem:SAPP bounding unlikely trade items}
We let  $\sapp(S)$ denote the optimal GFT attainable by any DSIC, ex-post IR, and ex-ante WBB \Msapp~mechanisms over items in $S$ for any subset $S\subseteq [n]$. $\sapp(L)\geq \frac{1}{4e}\cdot \sum_{i\in L}\E_{b_i,s_i}[(\widetilde{\varphi}_i(b_i)-s_i)^+]$.
\end{lemma}

\begin{proof}
Let $\Bb_L=\{b_i\}_{i\in L}$ and $\Bs_L=\{s_i\}_{i\in L}$. For every $i\in L$, define the event that only $i$ is tradeable: $$A_i=\left\{(\Bb_L,\Bs_L):b_i\geq s_i\wedge b_j<s_j,\forall j\in L\backslash\{i\}\right\}.$$ We consider the following allocation rule:
\begin{align*}
x_i(\Bb_L,\Bs_L)=
\begin{cases}
\ind[\widetilde{\varphi}_i(b_i)\geq s_i] &,~ \text{if }  (\Bb,\Bs)\in A_i\\
0 &,~\text{otherwise}
\end{cases}
\end{align*}

Notice that $(\Bb_L,\Bs_L)\in A_i$ implies that $(\Bb_L,s_i',\Bs_{L\backslash\{i\}})\in A_i$ for any $s_i'\leq s_i$. Thus, $x_i(\Bb_L,\Bs_L)$ is non-increasing in $s_i$. Similarly, it is easy to verify that $x_i(\Bb_L,\Bs_L)$ is non-decreasing in all $s_j$ where $j\in L\backslash\{i\}$. Furthermore, $\sum_{i\in L} x_i(\Bb_L,\Bs_L)\leq 1$ for all $\Bb_L,\Bs_L$. If we choose the posted prices according Lemma~\ref{sapp:main}, then the corresponding mechanism has GFT that is at least $\frac{1}{4}\E_{\Bb,\Bs}\left[\sum_i (\widetilde{\varphi}_i(b_i)-s_i)\cdot x_i(\Bb,\Bs)\right]$.

Moreover, by the definition of $x_i(\Bb,\Bs)$, 
\begin{align*}
\E_{\Bb,\Bs}\left[\sum_{i\in L} (\widetilde{\varphi}_i(b_i)-s_i)\cdot x_i(\Bb,\Bs)\right]=&\sum_{i\in L} \E_{b_i,s_i}[(\widetilde{\varphi}_i(b_i)-s_i)^+)]\cdot \prod_{j\in L\backslash\{i\}}\Pr_{b_j,s_j}[b_j<s_j]\\
\geq&\sum_{i\in L} \E_{b_i,s_i}[(\widetilde{\varphi}_i(b_i)-s_i)^+)]\cdot (1-\frac{1}{n})^{|L|}\\
\geq& \frac{1}{e}\cdot\sum_{i\in L}\E_{b_i,s_i}[(\widetilde{\varphi}_i(b_i)-s_i)^+)]
\end{align*}
The first inequality is because for each item $j\in L$, $\Pr_{b_j,s_j}[b_j<s_j]\geq 1-1/n$. Hence, $$\sapp(L)\geq \frac{1}{4e}\cdot \sum_{i\in L}\E_{b_i,s_i}[(\widetilde{\varphi}_i(b_i)-s_i)^+].$$
\end{proof}

}

\notshow{
}

\subsection{Bounding the Optimal Profit from the Unlikely to Trade Items}\label{sec:UB of unlikely to trade items}

In this section, we present the proof of Lemma~\ref{lem:unlikley trade item upper bound}.  To bound $\opts(L,\cF\big|_L)$, we need the following result from \cite{CaiZ19}. It provides a benchmark of the optimal profit using the Cai-Devanur-Weinberg duality framework~\cite{CaiDW16}: The profit of any BIC, IR mechanism is upper bounded by the buyer's virtual welfare with respect to some virtual value function, minus the sellers' total cost for the same allocation.

A sketch of the framework is as follows: we first formulate the profit maximization problem as an LP. Then we lagrangify the BIC and IR constraints to get a partial Lagrangian dual of the LP. Since the buyer’s payment is unconstrained in the partial Lagrangian, one can argue that the corresponding dual variables must form a flow for the benchmark to be finite. By weak duality, any choice of the dual variables/flow derives a benchmark for the optimal profit. In \cite{CaiZ19}, they also construct a canonical flow and prove that there exists a BIC and IR mechanism whose profit is within a constant factor times the benchmark w.r.t. the flow for any single constrained-additive buyer.     
\begin{lemma}\label{lem:mz19}
\cite{CaiZ19} For any $T\subseteq [n]$ and feasibility constraint $\cJ$ with respect to $T$, consider the super seller auction with item set $T$ and to a $\cJ$-constrained buyer. Any flow $\lambda_T$ induces a finite benchmark for the optimal profit, that is,

$$\opts(T,\cJ)\leq \max_{x\in P_{\cJ}}\E_{\Bb,\Bs}\left[\sum_{i\in T} x_i(\Bb,\Bs)\cdot (\Phi_i^T(\Bb)-s_i)\right]$$

where

$$\Phi_i^{T}(\Bb)=b_i-\frac{1}{f_i(b_i)}\sum_{\Bb'}\lambda_T(\Bb',\Bb)\cdot (b_i'-b_i)$$
can be viewed as buyer $i$'s virtual value function, and $P_{\cJ}$ is the set of all feasible allocation rules. More specifically, $\lambda_T(\Bb',\Bb)$ is the Lagrangian multiplier for the BIC/IR constraint that says when the buyer has true type $\Bb$ she does not want to misreport $\Bb'$. The equality sign is achieved when the optimal dual $\lambda_T^*$ is chosen.


\end{lemma}

\vspace{.2in}

Next, we show that $\opts(L,\cF\big|_L)$ is no more than $\opts(L,\text{ADD})$ using Lemma~\ref{lem:mz19}.

\begin{lemma}\label{lem:compare-to-additive}
$\opts(L,\cF\big|_L)\leq \opts(L,\text{ADD}).$
\end{lemma}
\begin{proof}
Let $\hat{\lambda}_L$ be the optimal dual in Lemma~\ref{lem:mz19} when the buyer is additive without any feasibility constraint, and $\hat{\Phi}_i^L(\cdot)$ be the induced virtual value function. We have that 
\begin{align*}
\opts(L,\cF\big|_L)&\leq \max_{x\in P_{\cF|_L}}\E_{\Bb,\Bs}\left[\sum_{i\in L}x_i(\Bb,\Bs)\cdot (\hat{\Phi}_i^L(\Bb)-s_i)\right]\\
&\leq \E_{\Bb,\Bs}\left[\sum_{i \in L}(\hat{\Phi}_i^L(\Bb)-s_i)^+\right]\\
&=\max_{x_i(\Bb,\Bs)\in [0,1]}\E_{\Bb,\Bs}\left[\sum_ix_i(\Bb,\Bs)\cdot (\hat{\Phi}_i^L(\Bb)-s_i)\right]\\
&=\opts(L,\text{ADD}). 
\end{align*}
\end{proof}




Cai and Zhao~\cite{CaiZ19} also give a logarithmic upper bound of the optimal profit for a single additive buyer, using the sum of optimal profit for each individual item.
\begin{lemma}\label{lem:mz19-additive}
\cite{CaiZ19} $$\opts(L,\text{ADD})\leq\log(|L|)\cdot \sum_{i\in L}\opts(\{i\})= \log(|L|)\cdot \sum_{i\in L}\E_{b_i,s_i}[(\varphi_i(b_i)-s_i)^+].$$
\end{lemma}

Together, Lemmas~\ref{lem:compare-to-additive} and \ref{lem:mz19-additive} conclude the proof of Lemma~\ref{lem:unlikley trade item upper bound}:
$$\opts(L,\cF\big|_L)\leq O\left(\log(|L|)\cdot \sum_{i\in L}\E_{b_i,s_i}\left[(\tilde{\varphi}_i(b_i)-s_i)^+\right]\right).$$


}

\section{Missing Details from Section~\ref{sec:reduction}}\label{appx:reduction}

\notshow{
\begin{lemma}\label{lem:sapp-matroid-rank}
Suppose the buyer's feasibility constraint $\cF$ is a matroid. Then
$$\sapp\geq \E_{\Bb,\Bs}\left[\max_{S\in \cF}\sum_{i\in S}(\tilde{\varphi}_i(b_i)-s_i))\right].$$
\end{lemma}

\begin{proof}
We construct the following allocation rule $x=\{x_i(\Bb,\Bs)\}_{i\in [n]}$. For every $\Bb,\Bs$, let $S^*(\Bb,\Bs)=\argmax_{S\in \cF}\sum_{i\in S}(\tilde{\varphi}_i(b_i)-s_i).$ $S^*(\Bb,\Bs)\in \cF$. For every $i$ and $\Bb,\Bs$, let $x_i(\Bb,\Bs)=\ind[i\in S^*(\Bb,\Bs)]$. Let $q_i(\Bs)=\E_{\Bb}[x_i(\Bb,\Bs)]$. Then for every $\Bs$, the vector $\{q_i(\Bs)\}_{i\in [n]}\in P_{\cF}$. 

Consider the SAPP mechanism with adjusted posted prices $\theta_i(\Bs)=F_i^{-1}(1-\frac{q_i(\Bs)}{3})$ for every $i$, $\Bs$ and sub-matroid $\cF'\subseteq \cF$ suggested by OCRS. 

\end{proof}
}

\begin{prevproof}{Theorem}{thm:reduction}

We construct the following allocation rule $x=\{x_i(\Bb,\Bs)\}_{i\in [n]}$. For every $i$ and $\Bb,\Bs$, let $$x_i(\Bb,\Bs)=\ind\left[i=\argmax_k (\tilde{\varphi}_k(b_k)-s_k)\wedge \tilde{\varphi}_i(b_i)\geq s_i\right].$$

Then $x$ satisfies both properties in the statement of Lemma~\ref{sapp:main}. Thus by Lemma~\ref{sapp:main},
$$\sapp\geq \E_{\Bb,\Bs}\left[\sum_i (\tilde{\varphi}_i(b_i)-s_i)\cdot x_i(\Bb,\Bs)\right]=\E_{\Bb,\Bs}\left[\max_i(\tilde{\varphi}_i(b_i)-s_i)^+\right].$$

Moreover, we use the upper-bound on $\optsd$ given by~Brustle et al.~\cite{BrustleCWZ17},
$$\optsd\leq \E_{\Bb,\Bs}\left[\max_i(\tilde{\varphi}_i(b_i)-s_i)^+\right]+\E_{\Bb,\Bs}\left[\max_i(b_i-\tilde{\tau}_i(s_i))^+\right].$$

Thus by Lemma~\ref{lem:mechanism for OPT-B}, we have
\begin{align*}
\max\{\optb,\sapp\}&\geq \frac{1}{2}\cdot \left(\E_{\Bb,\Bs}\left[\max_i(\tilde{\varphi}_i(b_i)-s_i)^+\right]+\E_{\Bb,\Bs}\left[\max_i(b_i-\tilde{\tau}_i(s_i))^+\right]\right)\\
&\geq \frac{1}{2}\cdot \optsd\geq \frac{1}{2c}\cdot \fbsd.    
\end{align*}
\end{prevproof}

\end{document}